\documentclass[12pt,a4paper]{amsart}
\usepackage[margin=2.5cm]{geometry}

\usepackage[utf8]{inputenc}
\usepackage[T1]{fontenc}

\usepackage{amsmath, amssymb, amsfonts,amsthm,amsopn,amscd,mathrsfs}
\allowdisplaybreaks

\usepackage{dsfont}
\usepackage{graphicx}
\usepackage{booktabs}\renewcommand{\arraystretch}{1.2}
\usepackage{titletoc}
\usepackage[section]{placeins}
\usepackage[shortlabels]{enumitem}
\usepackage{kbordermatrix}
\usepackage{mathtools}

\usepackage{upgreek}
\usepackage{xcolor}

\usepackage[breaklinks]{hyperref}
\hypersetup{
	linkcolor={red!40!black},
	citecolor={blue!50!black},
	urlcolor={blue!20!black}
}

\newcommand\overbar[1]{\accentset{\rule{.4em}{.8pt}}{#1}}

\DeclareMathOperator{\tr}{tr}

\DeclareMathOperator{\rank}{rank}

\DeclareMathOperator{\Aut}{Aut}

\DeclareMathOperator{\wt}{wt}
\DeclareMathOperator{\GF}{\mathrm{GF}}

\newcommand{\av}[1]{{\langle#1\rangle}}

\newcommand{\ket}[1]{\mathinner{|#1\rangle}}

\newcommand{\dyad}[1]{| #1\rangle \langle #1|}
\newcommand{\ot}[0]{\otimes}
\newcommand{\one}[0]{\mathds{1}}
\renewcommand{\a}{\alpha}

\newcommand{\II}{\mathcal{I}}

\newcommand{\Al}{A^{(\ell)}}
\newcommand{\Bl}{B^{(\ell)}}

\newcommand{\N}{\mathds{N}}
\newcommand{\R}{\mathds{R}}
\newcommand{\C}{\mathds{C}}

\newcommand{\by}[1]{{\boldsymbol{#1}}}

\newcommand{\qn}{\by{q}^n}

\newcommand{\EE}{\mathcal{E}}
\newcommand{\HH}{\mathcal{H}}

\newcommand{\aaa}{a}
\newcommand{\bbb}{b}
\newcommand{\ccc}{c}
\newcommand{\ddd}{d}
\newcommand{\x}{x}
\newcommand{\y}{y}

\usepackage{thmtools, thm-restate}
\newtheorem{theorem}{Theorem}
\newtheorem*{theorem*}{Theorem}
\newtheorem{proposition}[theorem]{Proposition}

\newtheorem{definition}[theorem]{Definition}
\newtheorem{corollary}[theorem]{Corollary}
\newtheorem{observation}[theorem]{Observation}
\newtheorem{example}[theorem]{Example}
\newtheorem{remark}[theorem]{Remark}

\newtheorem*{problem*}{Problem}
\newtheorem*{question*}{Question}

\newtheorem*{result*}{Result}

\newtheorem{thmA}{Result}

\newcommand{\nn}{\nonumber}

\renewcommand{\overbar}[1]{\mkern 1.5mu\overline{\mkern-1.5mu#1\mkern-1.5mu}\mkern 1.5mu}

\makeatletter
\renewcommand*\env@matrix[1][\arraystretch]{%
	\edef\arraystretch{#1}%
	\hskip -\arraycolsep
	\let\@ifnextchar\new@ifnextchar
	\array{*\c@MaxMatrixCols c}}
\makeatother

\begin{document}

	\title[
	SDP bounds on quantum codes
	]{SDP bounds on quantum codes}

	\date{\today}

	\author{Gerard Anglès Munné\textsuperscript{*}} %

	\address{
	Institute of Theoretical Physics and Astrophysics, University of Gdansk, ul. Wita Stwosza 57, 80-308 Gdansk, Poland\textsuperscript{*}}

	\author{Andrew Nemec\textsuperscript{†}} %
    	\address{
    	Department of Computer Science,
    	University of Texas at Dallas,
	800 Franklyn Jenifer Dr.,
    	Richardson, TX, 75080, USA\textsuperscript{†}
    	}

	\author{Felix Huber\textsuperscript{§}}
	\address{Division of Quantum Computing,
	Faculty of Mathematics, Physics and Informatics,
	University of Gdańsk,
	Wita Stwosza 57,
	80-308 Gdańsk, Poland\textsuperscript{§}}
	\email{felix.huber@ug.edu.pl}

	\thanks{GM and FH were supported by the Foundation for Polish Science through TEAM-NET (POIR.04.04.00-00-17C1/18-00).
	GM was also supported by NCN grant no. 2024/53/B/ST2/04103.
	FH's research was funded in whole or in part by
	the National Science Centre, Poland 2024/54/E/ST2/00451,
	the Polish National Agency for Academic Exchange
	under the Strategic Partnership Programme grant BNI/PST/2023/1/00013/U/00001,
	the Agence National de Recherche grant ANR-23-CPJ1-0012-01,
	and the Region Nouvelle-Aquitaine grant 34982420.
	AN was supported by QLCI.
	Much of the work was written while FH and GM were affiliated with the
	Jagiellonian University Kraków, Poland;
	while FH was affiliated with the
	Bordeaux Computer Science Laboratory (LaBRI), University of Bordeaux, France;
	and while AN was affiliated with the Duke Quantum Center, Duke University, USA.
	For the purpose of Open Access, the author has applied a CC-BY public copyright licence to any Author Accepted Manuscript (AAM) version arising from this submission.
	We thank
	Christine Bachoc,
	Simeon Ball,
	Bram Bekker,
	Sébastien Designolle,
	Dion Gijswijt,
	Marcus Grassl,
	Carlos de Gois,
	Otfried Gühne,
	Kiara Hansenne,
	Igor Klep,
	Victor Magron,
	Moisés Bermejo Morán,
	Claudio Procesi,
	Eric Rains,
	René Schwonnek,
	Jurij Volčič,
	Andreas Winter,
	Zhen-Peng Xu,
	and
	Karol Życzkowski
	for fruitful discussions and helpful comments.
	We especially want to thank the anonymous Referee ``2'' and Dion Gijswijt for extensive comments,
	which really helped to improve this manuscript.
}

\begin{abstract}
This paper provides a semidefinite programming hierarchy based on state polynomial optimization
to determine the existence of quantum codes with given parameters.
The hierarchy is complete, in the sense that a $(\!(n,K,\delta)\!)_2$ code exists
if and only if every level of the hierarchy is feasible.
It is not limited to stabilizer codes and thus is applicable generally.
While the machinery is formally dimension-free,
we restrict it to qubit codes through quasi-Clifford algebras.
We derive the quantum analog of a range of classical results:
first, from an intermediate level a Lov\'asz bound for self-dual quantum codes is recovered.
Second, a symmetrization of a minor variation of this Lov\'asz bound recovers the quantum Delsarte bound.
Third, a symmetry reduction using the Terwilliger algebra leads
to semidefinite programming bounds of size $O(n^4)$.
With this we give an alternative proof that there is no $(\!(7,1,4)\!)_2$ quantum code,
and show that $(\!(8,9,3)\!)_2$ and $(\!(10,5,4)\!)_2$ codes do not exist.
\end{abstract}

\maketitle

\section{Introduction}

Coding theory studies how to store or send messages in the presence of noise.
Typically, a message consists of a bit string of which a noisy channel may flip some bits.
A classical error-correcting code allows to encode messages in a redundant manner,
so that it can be recovered after being affected by noise.
In contrast, a quantum message is represented by a positive semidefinite matrix of trace one
on a $K$-dimensional Hilbert space
$\C^K$,
that is,
$\varrho \in L(\C^{K})$ with $\varrho \succeq 0$ and $\tr(\varrho)=1$.
A noisy quantum channel can be written in a Kraus decomposition as $\mathcal{N}(\varrho) = \sum_{E \in \mathcal{N}} E \varrho E^\dag$ with $\sum_{E \in \mathcal{N}} E^\dag E \preceq \one$.
A quantum error-correcting code then encodes a $K$-dimensional Hilbert space to a subspace
of a $n$-qubit Hilbert space $(\C^2)^{\ot n}$,
such that the action of a noisy quantum channel can be corrected.
The Knill-Laflamme conditions
provide necessary and sufficient conditions for this to be possible~\cite{PhysRevLett.84.2525}.
Formulated in terms of the projector $\Pi$ onto the code subspace they read:
if $\Pi E^\dag F \Pi = c_{EF} \Pi$
for all errors $E,F$ appearing in the channel $\mathcal{N}$, then there exists a recovery map $\mathcal{R}$ such that
$\mathcal{R}(\mathcal{N} (\varrho))\propto \varrho$.
The size $K$ of the quantum code equals the rank of $\Pi$.

A fundamental problem in both classical and quantum coding theory is to determine
the maximum code size $K$
for a given block-length~$n$ and minimum distance~$\delta$.
The maximum size can be bounded with a linear program using weight enumerators~\cite{10.1007/978-94-010-1826-5_7,PhysRevLett.78.1600}:
for a classical linear code ~$C$,
the weight enumerator counts the number of codewords of Hamming weight~$j$, $A_j(C) = |\{c \in C\, :\, \wt(c)=j\}|$.
Here $\wt(c)$ is the number of non-zero positions in the bit string $c$.
In contrast, the quantum weight enumerators $A_j(\Pi)$ decompose the
Hilbert-Schmidt norm
$\tr(\Pi^\dag \Pi)$
according to the weights
of quantum correlations~\cite{1751-8121-51-17-175301}:
one has $A_j = \sum_{E \in \EE_n,\, \wt(E)= j} \tr(\Pi E^\dag) \tr(\Pi E)$,
where $\EE_n$ is the $n$-qubit Pauli basis
and $\wt(E)$ the number of coordinates in which $E$ acts nontrivially.
The dual enumerator is given by
$B_j = \sum_{E \in \EE_n,\, \wt(E)= j} \tr(\Pi E^\dag \Pi E)$~\footnote{
For stabilizer and codeword-stabilized quantum codes,
the weight enumerators allow for combinatorial interpretations~\cite{9738648, lai2023semidefinite},
but for general quantum codes no such interpretation is known.}.
The introduction of weight enumerators allowed for efficiently computable linear programming bounds on the parameters $n$, $K$, and $\delta$ that a code may take.
In the classical case, these are strengthened by a complete hierarchy of semidefinite programs~\cite{10.1137/0801013,lasserre2001explicit}.
However, no such hierarchy for the quantum case exists in the literature. The aim of this paper is to fill this gap
and, similar to the classical case, link it to the Lov\'asz theta number of a graph and a Delsarte bound.

Our main result is a complete semidefinite programming (SDP)
hierarchy for the existence of quantum codes [Result~\ref{thm:A}, Section~\ref{sec:stcode}]:
a code with given parameters $(\!(n,K,\delta)\!)_2$ exist
if and only if every level of the hierarchy is feasible.
Inspired by~\cite{PhysRevLett.132.200202},
our hierarchy is based on the state polynomial optimization framework,
which provides a sequence of outer approximations
to the set of non-linear quantum correlations.
This allows us to approximate the quantum weight enumerators $A_j$ and $B_j$,
onto which the Knill-Laflamme and other conditions can be imposed.
While this machinery is formally dimension-free,
we pull this result back to $n$-qubit systems
through the characterization of quasi-Clifford algebras by Gastineau-Hills~\cite{gastineau-hills_1982}.

We then derive quantum Lov\'asz and Delsarte bounds.
Given a graph, a subset of vertices is termed independent if no two of its vertices are adjacent.
The Lov\'asz theta number $\vartheta$ then bounds the independence number $\alpha$ of a graph, that is the largest cardinality of an independent set
as $\alpha(G) \leq \vartheta(G)$.
The computation of $\alpha$ is NP-hard~\cite{Karp1972},
but $\vartheta$ itself can be expressed as a semidefinite program~\cite{1055985}.
This provides an efficiently computable upper bound on $\alpha$.
A classical code corresponds to a subset of the graph of all binary words of length $n$
for which words of Hamming distance less than $\delta$ are adjacent.
As a consequence, the Lov\'asz theta number also provides an upper bound
on the maximum code size.
A suitable symmetrization of the SDP for the Lov\'asz theta number
yields the classical Delsarte bound, which encodes the condition that the primal and dual weight enumerators are non-negative.
The derivation relies on the fact that, given the indexing vector of a code $\chi$,
the matrix $\chi \chi^T$ is positive semidefinite
and has zero entries when the row and column strings
have a distance in $1,\dots, \delta-1$~\cite{PhDGijswijt}.

We show that for self-dual quantum codes~\footnote{
Self-dual codes,
that is, quantum codes with $K=1$,
are always assumed to be pure,
meaning they must satisfy $\tr(E \Pi) = 0$ for all $0 < \wt(E)<\delta$.
Otherwise the Knill-Laflamme conditions would be satisfied trivially.},
quantum Lov\'asz and Delsarte bounds arise in a similar fashion [Result~\ref{thm:B}, Section~\ref{sec:Delsarte}],
with the added constraint that the matrix also has zeros
in rows and columns whose Pauli strings anticommute.

Finally, following the strategy of Ref.~\cite{GIJSWIJT20061719}
we give efficiently computable semidefinite programming bounds [Result~\ref{thm:C}, Section~\ref{sec:sym_red_Terwilliger}]:
similar to the classical case, these scale as $O(n^4)$, where~$n$ is the block length,
and build on the block-diagonalization of the Terwilliger algebra
of the non-binary Hamming scheme.
With this, we give an an alternative proof that there is no
$(\!(7,1,4)\!)_2$ quantum code~\footnote{
The original proof in Ref.~\cite{PhysRevLett.118.200502}
does not make use of mathematical optimization.}, as well as prove the non-existence of quantum codes with parameters $(\!(8,9,3)\!)_2$ and $(\!(10,5,4)\!)_2$ [Section~\ref{sec:appli_codes}].
This shows that the method can lead to tighter bounds than the
``enhanced'' quantum linear programming bounds
that include the shadow inequalities by Rains~\cite{817508, 796376},
which do not recover our results.

Overall, our approach is applicable to both stabilizer and non-stabilizer codes.
It is based on moment matrices and thus offers high flexibility on the constraints
to be included,
so that small relaxations can be constructed.
Lastly, because our formalism is dimension-free,
it can in principle also target scenarios where the Hilbert space is unknown.

\section{Contributions}
Our aim is to determine the existence of quantum codes with block-length $n$,
code distance~$\delta$,
and the dimension of the code space~$K$.
A quantum code with parameters $(\!(n,K,\delta)\!)_2$ is then able to encode a $K$-dimensional system into $n$ qubits,
such that errors acting on at most $\lfloor \tfrac{\delta-1}{2} \rfloor$
physical qubits can be corrected.
A fundamental problem in quantum coding theory is~\cite{1986Hill, dougherty2015open}:

\begin{question*}
Does there exist a quantum code with parameters $(\!(n,K,\delta)\!)_2$?
\end{question*}
Partial answers to this questions are known in form of the quantum linear programming~\cite{PhysRevLett.78.1600,
681316}
and analytical upper bounds~\cite{782103, NebeRainsSloane2001,
nemec2023hamminglikebounddegeneratestabilizer,
Huber_2020,
Grassl_2022}.
Our main results are the following:
\begin{thmA}\label{thm:A}
	A quantum code with parameters $(\!(n,K,\delta)\!)_2$ exists
	if and only if
	the semidefinite programming hierarchy of Eq.~\eqref{eq:code_SDP} is feasible
	at every level $\ell \in \N$.
\end{thmA}

Result~\ref{thm:A} is shown in Section~\ref{sec:stcode}, Theorem~\ref{thm:code}.
The proof is based on the characterization of quantum error correction through the Knill-Laflamme conditions,
the quantum MacWilliams identity, and the state polynomial optimization framework
to optimize over valid quantum weight enumerators.
From an intermediate level of the hierarchy we recover the following result for self-dual quantum codes:
\begin{thmA}\label{thm:B}
The Lov\'asz  theta number $\vartheta$ bounds the existence of self-dual quantum codes.
The symmetrization over distance preserving automorphisms of a minor variation of the Lov\'asz bound
recovers the quantum Delsarte bound.
\end{thmA}
Result~\ref{thm:B} is shown in
Section~\ref{sec:Lovasz}, Corollary~\ref{cor:lovasz}; and
Section~\ref{sec:Delsarte}, Theorem~\ref{thm:Lovasz2Delsarte}.
The proof is based on a confusability graph constructed from the Pauli basis $\EE_n$
in Definition~\ref{def:pure_conf_graph},
with additional edges arising from anti-commutation relations.
Finally, we provide a symmetry-reduced semidefinite program for general quantum codes with $K\geq1$:
\begin{thmA}\label{thm:C}
There is a symmetry-reduction of an intermediate level of the semidefinite programming
hierarchy based on the Terwilliger algebra of size $O(n^4)$.
This recovers the non-existence of the
$(\!(7,1,4)\!)_2$ code, and proves the non-existence of $(\!(8,9,3)\!)_2$ and $(\!(10,5,4)\!)_2$ codes.
\end{thmA}
Result~\ref{thm:C} is shown in
Section~\ref{sec:sym_red_Terwilliger}, Theorem~\ref{prop:sym_red_sdp}; and Section~\ref{sec:appli_codes},
Observations~\ref{prop:infeas_cert} and~\ref{prop:infeas_cert2}.
The proof is based on symmetrizing an intermediate level of the hierarchy over distance preserving automorphisms of the Pauli basis $\EE_n$ that leave the identity invariant.
The non-existence proof for the $(\!(7,1,4)\!)_2$ code is provided as an infeasibility certificate in Appendix~\ref{app:inf}. The certificates for $(\!(8,9,3)\!)_2$ and $(\!(10,5,4)\!)_2$ codes
can be found online in \url{https://github.com/ganglesmunne/SDP_bounds_for_quantum_codes} [Ref. \cite{Munné2024github}].

\subsection{Proof sketch: SDP hierarchy}
	We now give a sketch of the proof of Result~\ref{thm:A}.
	Consider a quantum channel
	$\mathcal{N}(\varrho) = \sum_{E \in \mathcal{N}} E \varrho E^\dag$.
	A quantum code requires that the Knill-Laflamme conditions~\cite{PhysRevA.55.900}
	are met:
	A projector~$\Pi$ corresponds to a quantum code able to correct errors induced by~$\mathcal{N}$,
	if and only if
	\begin{align}\label{eq:KnillLafl}
		\Pi E^\dag F \Pi = c_{EF} \Pi \quad \text{for all } \quad E, F \in \mathcal{N}\,, \quad c_{EF} \in \C\,.
	\end{align}

	\bigskip

	In this paper we consider errors that are Pauli strings,
	that is tensor products of Pauli matrices and the identity.
	The set of Pauli strings
	$\EE_n = \{I,X,Y,Z\}^{\ot n}$
	forms a complete orthonormal basis
	of operators acting on $(\C^2)^{\ot n}$.
	The weight~$\wt(E_a)$ of an operator $E_a \in \EE_n$
	is the number of positions it acts non-trivially on.
	If a code has distance~$\delta$,
	then all errors of weight $\lfloor \tfrac{\delta-1}{2}\rfloor$
	can be corrected.
	The size of the code space is $K = \tr(\Pi)$ with $\Pi$ being a projector acting on a $n$ qubit Hilbert space.

	By Eq.~\eqref{eq:KnillLafl}, an $(\!(n,K,\delta)\!)_2$  code exists
	if and only if the following program is feasible:
	\begin{align}\label{eq:find}
		\operatorname{find}\quad & \Pi\nn\\
		\text{subject to} 	\quad & \Pi^2 = \Pi \,, \quad
		\tr(\Pi) = K \,,\nn\\
		& \Pi E_a^\dag E_b \Pi = c_{ab} \Pi \quad\quad \text{for all}
		\quad E_a, E_b \in \EE_n\,:\, \wt(E_a^\dag E_b)<\delta \,.
	\end{align}
	Here $c_{ab}\in \C$ and $\Pi$ acts on $(\C^2)^{\ot n}$.

	We show how the feasibility of the program~\eqref{eq:find}
	can be formulated as an SDP hierarchy.
	The hierarchy is complete, in the sense
	that an $(\!(n, K, \delta)\!)_2$ code exist if and only if every level of the hierarchy is feasible.
	It rests on the following observations:
	\begin{itemize}

		\item[A.] The Knill-Laflamme conditions for a projector $\Pi$ of rank $K$ to correspond to a quantum code
		are nonlinear in $\Pi$.
		That is, they can be formulated as the requirement that $K B_j(\Pi) = A_j(\Pi)$ for all $j = 0,\dots,\delta-1$ \cite{681316}, where
		\begin{equation}
	 A_j(\Pi)= \sum_{\substack{E_a \in \EE_n \\ \wt(E_a)=j}} \tr(E_a\Pi)\tr(E_a^\dag\Pi)\,, \quad 		 B_j(\Pi)= \sum_{\substack{E_a \in \EE_n \\ \wt(E_a)=j}} \tr(E_a\Pi E_a^\dag\Pi)\,.
		\end{equation}

		\item[B.]
		The quantum MacWilliams identity allows one to formulate this condition in terms of the $A_j(\Pi)$ alone~\cite{681316}.
		Define $A(x,y) = \sum_{j=0}^n A_j(\Pi) x^{n-j} y^j$ and likewise for $B(x,y)$.
		The identity states that the $B_j(\Pi)$ can be expressed as linear combinations of the~$A_j(\Pi)$,
			\begin{align}\label{eq:mac0}
			B(x,y) & = A\Big(\frac{x+3y}{2},\frac{x-y}{2}\Big)\,.
		\end{align}
	\end{itemize}

	We are now left with the task to determine the set of possible
	$A_j(\Pi)$ where $\Pi$ is a projector acting on $(\C^2)^{\otimes n}$.

	\begin{itemize}
	\item[C.] First note that the conditions for a $n$-qubit state $\varrho$
		to be proportional to a projector~$\Pi$ of rank~$K$
		are nonlinear in its expectation values.
		To see this, note that swap operator exchanging the two tensor factors of
		$(\C^d)^{\ot n} \ot (\C^d)^{\ot n}$ can be written as
		$(1,2)=2^{-n}\sum_{E_a\in \EE_n} E_a^\dag \ot E_a$ where the sum is over the $n$-qubit Pauli basis $\EE_n$~ \cite{Eltschka2018distributionof}.
		Decomposing cycles into elementary transpositions,
		a generalization of the swap trick then allows one to evaluate traced powers of a state~\cite{huber2020positive},
		\begin{align}
			\tr(\varrho^{m})= \tr\big((1,2,\dots, m)\varrho^{\ot m}\big) \,.
		\end{align}
		The fact that $\varrho=\Pi/K$
		if and only if $\tr(\varrho^{m})=1/K^{m-1}$ for all $m\in \N$ recovers the conditions for $\Pi$ to be a projector.

		\item [D.]
		Second, note that the set of $n$-qubit Pauli operators can be defined algebraically:
		The characterization of algebras $\mathscr{C}$ over a complex field $\C$ whose elements satisfy $\alpha_i^2 = 1 $ and $\alpha_j \alpha_i = (-1)^{\chi_{ij}} \alpha_i \alpha_j$ with $\chi_{ij}  \in \{0,1\}$, $1\leq i<j\leq m$,
		was given by Gastineau-Hills~\cite{gastineau-hills_1982}.
		These quasi-Clifford algebras $\mathscr{C}$ are isomorphic to
		a direct sum of Pauli operators acting on $s$-qubit systems,
		\begin{equation}
	 \mathscr{C} \simeq \bigoplus^{2^r}_{i=1} (\C_2)^{\ot s} \,,
		\end{equation}
		where $\C_2$ is the space of complex $2\times 2$  matrices and $r+2s=m$.
		In combination with the constraints of the hierarchy,
		this restricts the Gelfand-Naimark-Segal construction of the state polynomial optimization framework
		to states and operators acting on the Hilbert space of $n$~qubits.

		\item [E.] Third, the framework of state polynomial optimization~\cite{klep2023state} (also known as scalar extension~\cite{pozas2019quantum,PhysRevLett.123.140503})
		allows us to optimize over expressions that are nonlinear in expectation values of a state, such as
		$A_j(\varrho)= \sum_{E_a \in \EE_n ,\, \wt(E_a)=j} \tr(E_a\varrho)\tr(E_a^\dag\varrho)$.
		This is done by a variant of the Navascués-Pironio-Acín hierarchy~\cite{doi:10.1137/090760155} for non-commutative optimization.
		It consists of a hierarchy of semidefinite programs
		whose objective value converges to that of a non-commutative state polynomial optimization problem.

	\end{itemize}
	This way, the state polynomial optimization hierarchy
	allows us to optimize over the set of possible weight enumerators arising from $n$-qubit quantum codes.
	This gives a complete semidefinite programming hierarchy
	for the existence of qubit quantum codes with parameters $(\!(n,K,\delta)\!)_2$.
	That is, a code with given parameters exist if and only if every level of the hierarchy is feasible.
	Semidefinite programming duality can then be used to find infeasibility certificates to prove that no code with the given parameters exists.

\subsection{Proof sketch: Lov\'asz, Delsarte, and symmetry-reduced SDP bounds}

	We sketch the proofs of Results~\ref{thm:B} and ~\ref{thm:C}.
	\begin{itemize}

	\item [A.]
	We consider a relaxation of the complete hierarchy of Result~\ref{thm:A} that is indexed by $\langle E_a \rangle E_a$ with
	$E_a \in \EE_n$.
	The key object is a matrix with entries
	$	\Gamma_{\aaa\bbb}=\text{Re}(\av{E_\aaa^\dagger} \av{E_\bbb} \av{E_\aaa E_\bbb^\dagger})$.
	Such matrix is positive semidefinite by construction, and satisfies $\Gamma_{00} = 1$,  $\Gamma_{a0} = \Gamma_{aa}$,
	and $\Gamma_{ab} = 0$ if $E_a E_b + E_b E_a = 0$.
	For pure codes we additionally have
	$\Gamma_{ab} = 0$ if $0<\wt(E_a)<\delta$,
	$0<\wt(E_b)<\delta$,
	or  $0<\wt(E_aE^\dagger_b)<\delta$.

	\item [B.]
	Consider the case of self-dual codes, i.e., pure codes with $K=1$.
	Then the set of matrices $\Gamma$ with this structure can be identified with the Lov\'asz theta body for a graph with loops,
	that is the feasible region of Eq.~\eqref{eq:lovasz_SDP2},
	\begin{equation}
	 \mathrm{TH}(G) = \Big\{ \operatorname{diag}(M)
	 \,\big|\, \begin{pmatrix}
	            1 &x^T \\
				x &M
	           \end{pmatrix}  \succeq 0\,,
	           x_a = M_{aa} \,\forall a\,, M_{ab} = 0 \text{ if } a \sim b \Big\} \nn \,.
	\end{equation}
	A projector onto a self-dual code  additionally satisfies $\sum_{E_a \in \EE_n} \Gamma_{aa} = 2^n$.
	As a consequence, the Lov\'asz theta number gives a bound on the existence of self-dual quantum codes.

	\item [C.] 	Consider the the anti-commutativity graph $G$ of the $n$-qubit Pauli group.
				Averaging $\mathrm{TH}(G)$ of under distance preserving automorphisms
				$\Aut(4,n)$ yields, along with additional equality constraints from projector conditions,
				the quantum Delsarte bound~[Eq.~\eqref{eq:lp_Delsarte}].

	\item [D.]  For arbitrary quantum codes,
				the matrix $\Gamma$ satisfies a range of additional constraints.
				These arise
				from the (projective) group structure of the indexing set
				and from the fact that the code subspace corresponds to a projector $\Pi^2 = \Pi$.
				In combination, this yields the main semidefinite programming relaxation in Eq.~\eqref{eq:sdpgamma}.

	\item [E.]  Using the Terwilliger algebra for $4$-ary codes, we average this main relaxation over distance-preserving automorphisms
				that keep the zero codeword invariant $\Aut_0(4,n)$. A block-diagonalization using the Terwilliger algebra~\cite{GIJSWIJT20061719} yields to a symmetry-reduced semidefinite program of size $O(n^4)$.
	\end{itemize}

	A variation of the Lov\'asz bound recovers the non-existence of the $(\!(7,1,4)\!)_2$ code,
	while a variation of the symmetry-reduced semidefinite programm shows the non-existence of
	$(\!(8,9,3)\!)_2$ and $(\!(10,5,4)\!)_2$ codes.

\section{Key concepts}
\subsection{Quantum error-correcting codes.}\label{subsec:qec}

	A quantum error-correcting code with parameters $(\!(n,K,\delta)\!)_2$ encodes a Hilbert space of dimension $K$ in $n$ qubits.
	A code of distance $\delta$ can correct all errors acting on at most
	$\lfloor \tfrac{\delta-1}{2}\rfloor$ qubits.

	Define a $n$-qubit error basis (or Pauli basis) ${\EE_n}$
	by considering all $n$-fold tensor products of
	the identity and the three Pauli matrices,
	\begin{equation}\label{eq:pauli}
	I =
	\begin{pmatrix}
		1 & 0 \\
		0 & 1
	\end{pmatrix}
	\,,\quad
	X =
	\begin{pmatrix}
		0 & 1 \\
		1 & 0
	\end{pmatrix}
	\,,\quad
	Y =
	\begin{pmatrix}
		0 & -i \\
		i & 0
	\end{pmatrix}
	\,,\quad
	Z =
	\begin{pmatrix}
		1 & 0 \\
		0 & -1
	\end{pmatrix}\,.
	\end{equation}
	The elements of $\EE_n$ then generate the
	$n$-qubit Pauli group~$P_n = \{i^b E_a \,|\, E_a \in \EE_n\,, b = 0,1,2,3\}$.
	Denote by $\wt(E_a)$ the Hamming weight of an error $E_a$, which is the number of positions an error $E_a\in {\EE_n}$ acts on non-trivially.

	Given a projector $\Pi$ onto a code subspace,
	the Knill-Laflamme conditions~\cite{PhysRevA.55.900} state that the code
	has distance $\delta$, if and only if
	\begin{align}\label{eq:KnillLafl2}
		\Pi E_a^\dag E_b \Pi = c_{ab} \Pi \quad  \text{for all} \quad\wt(E_a^\dag E_b)<\delta \,,
	\end{align}
	where  $c_{ab}\in \C$.
	A code is termed pure if $c_{ab} =\tr\!\left(E_a^\dagger E_b\Pi\right)/K= 0$ for all $E_a^\dag E_b\neq \one$
	with $0 < \wt(E_a^\dag E_b) < \delta$~\cite{681315, PhysRevA.69.052330}. Note that all codes with $K=1$ are assumed to be pure, as the Knill-Laflamme conditions would be trivial otherwise.
	Such codes are termed self-dual~\footnote{This is an extension of the nomenclature for stabilizer codes, where $k=0$ ($K=1$) implies the stabilizer group is isomorphic to a self-dual additive classical code.}.
	A motivation to study self-dual codes comes from the propagation rules
	for pure codes:
	any pure $(\!(n,K,\delta)\!)_2$ gives rise to a family of
	$(\!(n-s, 2^s K, \delta-s)\!)_2$ codes
	for $s=1, \dots, \delta-1$~\cite[Theorem 19]{681316}.
	An additional interest lies in the fact that self-dual codes correspond to absolutely maximally entangled and $m$-uniform states~\cite{Rather_2022, 1751-8121-51-17-175301}.	

	\subsection{Stabilizer codes}\label{subsec:stabcodes}

	Given $n-k$ commuting independent generators
	$g_i \in P_n$, let $S = \langle g_1, \dots, g_{n-k}\rangle$
	be a subgroup of the $n$-qubit Pauli group $P_n$ not containing $-\one$.
	The stabilizer group $S$ defines a code subspace,
	whose corresponding projector is,
	\begin{align}\label{eq:stab_code_proj}
		\Pi= \frac{1}{2^{n-k}}\prod^{n-k}_{i=1} (\one+g_i) =\frac{1}{2^{n-k}} \sum_{s \in S} s\,.
	\end{align}
	In particular, the $s_i \in S$ stabilize $\Pi$, so that
	$s\Pi =\Pi s=\Pi$ for all $s \in S$.
	A stabilizer code encodes $k$ logical qubits into $n$ physical qubits with $k\leq n$ and $K=2^k$.
	The size of the stabilizer group is $|S|=\frac{2^{n}}{K}=2^{n-k}$ and when $k=0$ ($K=1$) one has a stabilizer state.

	\subsection{Quantum weight enumerators}
	Given a projector $\Pi$ onto a code space,
	we define two quantum weight enumerators\footnote{We follow Rains \cite{681316} and omit the normalization factor from the original
	formulation by Shor and Laflamme~\cite{PhysRevLett.78.1600}.} as
	\begin{align}\label{eq:enum}
		A_j(\Pi)=& \sum_{\substack{E_a\in {\EE_n} \\ \wt(E_a)=j}} \tr(E_a^\dag \Pi) \tr(E_a\Pi) \,,\\
		B_j(\Pi)=& \sum_{\substack{E_a\in{\EE_n} \\ \wt(E_a)=j}} \tr(E_a\Pi E_a^\dagger\Pi)\,,
	\end{align}
	where the sum runs over all errors $E_a$ of weight $j$ in the $n$-qubit Pauli basis ${\EE_n}$.
	Let $K=\rank(\Pi)$, then we have
	\begin{align}\label{eq:enumcond}
	 A_0(\Pi) &= K^2\,,\nn\\
	 A_j(\Pi) &\geq 0 \,, \nn\\
	 B_j(\Pi) &\geq 0 \,, \nn\\
	 KB_j(\Pi) 			&\geq A_j(\Pi) \,,
	\end{align}
	for all $j=0,\dots, n$.
	Self-dual codes satisfy $A_j(\Pi) = B_j(\Pi)$ for all $j$.
	This is seen by expanding $A_j$ and $B_j$
	in terms of a projector onto a pure state $\Pi = \dyad{\psi}$.

	Interestingly, the Knill-Laflamme conditions~[see Eq.~\eqref{eq:KnillLafl2}] can be formulated in terms of the quantum weight enumerators.
	\begin{theorem}{\cite[Theorem 17]{681316}}\label{thm:Rains}
		A projector $\Pi$ of rank $K$ on $n$ qubits corresponds to a quantum code of distance $\delta$,
		if and only if
		\begin{align}\label{eq:klenum}
			KB_j(\Pi) &= A_j(\Pi) \quad \text{for}	\quad  0\leq j < \delta\,.
		\end{align}
	\end{theorem}
	For pure codes, the enumerators additionally fulfill $A_j(\Pi)=0$ for $0 < j < \delta$.
	The quantum MacWilliams identity relates the two enumerators
	through a polynomial transform~\cite{681316},
	\begin{align}\label{eq:mac}
			B(x,y) & = A\Big(\frac{x+3y}{2},\frac{x-y}{2}\Big)\,,
	\end{align}
	where $A(x,y)$ is a polynomial defined as $A(x,y) = \sum^n_{j=0} A_j \, x^{n-j} y^{j}$
	and likewise for $B(x,y)$~\footnote{
	Technically speaking, the
	$A(x,y)$ is the weight enumerator, whose coefficients $A_j$ form the weight distribution. We use the notions interchangeably.}.
	The coefficients of this polynomial transformation can be expanded as
	\begin{align}\label{eq:macexp}
	B_j(\Pi) & = 2^{-n}\sum^n_{i=0} K_j(i;n)A_i(\Pi)\,,
	\end{align}
	where $K_j(i;n):=K_j(i;n,4)$ are the quaternary Krawtchouk polynomials~\cite{macwilliams1977theory}, defined as
	\begin{align}\label{eq:krawpoly}
		K_j(i;n) = \sum^n_{\alpha=0} (-1)^\a 3^{j-\a}
		\binom{i}{\alpha}  \binom{n-i}{j-\alpha}     \,.
	\end{align}
	Note that $KB_0(\Pi)=A_0(\Pi) = K^2$
	corresponds to $\sum^n_{i=0} A_i(\Pi)= 2^nK$.
	The weight enumerator of a self-dual code is invariant under the quantum MacWilliams transform.
	In the case of qubit stabilizer codes, one has the additional constraint that
 	$\sum_{j=0}^{\lfloor n /2 \rfloor} A_{2j}(\Pi) = 2^{n-\log_2(K)-1}$ or
 	$\sum_{j=0}^{\lfloor n /2 \rfloor} A_{2j}(\Pi) =2^{n-\log_2(K)}$
 	for codes of type I and II respectively~\cite{681315,PhysRevA.69.052330}, as well as the constraint that the weight enumerators must take integer values.

	The quantum shadow enumerator $S_j$ can be obtained through the polynomial transform~\cite{796376}
	\begin{align}\label{eq:shadow}
		S(x,y) & = A\Big(\frac{x+3y}{2},\frac{y-x}{2}\Big) \,,
	\end{align}
	and its coefficients expand as

	\begin{align}\label{eq:shadowexp}
			S_j(\Pi)& = 2^{-n}\sum^n_{i=0} (-1)^i K_j(i;n)A_i(\Pi) \,.
	\end{align}
	It is known that $S_j(\Pi)\geq 0$ for all $j=0,\dots, n$,
	giving the so-called shadow bounds on quantum codes~\cite{651000,796376}.
 	If $K=1$ and $n-j$ is odd then $S_j(\Pi) = 0$~\cite[Theorem 12]{796376}.

	\subsection{Linear programming bounds}\label{subsect:LP}
	The enhanced linear programming bounds are given by the conditions from Eqs.~\eqref{eq:enumcond}--\eqref{eq:shadowexp}~\cite{PhysRevLett.78.1600,681316},
	\begin{align}\label{eq:lp}
		\text{find} 	\quad & A_j,\,B_j,\,S_j\geq 0 \nn\\
		\text{subject to}	\quad & A_0=K^2\,, \nn\\
		&  KB_j \geq A_j  \quad\quad\text{with equality for } 0 \leq j < \delta\,,
	\end{align}
	where
	$A_j$, $B_j$, $S_j$ are related by
	Eqs.~\eqref{eq:mac} and \eqref{eq:shadow}.
	In case of self-dual, pure, or stabilizer codes the above-mentioned
	extra constraints apply.
	For given parameters $(\!(n,K,\delta)\!)_2$, if the linear program in Eq.~\eqref{eq:lp} is infeasible, then a corresponding quantum codes does not exist.

	In analogy to the classical case
	(see Eq.~\eqref{eq:lp_Delsarte_class} and \cite{10.1007/978-94-010-1826-5_7,HuffmanKimSole2021}),
	for self-dual codes
	we will call the following the quantum Delsarte linear programming bound,
	\begin{align}\label{eq:lp_Delsarte}
	\eta = \text{max} 	\quad & \sum^n_{j=0} A_j\,, \nn\\
	\text{subject to} 	\quad & A_0 = 1\,, \nn\\
	& A_j \geq 0
	\quad\quad \text{with equality for } 0 < j < \delta\,,
% 	\text{with equality for } \,\, 0<j<\delta\,,
	\nn\\
	& \sum^n_{i=0} K_j(i;n)A_i \geq 0 \quad \text{for all } j=0,\dots, n\,.
	\end{align}
	From $\sum_{i=0}^n A_i(\Pi) = 2^n$,
	it follows that if $\eta<2^n$,
	then
	a code with parameters $(\!(n,1,\delta)\!)_2$
	does not exist.
	It is clear that Eq.~\eqref{eq:lp_Delsarte} relaxes the constraints in Eq.~\eqref{eq:lp}.
	As we show in Section~\ref{sec:Delsarte}, it can also be obtained by
	symmetrization of a Lov\'asz bound for self-dual quantum codes.
	This mirrors the derivation of the classical Delsarte bound [Eq.~\eqref{eq:lp_Delsarte_class}].

	\begin{remark}
	Note that for $K >1$ and without adding extra constraints such as the quantum MacWilliams identity or the Knill-Laflamme conditions
	(i.e.,
	$B_j = \sum_{i=0}^n K_j(i;n) A_i$ for $j=0,\dots,n$ and
	$KB_j = A_j$  for $j=0,\dots, \delta-1$)
	the quantum Delsarte bounds becomes weaker.
	\end{remark}

\subsection{Lov\'asz theta number of a graph}
\label{sec:lovasz_number_def}
	Let $G$ be a simple graph defined by a set of $N$ vertices~$V$ connected by a set of edges $E$.
	For the Lov\'asz theta number there exist a number of equivalent definitions.
	We consider the following two semidefinite programming formulations~\cite{GALLI2017159}.
	We call the following SDP1 for the Lov\'asz theta number:
	\begin{align}\label{eq:lovasz_SDP1}
		\vartheta(G)= \max_M \quad & \sum_{i,j = 1}^N M_{ij}\,,  \nn \\
		\text{subject to} \quad &  \tr(M) =1 \,,  \nn \\
		& M_{ij}=0  \quad \,\,\text{for all} \quad \{i,j\} \in E \,,  \nn \\
		& M \succeq 0 \,.
	\end{align}
	We call the following SDP2 for the Lov\'asz theta number:
	\begin{align}\label{eq:lovasz_SDP2}
		\vartheta(G)=\max_M \quad & \sum^N_{i=1} M_{ii} \,, \nn \\
		\text{subject to} \quad & M_{ii}=a_i \quad\, \text{for all}\quad i \in V \,,  \nn\\
		& M_{ij}=0  \quad \,\,\text{for all}\quad \{i,j\} \in E \,,\nn \\
		& \Delta = \begin{pmatrix} 1 & a^T\\ a & M \end{pmatrix} \succeq 0 \,.
	\end{align}
	The Lov\'asz theta number upper bounds the independence number of a graph,
	that is, the maximum cardinality of a subset of vertices that do not share any edge,
	by $\a(G) \leq \vartheta(G)$. In turn, this also bounds the Shannon capacity of a graph~\cite{1055985}.
	Other mathematical formulations of the Lov\'asz theta number exist~\cite{Knuth1994,porumbel2022demystifyingcharacterizationsdpmatrices, GALLI2017159}.

	A slight variation of SDP1 from Eq.~\eqref{eq:lovasz_SDP1} with the additional entry-wise non-negativity constraint $M_{ij} \geq 0$ for all $i,j\in V$, defines $\vartheta'(G)$.
	For binary codes,
	the computation of $\vartheta'(G)$
	can be shown to reduce to the classical Delsarte bound~\cite{1056072,HuffmanKimSole2021},
	\begin{align}\label{eq:lp_Delsarte_class}
	 \eta_C = \max \quad & 2^n \sum_{i=0}^n x_i \nn\\
				\text{subject to} \quad & x_0 = \frac{1}{2^n} \,,\nn \\
				& x_1 = \dots =x_{\delta-1} = 0 \,,\nn \\
				& x_\delta, \dots, x_n \geq 0 \,,\nn \\
				& \sum_{i=0}^n K_j(i;n,2)x_i  \geq 0 \quad \text{for all } j= 0, \dots, n\,.
	\end{align}
	where $K_j(i;n,2)$ are the binary Krawtchouk polynomials defined as~\cite{macwilliams1977theory},
	\begin{align}\label{eq:krawpolybi}
		K_j(i;n,2) = \sum^i_{\alpha=0} (-1)^\a
		\binom{i}{\alpha}  \binom{n-i}{j-\alpha}     \,.
	\end{align}
	In other words, $\vartheta'(G)=\eta_C$.

\subsection{Semidefinite programming}

	A semidefinite program optimizes a linear function over the self-dual cone of positive semidefinite matrices under a set of linear matrix constraints.
	Denote by $S^k$ the set of real symmetric matrices of size $k$.
	We now follow the exposition from Boyd and Vandenberghe~\cite[Section 4.6.2, Example 5.12]{boyd2004convex}
	for the cone of positive semidefinite matrices, with the Hilbert-Schmidt inner product $\langle A,B\rangle = \tr(A^\dag B)$.

	Consider a semidefinite
	program in standard form
	\begin{align}\label{eq:primal}
		\text{minimize} & \quad \tr(CX) \nn \\
		\quad \text{subject to} & \quad \tr(A_i X)=b_i \quad \text{for} \quad i=1,\dots, p \,,\nn \\
		& \quad X\succeq 0 \,.
		\end{align}
	where $A_1,\dots,A_p, C \in S^k$ and $b_i \in \R$.
	Eq.~\eqref{eq:primal} is called the primal program and $\tr(CX)$ is its objective function.
	To each primal program is associated a dual program which reads
	\begin{align}
		\text{maximize} \quad & \sum^p_{i=1} y_i b_i \nn \\
		\text{subject to} \quad &  C-\sum^p_{i=1} y_iA_i\succeq0 \quad \text{for} \quad  i=1,\dots,p \,.
	\end{align}

	A variable $X$ (or set of $y_i$) satisfying the constraints of the primal (or dual) problem is said to be primal (or dual) feasible.
	Any pair of feasible $X$ and $\{y_i\}_{i=1}^p$ satisfy \emph{weak duality}:
	\begin{align}\label{eq:weak}
		\tr(CX) -\sum^p_{i=1} y_ib_i \geq 0 \,.
	\end{align}
	An optimization problem is converted to a feasibility problem by setting $C=0$.
	Weak duality can then be used to prove primal infeasibility: if there exists a dual feasible solution $\{y_i\}_{i=1}^p$ with $\sum^p_{i=1} y_ib_i >  0$, then Eq.~\eqref{eq:weak} is violated and the primal problem is infeasible.

	It is straightforward to embed complex semidefinite programs in the above formalism, by
	encoding a complex Hermitian matrix $X$ as,
	\begin{equation}\tilde X=
	 \begin{pmatrix}
	  \operatorname{Re}(X) & - \operatorname{Im}(X) \\
	  \operatorname{Im}(X) & \operatorname{Re}(X)
	 \end{pmatrix},
	\end{equation}
	where $\operatorname{Re}(X)$ and the $\operatorname{Im}(X)$ are the real and the imaginary part of $X$, respectively.
	Then $X\succeq 0$ if and only if $\tilde X\succeq 0$.

\subsection{State polynomials}\label{subsec:stpoly}

	Given a state $\varrho$ and a set of observables $\{X_i\}^m_{i=1}$,
	denote by $\av{X_i}_\varrho = \tr(X_i \varrho)$ their expectation values~\footnote{
	Technically, $\av{X_i}_\varrho = \tr(X_i \varrho)$
	only holds for bounded linear trace-class operators $\varrho$ on a separable Hilbert space.
	In general, $\av{X_i}_\varrho = \varrho(X_i)\in \C$.}.
	Noncommutative polynomials over $\C$
	in the products of operators and their expectations are termed state polynomials.
	An example is
	$
	\av{X_1^2}_\varrho X_2 X_1 + \av{X_1}_\varrho \av{X_2 X_1 X_2}_\varrho$
	where the second term is understood as $\one \av{X_1}_\varrho \av{X_2 X_2 X_2}_\varrho$.
	Polynomials containing ``bare'' operators (e.g.,
	$\av{X_1^2}_\varrho X_2 X_1$) are non-commutative state polynomials,
	whereas those only containing expectation values
	(e.g., $\av{X_1}_\varrho \av{X_2 X_1 X_2}_\varrho$)
	are commutative or pure.
	How can one optimize state polynomials over the set of all states and operators satisfying
	a given set of constraints?

	Let $f$ and $\{g_j\}_{j=1}^r$ be state polynomials.
	Consider optimizing $f$ under positive semidefinite constraints
	$g_k \succeq 0$ over the set of quantum states $\varrho$ and operators
	$X_i$ acting on a separable Hilbert space $\HH$,
	\begin{equation}\label{eq:optprob}
		\begin{aligned}
			\zeta^*=\inf_{\mathcal{H}, \{X_i\} ,\varrho} \quad &   \big\langle \, f(X_1,\dots, X_m) \,\big\rangle_\varrho  \\
			\text{subject to} \quad & g_k(X_1,\dots,X_m)\succeq 0 \quad \text{for} \quad k=1, \dots,r \,.
		\end{aligned}
	\end{equation}
	
	Note that Eq.~\eqref{eq:optprob} requires an optimization over the Hilbert space $\mathcal{H}$ and all operators and quantum states with support on $\mathcal{H}$ satisfying the positive semidefinite constraints.
	The framework of state polynomial optimization~\cite{klep2023state}
	(also known as scalar extension~\cite{pozas2019quantum,PhysRevLett.123.140503})
	provides a hierarchy of semidefinite programming approximations to Eq.~\eqref{eq:optprob}.

	To see how this works,
	consider a set of (non-commutative) letters $\{x_i\}^m_{i=1}$
	and form words $w=x_{a_1}\dots x_{a_\xi}$.
	Denote the empty word as $1$, satisfying $1 w=w 1=w$ for all words $w$.
	The degree of a word
	is defined as the number of letters in it.
	The involution of a word is given by $(x_1 \cdots x_m)^*=(x^*_m \cdots x^*_1)$.
	We associate to each word $w$ a scalar with the symbol $\av{w}$.
	We refer to $\av{w}$ as the ``expectation value'' of a word.
	Expectations satisfy 
	\begin{align}
	 w \av{v} &=\av{v}w\,, \quad\quad\quad\quad
	 \av{w \av{v}} =\av{w}\av{v}\nn\\
	 \av{w}\av{v}&=\av{v}\av{w}\,, \quad\quad\quad
	 (\av{v} w )^*=\av{v^*} w^*\,,
	\end{align}
	for all words $v,w$~\footnote{Note the difference to trace polynomials,
	which additionally satisfy $\tr(vw) = \tr(wv)$.
	However, state polynomials generally do not satisfy $\av{vw} = \av{wv}$.}. In general, $x_i^*$ are considered separate variables from the $x_i$.
	They are related to each other by their real and imaginary parts,
	so that
	$\av{x_i^*}^{\text{Re}} =   \av{x_i}^{\text{Re}}$ and
	$\av{x_i^*}^{\text{Im}} = - \av{x_i}^{\text{Im}}$.
	This allows one to define $\av{v^*} = \overbar{\av{v}}$.
	In case of Hermitian variables one identifies $x_i=x_i^*$.

	Products of words and their expectations are known as (letter)
	state monomials, e.g., $w_{i_1} \av{w_{i_2}} \dots  \av{w_{i_m}}$.
	The degree of a state monomial is the sum over the degree of all words that compose it, including the words within expectations.
	State polynomials are formed by a linear combination of monomials over~$\C$,
	where the expectation $\langle \,\cdot\, \rangle$ is extended linearly.
	The degree of a state polynomial is the largest degree of the monomials that it is composed of.
	Let $\C \langle x \rangle$ be the set of non-commutative polynomials over $\C$,
 	and $\C \langle x \rangle_k$ its subset of degree at most $k$.
	Let $\bf S$ be
	the set of non-commutative state polynomials over $\C$,
	and $S = \{\langle p \rangle \,|\, p \in \bf{S}\}$ be the set of (commutative/pure) state polynomials.

	Naturally, (letter) state polynomials behave in the same manner as polynomials in operators and their expectations, i.e., operator state polynomials.
	A word $w = x_{a_1} \cdots x_{a_\xi}$ and its expectation value $\av{w}$
	can be evaluated on a state $\varrho$ and operators
	$X_{a_1},  \dots, X_{a_\xi}$
	so that $W = X_{a_1} \cdots X_{a_\xi}$
	and $\av{W}_\varrho = \av{X_{a_1} \cdots X_{a_\xi}}_\varrho$.
	Then,
	the empty word $1$ corresponds to the identity $\one$
	and $\av{1} = 1$.
	A monomial
	$(x_{a_1} \cdots x_{a_\xi}) \av{x_{b_1} \cdots x_{b_\mu}} \cdots
	\av{x_{\ell_1} \cdots x_{\ell_\nu}}$
	is evaluated as
	$(X_{a_1} \cdots X_{a_\xi}) \av{X_{b_1} \cdots X_{b_\mu}}_\varrho \cdots
	\av{X_{\ell_1} \cdots X_{\ell_\nu}}_\varrho$.
	By linear extension any state polynomial can be evaluated.

	In what follows we use the term state polynomials loosely
	to denote either type, that is letter or operator state polynomials.
	The important difference between the two types is
	that when assigning a set of values to the $\av{w}$,
	there does not necessarily exists a compatible quantum state and operators. Thus with some abuse of notation,
	given a set of $\av{w} \in \C$,
	it can be that there does not exist a $\varrho$
	and a set of operators $W$
	such that $\av{w}= \av{W}_\varrho$.

	\subsection{State polynomial optimization}\label{subsec:stpolyopt}

	Let $\varrho$ be a state and $S_0, \dots, S_u$ be the set
	of all operator state monomials of degree at most $\ell \in \N$
	in operators $X_1, \dots, X_m$.
	We set $S_0 = \one$ for simplicity.
	One can then form a moment matrix with entries
	\begin{align}
		M^{(\ell)}_{ij}=
		\tr\big( {S_i^\dagger} S_j \varrho\big) \quad \in \C^{(u+1) \times (u+1)}\,.
	\end{align}
	Then $M_{00} = 1$ and
	$M^{(\ell)}_{ij} = M^{(\ell)}_{rs}$ if $\langle S_i^\dag S_j \rangle = \langle S_r^\dag S_t \rangle$.
	The moment matrix $M^{(\ell)}$ is positive semidefinite.
	This can be seen from the spectral decomposition
	$ \varrho = \sum_v \lambda_v \dyad{\phi_v}$ with $\lambda_v \geq 0$;
	then for each $\ket{\phi_v}$ define the vector
	\begin{equation}
		\phi_v = \big( S_0 \ket{\phi_v}, S_1 \ket{\phi_v}, \dots, S_u \ket{\phi_v}\big)^\dagger \,.
	\end{equation}
	One sees that $M^{(\ell)} = \sum_v \lambda_v \phi_v \phi_v^\dagger \succeq 0$.

	Likewise, given a state polynomial inequality $g_k(X_1, \dots, X_m)\succeq 0$,
	we can form a moment matrix using all state monomials of up to degree
	$\ell - q_k$ where $q_k :=\lceil \frac{\text{deg}(g_k)}{2} \rceil$~\footnote{Choose $\ell$ large enough so that $\ell - q_k\geq 0$.},
	\begin{align}
		G^{(\ell-q_k)}_{ij} =  \tr( {S_i^\dagger} g_k(X_1,\dots,X_m) S_j \varrho)\,,
	\end{align}
	Also $G^{(\ell-q_k)}_{ij}$ is positive semidefinite,
	which follows from the fact that for every feasible point, $g_k\succeq 0$ can be factorized as a product $g_k=b_kb_k^\dagger$ of a matrix $b_k$.
	 Note that the entries of $G^{(\ell-q_k)}$
	 are linear combinations of the entries of $M^{(\ell)}$.

	The idea of state polynomial optimization is to optimize over a semidefinite program
	where $M^{(\ell)}$ and the matrices $G^{(\ell - q_k)}$, $k=1, \dots, r$
	are positive semidefinite matrix variables,
	such that they satisfy all linear constraints arising from relations among the entries of $M^{(\ell)}$ and the $G^{\ell - q_k}$.
	This provides a relaxation for the optimization of a state polynomial
	under a set of state polynomial constraints.

	To formulate the optimization problem, one uses letter state polynomials.
	For $s,t$ state polynomials we write the entries of the
	state moment matrix as
	$M^{(\ell)}_{st}=\av{s^*t}$.
	To simplify the notation, define
	$M^{(\ell)} (\av{s^* t}) := M^{(\ell)}_{st}$.

	The optimization problem from Eq.~\eqref{eq:optprob} is then approximated
	by the following semidefinite programming hierarchy~\cite{klep2023state,PhysRevLett.123.140503,pozas2019quantum} indexed by $\ell \in \N$,
	\begin{equation}\label{eq:genstpoly}
			\begin{aligned}
		\zeta_{\ell} \quad = \quad \inf_{M^{(\ell)}} \quad &   \langle \, f(x_1,\dots, x_m) \,\rangle\,,  \\
		\text{subject to} \quad & M^{(\ell)}(1)=1 \,,  \\
		& M^{(\ell)}\succeq 0 \,, \\ &G^{(\ell-q_k)} \succeq 0 \quad \quad\quad\quad\quad\quad\,\,\text{for}\quad 1\leq k \leq r\,, \\
		& M^{(\ell)}_{st} = M^{(\ell)}_{\alpha \beta} \quad\quad\quad\quad\quad\quad
		\text{if} \quad\quad\av{s^*t} = \av{\alpha^*\beta}  \,,  \\
		& G^{(\ell-q_k)}_{st}= \sum_{\alpha \beta} c_{\alpha\beta} M^{(\ell)}_{\alpha \beta} \quad\quad
		\text{if} \quad\quad \av{s^*g_k t}= \sum_{\alpha \beta} c_{\alpha\beta} \av{\alpha^*\beta}\,,
		\end{aligned}
	\end{equation}
	where $r$ is the number of constraints,
	$\ell$ is greater than or equal to the degree of the polynomial $f$,
	and $c_{\alpha\beta} \in \C$.

	The program~\eqref{eq:genstpoly} provides a lower bound $\zeta_{\ell}\leq \zeta^*$, but in general $\zeta_{\ell}\neq \zeta^*$. Let us consider convergence of $\zeta_\ell$ to $\zeta^*$.
	A constraint set is {\em balanced}, if $C^* = C$ and for every non-symmetric element $g_k^* \neq g_k$
	we also have $- g_k \in C$. This allows us to handle equality constraints by requiring both $g_k \succeq 0$ and $-g_k \succeq 0$.
	Corollary 6.1 in combination with Corollary 6.7 in \cite{klep2023state}
	states that
	if the constraint set is balanced and
	if there exists a $N>0$ such that
	$N - \sum_i x_i^* x_i  = \sum_j p_j^* p_j + \sum_k \lambda_k g_k$ for some
	$p_j \in \C\langle x \rangle_1$,
	$g_k \in C \cap \C \langle x \rangle_2$,
	and $\lambda_k§ \in \R_{\geq 0}$,
	then the sequence of relaxations in
	Eq.~\eqref{eq:genstpoly} converges to the optimum $\zeta^*$ of Eq.~\eqref{eq:optprob},
	\begin{align}\label{eq:arch}
		\lim_{\ell\rightarrow \infty} \zeta_{\ell} = \zeta^* \,.
	\end{align}
	Under certain flatness conditions,
	a variant of the Gelfand-Naimark-Segal (GNS) construction
	allows one to construct 
	state $\varrho$ and operators $\{X_i\}$
	on a Hilbert space
	$\mathcal{H}$ from the sequence $\{M^{(\ell)}\}_{\ell = 0}^\infty$
	satisfying the constraints of the optimization problem~\cite[Proposition 6.10]{klep2023state}.

	A short remark on imposing constraints:
	inequality constraints are imposed by building a localizing matrix
	$G^{(\ell-q_k)}_{st} \succeq 0$ from each constraint $g_k$.
	Equality constraints $g_i=0$ can be imposed in two different ways:
	either they are enforced by requiring both $g_k \succeq 0$ and $-g_k \succeq 0$, or 
	alternatively,
	a set of equalities $E = \{e_i \,|\, i=1,\dots, m_e\}$
	can be modeled by replacing $\bf S$
	by the quotient ring ${\bf S} / I$,
	where 
	$I=\left\{ \sum_{i} p _{i}g_i q_{i}: p_i,q_i\in {\bf S} \right\}$.
	A simplified construction of this can be done for group algebras~\cite[Section 7.1]{klep2023state}.

	\subsection{Quasi-Clifford algebras}\label{sect:qca}
	Let $F$ be a commutative field of characteristic not equal to $2$ and $m$ be a positive integer.
	Let $\{k_j\}_{j=1}^m \in F$ be non-zero elements
	and
	$\chi_{ij} \in \{0,1\}$ for $1\leq i<j\leq m$.
	A quasi-Clifford (QC) algebra is defined as an associative algebra over $F$ containing the identity
	formed by $m$ generators $ \alpha_1,\dots,\alpha_m$,
	such that
	\begin{equation}
	 \alpha_i^2 = k_i\,, \quad \quad \alpha_j \alpha_i = (-1)^{\chi_{ij}} \alpha_i \alpha_j \quad \text{for }i<j\,.
	\end{equation}
	Denote by $\mathbb{C}_{b}$ the QC algebra of one generator $\beta$ with $\beta^2=b$ and
	by $\mathbb{Q}_{cd}$ the QC algebra of two generators $\gamma$, $\delta$ with $\gamma^2=c$, $\delta^2=d$, and $\gamma\delta=-\delta\gamma$.
	\begin{theorem}{\cite[Theorem 2.7]{gastineau-hills_1982}}\label{thm:QC}
		A QC algebra $\mathscr{C}$ with $m$ independent generators is isomorphic to
		\begin{equation}\label{eq:thmQC}
			\mathscr{C}\cong\mathbb{C}_{b_1} \ot \dots \ot \mathbb{C}_{b_r} \ot \mathbb{Q}_{c_1d_1} \ot \dots \ot \mathbb{Q}_{c_sd_s} \,,
		\end{equation}
		where  $m=r + 2s$  with $r$, $s \geq 0$ and  $b_i$, $c_i$, $d_i$ are plus or minus products of some $k_j$'s.
		\end{theorem}

		Restricting to $\alpha_i^2 = k_i=1$ and $F$ to $\C$, then $\mathscr{C}$ contains only $\mathbb{C}_{\pm 1}$ and $\mathbb{Q}_{\pm1\pm1}$ as tensor factors.
		Ref.~\cite[Eq.~(3.1)]{gastineau-hills_1982} lists the irreducible representations
		of the generators.
\begin{equation}\label{eq:CQ_qubit}
\begin{array}{ll}
	\text{Field}  & \text{Generators} \\[0.1cm] \hline \\[-0.4cm]
	\mathbb{C}_{1} & \beta\rightarrow1 \quad \text{or} \quad \beta\rightarrow-1  \\[0.2cm]
	\mathbb{C}_{-1} &  \beta\rightarrow i \quad \text{or} \quad \beta\rightarrow-i  \\[0.2cm]
	\mathbb{Q}_{\pm1,1}&  \gamma\rightarrow\begin{pmatrix} 0 & \pm1 \\ 1 & 0 \end{pmatrix}, \quad  \delta\rightarrow\begin{pmatrix} 1 & 0 \\ 0 & -1 \end{pmatrix}  \\[0.5cm]
	\mathbb{Q}_{1,-1}&  \gamma\rightarrow\begin{pmatrix} 0 & \phantom{-}1 \\ 1 & \phantom{-}0 \end{pmatrix}, \quad \delta\rightarrow\begin{pmatrix}[1] 0 & -1 \\ 1 & 0 \end{pmatrix}  \\[0.5cm]
	\mathbb{Q}_{-1,-1}&  \gamma\rightarrow\begin{pmatrix} 0 & -1 \\ 1 & 0 \end{pmatrix}, \quad \delta\rightarrow\begin{pmatrix} i & 0 \\ 0 & -i \end{pmatrix}
\end{array}
\end{equation}

		This shows that
		$\mathbb{C}_{\pm 1} \cong 2F$ and  $\mathbb{Q}_{\pm1,\pm1} \cong F_2$,
		where $2F = F \oplus F$ and $F_2$ is the algebra of complex $2\times 2$ matrices.
		Thus, Eq.~\eqref{eq:thmQC} can be written as a direct sum
		\begin{align}\label{eq:qc_direct_sum}
		\mathscr{C}\cong 2F \ot \cdots \ot 2F \ot F_2 \ot \cdots \ot F_2=2^rF_{2^s} \,.
		\end{align}
		Here we use that $\otimes$ distributes over $\oplus$ and $F_m \ot F_n \cong F_{nm}$.
 		As a consequence, any quasi-Clifford algebra over $\C$ with $k_i = 1$

		is isomorphic
		to a direct sum of
		$2^r$ algebras, each formed (up to phases)
		by $s$-qubit Pauli operators,
 		where $m = r+2s$.

\begin{remark}
The direct sum in Eq.~\eqref{eq:qc_direct_sum} comes from the fact that given a set of operators with some commutation/anti-commutation relations and squaring to $\one$,
the mapping to Pauli operators is not unique. For example, given $x_1,x_2,x_3$ with $x_i x_j = -x_j x_i$ and $x_i^2=\one$,
the elements $x_1$, $x_2$, $x_3$ can be mapped to
the Pauli matrices $X$, $Z$, $Y$, respectively, but also to $X$, $Z$, $-Y$, respectively. Different blocks in the direct sum
correspond to different inequalivalent mappings to Pauli operators.
\end{remark}

\section{Complete SDP hierarchy bounding quantum codes}\label{sec:stcode}

Our aim is now to express the conditions for the existence of quantum codes
as a state polynomial optimization problem.
For this we formulate the Knill-Laflamme conditions,
the quantum MacWilliams identity,
and the constraints for the state to be proportional to a projector of rank $K$ as state polynomial constraints.
This gives rise to a complete hierarchy for the existence of an $(\!(n,K,\delta)\!)_2$ code.

\subsection{Framework}

	Consider the $n$-qubit error basis ${\EE_n}$ formed by $n$-fold tensor products of Pauli matrices~[see Eq.~\eqref{eq:pauli}].
	Define a set of letters $\mathscr{E}_n$ so that to each Pauli error $E_a \in {\EE_n}$ we associate a letter $e_a \in \mathscr{E}_n$. We require this to be a group isomorphism, so that
	\begin{equation}\label{eq:rel}
	 	 e_ a e_b= \omega_{abc} e_c \quad \text{if} \quad E_a E_b= \omega_{abc} E_c
	 	 \quad\quad \text{for all $a,b,c$}\,.
	\end{equation}
	For the Pauli group it holds that $\omega_{abc} \in \{\pm1, \pm i\}$.
	The Pauli relations also imply that $ e_a^* = e_a$ and $e_a^*e_a = 1$,
	where the empty word $1$ corresponds to identity operator~$\one$.
	For state polynomials,
	we extend the relations in Eq.~\eqref{eq:rel}
	linearly to expectations,
	so that additionally we also have
	$\langle e_ a e_b \rangle = \omega_{abc} \langle e_c \rangle$ if $E_a E_b= \omega_{abc} E_c$.
	In effect, this imposes the Pauli relations onto the hierarchy
	in terms of the quotient ring
	$\bold{S}_n / I$ where $\bold{S}_n$ is the set of all state polynomials formed by the letters in $\mathscr{E}_n$ and $I$ is the ideal generated by the $n$-qubit Pauli relations.
	Let $M^{(\ell)}$ be the moment matrix from the state polynomial optimization framework [see Eq.~\eqref{eq:genstpoly}],
	indexed by all state monomials of degree at most~$\ell$.
	We impose all constraints from Eq.~\eqref{eq:rel} on $M^{(\ell)}$.
	Recall that
	the entries of the state moment matrix are denoted by 
	\begin{equation}
		M^{(\ell)} (\av{s^* t}) := M^{(\ell)}_{st}\,.
	\end{equation}

	We now formulate all the conditions for a quantum code
	in terms of constraints on state polynomials, that is, as
	constraints on the entries of $M^{(\ell)}$.

	\subsection{Quantum MacWilliams identity} The quantum weight enumerator $A_j(\varrho)$ in Eq.~\eqref{eq:enum}
	is a state polynomial in~$\varrho$.
	We can then define a state polynomial representing the quantum weight enumerator,
	\begin{align}\label{eq:Astpoly}
		A_j^{(\ell)}
		=\sum_{\substack{e\in \mathscr{E}_n \\\wt(e)=j}} M^{(\ell)}(\av{e^*}\av{e\,}) \,,
	\end{align}
	so to correspond to
	$
	A_j(\varrho) =  \sum_{E\in {\EE_n}\,, \wt(E) = j}
	\tr(E^\dag \varrho) \tr(E \varrho)
	$.

	However, the quantum weight enumerator $B_j(\varrho) = \sum_{E\in {\EE_n}\,, \wt(E) = j} \tr(\varrho E^{\dag} \varrho E)$ does not have such a straightforward interpretation as a state polynomial.
	But recall that the quantum MacWilliams identity [see Eq.~\eqref{eq:mac}] linearly transforms the $A_i(\varrho)$ to the $B_j(\varrho)$ enumerators. We can then linearly transform $A^{(\ell)}_j(\varrho)$ to $B^{(\ell)}_j(\varrho)$ using
	\begin{align}\label{eq:QMW_spo}
	B^{(\ell)}_j & = 2^{-n}\sum^n_{i=0} K_j(i;n)A^{(\ell)}_i\,,
	\end{align}
	where $K_j(i;n)$ are the Krawtchouk polynomials defined in Eq.~\eqref{eq:krawpoly}.

	\subsection{Knill-Laflamme conditions}
	We now impose the Knill-Laflamme conditions as constraints on $M^{(\ell)}$. These conditions can be formulated in terms of
	quantum weight enumerators.
	Thus we demand that
	\begin{equation}\label{eq:KLM_spo}
	 K \Bl_j = \Al_j \quad \text{for } \quad  0\leq j < \delta\,.
	\end{equation}

	\begin{remark}\label{rmk:extraconst}
	One can additionally include the remaining linear programming constraints shown in Eq.~\eqref{eq:lp},
	\begin{align}
	S^{(\ell)}_j &\geq 0 \,,\nn \\
	K \Bl_j &\geq \Al_j \,.
	\end{align}

	Here $S^{(\ell)}_j$ mirrors the shadow enumerator of Eq.~\eqref{eq:shadowexp}.
	The condition $A^{(\ell)}_j, \, B^{(\ell)}_j \geq 0$ is redundant, since $\Gamma \succeq 0$ contains the elements of $A^{(\ell)}_j$ in its diagonal.
	These extra constraints do not have an impact on completeness.
	\end{remark}

\subsection{Projector constraints}

	We now impose conditions for a state $\varrho$ to be proportional to a projector $\Pi$ of rank $K$ on $M^{(\ell)}$.
	The state $\varrho$ then fulfills for all $m \in \N$,
	\begin{align}\label{eq:projconst}
	\tr(\varrho^{m})=\frac{1}{K^{m-1}}\,.
	\end{align}

	The converse also holds, namely that
	a state $\varrho$ satisfying Eq.~\eqref{eq:projconst} for all $m \in \N$ must be proportional to a projector of rank $K$.
	To see this, note that the set $\{\tr(A^m)\,:\, m=1,\dots, r\}$
	determines the characteristic polynomial of a square matrix $A$ of size $r$.
	The zeros of this characteristic polynomial then determine the eigenvalues of~$A$.
	Fixing the moments as in Eq.~\eqref{eq:projconst} implies that the state $\varrho$ is proportional to a projector of rank $K$.

	We now show how to write Eq.~\eqref{eq:projconst} in terms of state polynomials.
	For this we
	decompose the cyclic permutation $(1,\dots,m)$ into elementary transpositions,
	\begin{equation}
	 (1,\dots, m)=(m ,\, m-1) \dots(2, 3) (1, 2) \,.
	\end{equation}
	Second, expand the swap operator through the Pauli basis of $(\C^2)^{\ot n}$~\cite{lecture_Wolf},

	\begin{equation}
	 (1,2) ={2^{-n}}\sum_{E \in {\EE_n}} E\ot E^\dagger\,.
	\end{equation}
	The generalization of the swap trick~\cite{huber2020positive} then gives
	\begin{align}\label{eq:trace}
		\tr(\varrho^{m})&=\tr\big((1,\dots, m)\varrho^{\ot m}\big) \nn\\
		&=2^{-n(m-1)}\!\!\!\!\!\!\sum_{E_1,\dots,E_{m-1} \in {\EE_n}}
		\tr(E_{1} \varrho)
		\tr(E^\dagger_{1} E_{2} \varrho)
		\dots
		\tr(E^\dagger_{m-2} E_{m-1} \varrho)
		\tr(E^\dagger_{m-1} \varrho)\,.
	\end{align}
	Therefore, Eq.~\eqref{eq:projconst} is now expressed as
	 \begin{align}\label{eq:stprop}
	\sum_{E_1,\dots,E_{m-1} \in {\EE}_n}
	\tr(E^\dagger_{1} \varrho)
	\tr(E_{1} E^\dagger_{2}\varrho)
	\dots \tr(E_{m-2} E^\dagger_{m-1} \varrho)
	\tr(E_{m-1} \varrho)
	= \Big( \frac{2^n}{K}\Big)^{m-1}\,.
	\end{align}
	The condition~\eqref{eq:projconst} for a state to be proportional to a projector of rank $K$, i.e.
	$\varrho = \Pi / K$, corresponds to the following constraint on the moment matrix for $m =2,\dots,2\ell$.
	\begin{align}\label{eq:stpropM}
	P^{(\ell)}_m =
	\sum_{e_{1},\dots,e_{{m-1}} \in {\mathscr{E}_n}}
	M^{(\ell)} \big(
	\langle e^*_{1} \rangle
	\langle e_{1} e^*_{2} \rangle
	\dots
	\langle e_{{m-2}} e^*_{{m-1}} \rangle
	\langle e_{{m-1}} \rangle \big)
	=\Big( \frac{2^n}{K}\Big)^{m-1} \,.
	\end{align}

\begin{remark}\label{rmk:alt_proj} As pointed out to us by Dion Gijswijt, alternatively to Eq.~\eqref{eq:projconst}, the constraints
\begin{align}\label{eq:alt_proj}
\tr(E_a\varrho^2)=\frac{1}{K}\tr(E_a\varrho) \quad \text{for} \quad a=0,\dots, 4^n-1
\end{align}
also give necessary and sufficient conditions for $\varrho$ to be proportional to a projector of rank $K$, due to the fact that the $E_a$ form a complete basis.
These conditions can also be written in terms of state polynomials:
similar to Eq.~\eqref{eq:trace}, expand
\begin{align}\label{eq:alt_proj_swap}
	\tr(E_a \varrho^2) &=
	\tr\left((1,2)E_a \varrho \ot \varrho\right) \nn \\
	&= \frac{1}{2^n}\sum_{E_b\in \EE_n}\tr(E_a E_b^\dagger\varrho)  \tr(E_b\varrho) \,.
\end{align}
Then, the condition from Eq.~\eqref{eq:alt_proj} can be written in terms of elements of $M^{(\ell)}$ as
\begin{align}\label{eq:alt_proj_spo}
	\frac{1}{2^n} \sum_{e_b \in \mathscr{E}_n} M^{(\ell)}\left(\av{e_a e_b^\dagger}\av{e_b}\right) =  \frac{1}{K}M^{(\ell)}\left(\av{e_a}\right) \quad \text{for} \quad a=0,\dots, 4^n-1
	\end{align}
The main advantage of Eq.~\eqref{eq:alt_proj_spo} in comparison to Eq.~\eqref{eq:alt_proj} is that already at the level $\ell=2$
a full set of projector conditions is included.

\end{remark}

\subsection{An SDP hierarchy for qubit codes}\label{sect:SDPhierc}
	We now formulate a semidefinite programming hierarchy
	that contains the necessary and sufficient conditions for a state to correspond to a quantum code:
	the projector constraints guarantee that the state is
	proportional to a projector of rank $K$, while the Knill-Laflamme conditions ensure that this projector 
	corresponds to a quantum code subspace of minimum distance $\delta$.

	From now on consider state polynomials in the quotient ring $\bold{S}_n/I$,
	where $I$ is the ideal generated by the $n$-qubit Pauli relations.
	Then consider the following hierarchy indexed by $\ell \in \N$,
	with the $M^{(\ell)}$ is indexed by all state monomials  in $\bold{S}_n/I$ of degree at most $\ell$ in hermitian variables in $\EE_n$.
	\begin{align}\label{eq:code_SDP}
	\text{find} \quad &   M^{(\ell)} \succeq 0  \nn \\
	\text{subject to} \quad
	& M^{(\ell)}(1)=1 \,, \nn \\
	& M^{(\ell)}_{st} = M^{(\ell)}_{\alpha \beta}  \quad\quad\quad\,\,\,\,\text{if} \quad \av{s^*t}= \av{\alpha^*\beta}\,, \nn \\
	&KB^{(\ell)}_j=A^{(\ell)}_j \quad\quad\quad\,\, \text{for}
	\quad 0 \leq j < \delta \nn, \\
	&  P^{(\ell)}_{m} =
	\Big(\frac{2^n}{K}\Big)^{m-1} \quad\quad \text{for }\,\,\, m=2,\dots, 2\ell\,,
	\nn \\
	&\text{$A_j^{(\ell)}$, $B^{(\ell)}_j$, $P^{(\ell)}_{m}$ are given by Eqs.~\eqref{eq:Astpoly}, \eqref{eq:QMW_spo}, \eqref{eq:stpropM}  respectively.}
	\end{align}

 	For every $\ell \in \N$, the SDP~\eqref{eq:code_SDP}
 	is a relaxation of the optimization problem in Eq.~\eqref{eq:find}.
	\setcounter{thmA}{0}

	\begin{theorem}\label{thm:code}
		A quantum code with parameters $(\!(n,K,\delta)\!)_2$ exists
		if and only if
		the semidefinite programming hierarchy of Eq.~\eqref{eq:code_SDP} is feasible
		at every level $\ell \in \N$.
	\end{theorem}
	\begin{proof}
	\noindent
	($\Longrightarrow$):
	Let $\varrho$ be a state proportional to the projector of a quantum code with parameters $(\!(n,K,\delta)\!)_2$.
	Given the set of Pauli operators $\EE_n$, we can construct a moment matrix
	$M^{(\ell)}_{ST} = \tr(S^\dag T \varrho)$
	indexed by all Pauli state monomials of at most degree $\ell$.
	It is clear that for all $\ell \in \N$,
	$M^{(\ell)}\succeq 0$ and $M^{(\ell)}(\one)=1$ [see Section~\ref{subsec:stpoly}].
	Additionally, $M$ satisfies the projector conditions in Eq.~\eqref{eq:stpropM} as $\varrho$ is a density matrix proportional to a projector of rank $K$.
	Finally,
	$M$ satisfies the state polynomial quantum MacWilliams identity [Eq.~\eqref{eq:QMW_spo}] and Knill-Laflamme conditions in enumerator form [Eq. ~\eqref{eq:KLM_spo}].
	Thus, if there exists a code with parameters
	$(\!(n,K,\delta)\!)_2$, then
	every level of the hierarchy is feasible.

	\noindent
	($\Longleftarrow$):
	Now given the parameters $(\!(n,K,\delta)\!)_2$, suppose every level of the hierarchy Eq.~\eqref{eq:code_SDP} is feasible.
Corollary 6.1 and Proposition 6.7 of \cite{klep2023state}
imply that the sequence of relaxations in Eq.~\eqref{eq:code_SDP}
converges to the optimum of

	\begin{align}\label{eq:code_SDP__}
	\text{find} \quad \quad& \HH, \{X_a\}_{a=0}^{4^n-1}, \varrho \nn\\
	\text{subject to}\quad
	& KB_j(\varrho) = A_j(\varrho) \quad \text{for}
	\quad 0 \leq j < \delta \,,\nn\\
	&
	A_j(\varrho)
	=\sum_{{X_a, \, \wt(X_a)=j}}
		 \tr(X_a \varrho)  \tr(X_a^\dag \varrho)  \,,\nn\\
		 	&B_j(\varrho) = 2^{-n}\sum^n_{i=0} K_j(i;n)A_i(\varrho)\,, \nn\\
		 	%&\text{$B_j(\varrho)$ given by Eq.~\eqref{eq:macexp}}\,, \nn\\
		 &\text{$\{X_a\}_{a=0}^{4^n-1}$ satisfies all $n$-qubit Pauli relations}\,, \nn\\
		& P_m =
	\sum_{ {a_1},\dots,{a_{m-1}} = 0}^{4^n-1}
	\tr(X^\dagger_{a_1} \varrho)
	\tr(X_{a_1} X^\dagger_{a_2}\varrho)
	\dots \tr(X_{a_{m-2}} X^\dagger_{a_{m-1}} \varrho)
	\tr(X_{a_{m-1}} \varrho) \nn
	\\ &=\Big(\frac{2^n}{K}\Big)^{m-1} \quad \text{for all } m \in \N\,,
	\end{align}
	where $\wt$ is a suitable function defined on the $X_i$ corresponding to the weights of the elements in $\mathcal{E}_n$, and the $X_i$ are hermitian.. This is because the set of constraints from~\eqref{eq:code_SDP__} is balanced and
	$4^n - \sum_{e\in \mathcal{E}_n} e e^* = 0$.

	Now we show that the optimization problem in Eq.~\eqref{eq:code_SDP__} completely characterizes a $n$-qubit quantum code.
	First, every finite group algebra decomposes by the Artin-Wedderburn theorem into a direct sum of matrix algebras.
	In our case, Theorem~\ref{thm:QC} restricts the set of operators $X_a$
	to be isomorphic to a direct sum of tensor products of Pauli matrices:
	considering the commutation relations alone, the operators form a quasi-Clifford algebra over $\C$,
	whose elements square to the identity, and no other relations.
	The representation in Eq.~\eqref{eq:CQ_qubit} together with the $n$-qubit Pauli relations
	baked into $\mathscr{E}_n/I$
	implies that these operators are isomorphic
	to the $n$-qubit Pauli basis $\EE_n$.
	As a consequence, $\varrho$ is isomorphic to a $n$-qubit quantum state.
	Finally, the Knill-Laflamme and projector condiditions force $\varrho$ to be the projector onto a code subspace.
	This ends the proof.
	\end{proof}

	In the following, we show how an intermediate level of this hierarchy
	yields a Lov\'asz bound [Section~\ref{sec:Lovasz}], whose symmetrization results a quantum Delsarte bound
	[Section~\ref{sec:Delsarte}].
	Using the Terwilliger algebra of the Hamming association scheme, we also provide efficiently computable symmetry-reduced SDP relaxations [Section~\ref{subsect:sdp_aver}]. Finally, we provide an alternative nonexistence proof of the
	$(\!(7,1,4)\!)_2$ code [Section~\ref{sec:7qb}],
	and show that $(\!(8,9,3)\!)_2$ and $(\!(10,5,4)\!)_2$ codes do not exist [Section~\ref{sec:sym_red_numerics}].

\section{A Lov\'asz bound}
\label{sec:Lovasz}

We now present a relaxation of the hierarchy from Eq.~\eqref{eq:code_SDP}
that leads to a bound involving the
Lov\'asz theta number [Section~\ref{sec:lovasz_number_def}].
The relaxation is indexed by state monomials of the form
$\langle e_a \rangle e_a$ where the $e_a$ correspond to letters from $\mathscr{E}_n$ [see Eq.~\eqref{eq:rel}].
For readability, we will use elements from the $n$-qubit Pauli basis $E_a \in \EE_n$ for indexing rather the corresponding letters $e_a \in \mathscr{E}_n$.

\subsection{An intermediate relaxation level}\label{subsec:inter_relax}
Consider a moment matrix $\Gamma$ indexed by the symbols $\av{E^\dagger_\aaa}E_\aaa$ with $E_\aaa\in \EE_n$,
where $\av{E_\aaa}$ inherits the properties of an expectation $\tr(E_\aaa\varrho)$ for some state $\varrho$,
\begin{align}\label{eq:relax}
	\Gamma_{\aaa\bbb}=\av{E_\aaa^\dagger} \av{E_\bbb} \av{E_\aaa E_\bbb^\dagger}\,.
\end{align}
Such matrix has the form,
\begin{align}
	\kbordermatrix{
		& \one    & \av{E_1} E_1^\dagger & \cdots    & \av{E_{N}} E_{N}^\dagger \\
		\one &  1   & \Gamma_{01}           &\cdots     & \Gamma_{0 N} \\
		\av{E_1^\dagger} E_1 & & \Gamma_{11}    & \hdots   & \Gamma_{1N}\\
		\vdots &       &   & \ddots    & \vdots\\
		\av{E_{N}^\dagger} E_{N}& &           &           & \Gamma_{NN}
	} \succeq 0 \,,
	\label{eq:lovasz-moments}
\end{align}
with $N=4^n-1$ and  $\Gamma_{\aaa0} = \Gamma_{\aaa\aaa}$,
which represents $\tr(E^\dagger_\aaa \varrho) \tr(E_\aaa \varrho)$.
Now note that if $\Gamma$ is a feasible solution of a semidefinite program
with objective function and constraints that depend only on the real part of $\Gamma$,
then $\operatorname{Re}(\Gamma) = (\Gamma + \Gamma^T)/2$
is also feasible with the same objective value.
Without loss of generality, we can therefore restrict $\Gamma$ to be real.
This leads to the constraint
\begin{equation}\label{eq:sim_anti_comm}
 \Gamma_{\aaa\bbb}=0 \quad\quad  \text{if} \quad\quad  E^\dagger_\aaa E_\bbb  + E_\bbb E^\dagger_\aaa = 0\,.
\end{equation}
This follows from the fact that when two Pauli strings
$E_\aaa$ and $E_\bbb$ anti-commute, then $E_
\aaa E_\bbb^\dagger$ is imaginary.

\subsection{Structure constraint}
When $E_\aaa$ and $E_\bbb$ commute,
the matrix $\Gamma$ satisfies an additional constraint
related to a multiplication of three terms:
$\tr(E^\dagger_\aaa \varrho)$,
$\tr(E_\bbb \varrho)$,
and $\tr(E_\aaa E^\dagger_\bbb \varrho)$.
For example, if
$E^\dagger_{\ccc'} = \omega_{\aaa\bbb}E_\aaa E^\dagger_{\bbb}$
and
$E_{\ddd'} = E^\dagger_\aaa$,
then
$\omega_{ab} E_{\ccc'}E^\dagger_{\ddd'} = E_\bbb$,
leading to the condition
\begin{align}
\av{E^\dagger_\aaa}\av{E_\bbb}\av{E_\aaa E^\dagger_\bbb} = \av{E^\dagger_{\ccc'}}\av{E_{\ddd'}}\av{E_{\ccc'} E^\dagger_{\ddd'}}\,.
\end{align}
Therefore, we can write such constraint as
\begin{align}\label{eq:structure_constraint}
	\Gamma_{\aaa\bbb}=\Gamma_{\ccc'\ddd'}\,, \quad \quad\quad  &\text{if} \quad\quad  E_\aaa E^\dagger_\bbb  - E^\dagger_\bbb E_\aaa = 0\quad \text{and} \quad\, \nn\\ \quad & (E^\dagger_\aaa,E_\bbb,\omega_{\aaa\bbb}E_\aaa E^\dagger_\bbb)\quad\text{a permutation of} \quad(E^\dagger_{\ccc'},E_{\ddd'},\omega_{\ccc'\ddd'}E_{\ccc'}E^\dagger_{\ddd'})\,,
\end{align}
where $(\omega_{\aaa\bbb}E_{\aaa}E^\dagger_{\bbb}) \in \EE_n$ and $\omega_{\aaa\bbb}\in \{\pm 1\}$, since when $E_\aaa$ and $E^\dagger_\bbb$ commute $\omega_{\aaa\bbb}$ must be real.
Note that Eq.~\eqref{eq:structure_constraint} includes $\Gamma_{\aaa0}=\Gamma_{\aaa\aaa}$ as a special case.

\subsection{Projector constraints}
As done for the matrix $M^{(\ell)}$ in Eq.~\eqref{eq:stpropM},
we relax the projector constraints $\tr(\rho^2) =1/K$ and $\tr(\rho^3) =1/K^2$ by expressing them, respectively,
through sums over elements of $\Gamma$,
\begin{align}\label{eq:projconst_lvl}
	\sum^{N}_{\aaa=0}\Gamma_{\aaa\aaa}=\frac{2^n}{K}\,, \quad\quad\quad  \sum^{N}_{\aaa,\bbb=0}\Gamma_{\aaa\bbb}=\frac{4^n}{K^2}\,.
\end{align}
Another constraint follows from
$\tr(E_a\varrho^2)=\frac{1}{K}\tr(E_a\varrho)$ as shown in Remark~\ref{rmk:alt_proj}: the sum $\sum^N_{b=0}\Gamma_{ab}$ approximates 
\begin{align}\label{eq:1-sum}
\sum^N_{b=0}\tr(E^\dagger_a\rho)\tr(E_b\rho)\tr(E_a E^\dagger_b\rho) &= \tr(\rho E^\dagger_a) \sum^N_{b=0} \tr((E_b\ot E_a E^\dagger_b)\rho^{\ot 2}) \nn \\
&=\tr(\rho E^\dagger_a) \tr( (\one \ot E_a)\sum^N_{b=0} (E_b\ot E^\dagger_b)\rho^{\ot 2})  \nn \\
&=\tr(\rho E^\dagger_a) \tr( (\varrho \ot \varrho E_a)(1,2)) \nn \\
&= 2^n \tr(\rho E^\dagger_a)\tr(\rho^2 E_a )\,,
\end{align}
where we used that $ \sum_{E_a \in \EE_n} E_a \ot E^\dag_a =2^n(1,2)$. Since $\tr(E_a\varrho^2)=(1/K)\tr(E_a\varrho)$, we obtain
 \begin{align}\label{eq:1-sum_p}
 \sum^{N}_{\bbb=0}\Gamma_{\aaa\bbb}=\frac{2^{n}}{K} \Gamma_{\aaa\aaa} \,.
\end{align}
Note that Eq.~\eqref{eq:1-sum_p} implies the first and second conditions of Eq.~\eqref{eq:projconst_lvl}, 
which corresponds to the projector constraint of degree two, $\tr(\rho^2)=1/K$, and three, $\tr(\rho^3)=1/K^2$, respectively.
The first is obtained by setting $a=0$ in the Eq.~\eqref{eq:1-sum_p} and noting that $\Gamma_{\aaa\aaa} = \Gamma_{\aaa0}$ and $\Gamma_{00} = 1$;
the second is obtained by summing Eq.~\eqref{eq:1-sum_p} over all values of $a$
and taking into count the first condition of Eq.~\eqref{eq:projconst_lvl}.

\subsection{Knill-Laflamme conditions}
Similar to Eq.~\eqref{eq:KLM_spo}, we can approximate the Knill-Laflamme conditions $KB_j(\varrho) = A_j(\varrho)$ if $0\leq j<\delta$ using elements of $\Gamma$
together with $KB_j(\varrho) \geq A_j(\varrho)$ for $j>\delta$ [see Remark~\ref{rmk:extraconst}],
\begin{align}\label{eq:klred}
	\frac{K}{2^n}\sum^n_{i=0} K_j(i;n) \!\!\sum_{\substack{E_{\aaa} \in \EE_n \\ \wt(E_\aaa)=i}} \Gamma_{\aaa0}
	\geq \sum_{\substack{E_\bbb \in \EE_n \\ \wt(E_\bbb)=j}} \Gamma_{0\bbb} \quad\quad \text{with equality for} \quad 1 <j < \delta  \,.
\end{align}
Note that we did not include the condition $KB_0(\varrho)=A_0(\varrho)$ because it is already approximated by $\sum^{N}_{\aaa=0}\Gamma_{\aaa\aaa}=2^n/K$.
This can be seen using Eq.~\eqref{eq:krawpoly}, which gives $K_0(i,n)=1$ for all~$i$. Therefore,  $K2^{-n}\sum^n_{i=0} A_i(\varrho) = A_0(\varrho)$. Since $A_0(\varrho)=1$ and $\Gamma_{\aaa0}=\Gamma_{\aaa\aaa}$, the trace of $\Gamma$ approximates $\sum^n_{j=0} A_j(\varrho)=2^n/K$.

\subsection{Stabilizer codes}

Let $S$ be the stabilizer group of a stabilizer code.
For a stabilizer code, $\Gamma$ represents $\chi \chi^T$ where $\chi$ is the incidence vector of the stabilizer group.
This can be seen in the following way:
given a code, $\Gamma$ has entries $\tr(E_\aaa^\dagger \varrho)\tr(E_\bbb \varrho)\tr(E_\aaa E^\dagger_\bbb \varrho)$.
But $\tr(E_\aaa \varrho) = 1$ if $E_\aaa \in S$, and $\tr(E_\aaa \varrho) = 0$ otherwise.
Thus $\Gamma_{\aaa\bbb} = 1$ if both $E_\aaa, E_\bbb \in S$ and zero otherwise.

We can use this fact in the following way.
For a positive semidefinite matrix, the determinant of every $2\times 2$ minor is non-negative,
leading to
$\Gamma_{\aaa\bbb} \leq\sqrt{\Gamma_{\aaa\aaa}\Gamma_{\bbb\bbb}}$.
Due to $\Gamma_{\aaa\bbb} \in \{0,1\}$ this implies that $\Gamma_{\aaa\bbb} \leq\sqrt{\Gamma_{\aaa\aaa}}$,
giving a constraint for stabilizer codes,
\begin{align}\label{eq:stab_code_Gamma}
	0 \leq \Gamma_{\aaa\bbb} \leq \Gamma_{\aaa\aaa}\,.
\end{align}

\subsection{Pure codes}
In the special case of pure codes, an extra constraint can be added.
Recall that a code is said to be pure if 
$\tr(E^\dagger_\aaa E_\bbb \Pi)=0$ for all $0 < \wt(E^\dagger_\aaa E_\bbb) <\delta$.
Because $\Gamma_{\aaa\bbb}=\av{E^\dagger_\aaa}\av{E_\bbb}\av{E_\aaa E_\bbb^\dagger}$, we have additionally that
\begin{align}\label{eq:extra_zeros}
	\Gamma_{\aaa\bbb}=0
	\quad\quad
	\text{if} \quad\quad 0 < \wt(E^\dagger_\aaa E_\bbb)<\delta\,, \quad 0 < \wt(E^\dagger_\aaa)<\delta\,, \quad \text{or} \quad  0 < \wt(E_\bbb)<\delta\,.
\end{align}
In particular, choosing $b=0$ (since $E_0= \one$) implies that $\Gamma_{\aaa\aaa}=\Gamma_{\aaa0}=0$ for all $\aaa$ with $0<\wt(E_\aaa)<\delta$.
Eq.~\eqref{eq:extra_zeros} strengthens Eq.~\eqref{eq:klred}.

\subsection{Main semidefinite program relaxation for qubits}

Let $\Gamma$ be indexed as in Eq.~\eqref{eq:relax} and
without loss of generality let $\Gamma$ be real.
Consider a SDP relaxation to Eq.~\eqref{eq:find}:
\begin{align}\label{eq:sdpgamma}
	\text{find} 		\quad & \Gamma \succeq 0\,, \nn \\
	\text{subject to} 	\quad
	& \Gamma_{00}=1 \,, \nn\\
	&	\Gamma_{\aaa\bbb}=\Gamma_{\ccc\ddd} \,\,\quad\text{if} \quad  E_\aaa E^\dagger_\bbb  = E^\dagger_\bbb E_\aaa
	\quad \text{and}
	\quad (E^\dagger_\aaa,E_\bbb,\omega_{\aaa\bbb}E_\aaa E^\dagger_\bbb)\quad\, \nn\\ \quad & \hspace{3.3cm} \text{a permutation of} \,\,(E^\dagger_{\ccc},E_{\ddd},\omega_{\ccc\ddd}E_{\ccc}E^\dagger_{\ddd})\,, \nn\\
	&\Gamma_{\aaa\bbb}=0 \quad\quad\quad \text{if}\quad  E^\dagger_\aaa E_\bbb + E_\bbb E^\dagger_\aaa = 0\,, \nn \\
	&\sum^{N}_{\aaa=0}\Gamma_{\aaa\aaa}=\frac{2^n}{K} \,, \nn \\
	& \sum^N_{\bbb=0}\Gamma_{\aaa\bbb}=\frac{2^{n}}{K} \Gamma_{\aaa\aaa}\,, \nn \\
	& \frac{K}{2^n}\sum^n_{i=0} K_j(i,n) \!\!\sum_{\substack{E_{\aaa} \in \EE_n \\ \wt(E_\aaa)=i}} \Gamma_{\aaa0}
	\quad \geq \quad
	\sum_{\substack{E_\bbb \in \EE_n \\ \wt(E_\bbb)=j}} \Gamma_{0\bbb}  \quad \quad \text{with equality for} \quad 0 <j < \delta \,.
\end{align}
The fact that $\Gamma \succeq 0$ implies that its diagonal elements are non-negative.
This recovers the non-negativity of $A_j(\varrho) = \sum_{E_\aaa \in \EE_n ,\, \wt(E_\aaa) = j} \langle E_\aaa^\dag \rangle \langle E_\aaa \rangle \geq 0$
corresponding to $\sum_{E_\aaa \in \EE_n ,\, \wt(E_\aaa) = j} \Gamma_{\aaa\aaa}$.
Eq.~\eqref{eq:sdpgamma} incorporates the Knill-Laflamme and the projector conditions and the structure relations and therefore,
it contains all constraints of the hierarchy in Eq.~\eqref{eq:code_SDP} for the indexing set $\av{E^\dagger_\aaa} E_\aaa$.

We can further constrain the SDP above by including [see Section~\ref{subsect:LP}]:

\begin{enumerate}
\item The shadow inequalities $S_j(\rho) \geq 0$ [Eq.~\eqref{eq:shadowexp}],
\begin{equation}\label{eq:shadowexp_x}
	%S_j(\Pi) =
	2^{-n}\sum^n_{i=0} (-1)^i K_j(i;n) \sum_{\substack{E_\aaa \in \EE_n  \\ \wt(E_\aaa) = i}} \Gamma_{\aaa\aaa} \quad \geq \quad 0
\end{equation}
with equality for $K=1$ and $n-j$ odd~\cite[Theorem 12]{796376}.

\item For stabilizer codes add $0 \leq \Gamma_{\aaa\bbb} \leq \Gamma_{\aaa\aaa}$ [Eq.~\eqref{eq:stab_code_Gamma}] and \cite{681315, PhysRevA.69.052330}
\begin{align}
	\sum_{j\geq 0} \sum_{E_\aaa \in \EE_n ,\, \wt(E_\aaa) = 2j} \Gamma_{\aaa\aaa} =
	\begin{cases}
2^{n-\log_2(K)-1} \quad \text{(Type I)}\,, \\  2^{n-\log_2(K)} \quad \quad \!\text{(Type II)}\,.
	\end{cases}
\end{align}

\item For pure codes add [Eq.~\eqref{eq:extra_zeros}],
\begin{align}\label{eq:pure_code_conditions_Gamma}
\Gamma_{\aaa\bbb}=0
\quad\quad
\text{if} \quad\quad 0 < \wt(E^\dagger_\aaa E_\bbb)<\delta\,, \quad 0 < \wt(E^\dagger_\aaa)<\delta\,, \quad \text{or} \quad  0 < \wt(E_\bbb)<\delta\,.
\end{align}

\end{enumerate}

A relaxation of the complete hierarchy in Theorem~\ref{thm:code} is then:
\begin{corollary} If a $(\!(n,K,\delta)\!)_2$ quantum code exists,
then the SDP in Eq.~\eqref{eq:sdpgamma} is feasible.
\end{corollary}

\subsection{Lov\'asz bound for self-dual codes}
Recall that a quantum code with $K=1$ is called self-dual. By definition, such code is pure and we can apply the conditions of the previous section.
\begin{definition}[Confusability graph for self-dual codes]\label{def:pure_conf_graph}
Let $G$ be a graph whose vertices correspond to elements of the $n$-qubit Pauli basis~$\EE_n$.
Let two vertices $a$ and $b$ be connected, written $a \sim b$,
if $E_a,E_b \in \EE_n$ satisfy
$E_a E^\dag_b +  E_b E^\dag_a = 0$ or $0<\wt(E_a E^\dag_b)<\delta$;
and let $a \sim a $ if $0<\wt(E_a)<\delta$.
\end{definition}
This simply incorporates the two constraints of Eq.~\eqref{eq:sim_anti_comm} and \eqref{eq:extra_zeros}.
To see this, note that if a positive semidefinite matrix has an element in the diagonal equal to zero,
then all elements in the same row and column are zero, too.
With this notation, we relax the SDP from Eq.~\eqref{eq:sdpgamma}
[that include the pure code conditions in Eq.~\eqref{eq:pure_code_conditions_Gamma}]
for self-dual codes to:
\begin{align}\label{eq:sdpgamma_pure}
	\text{maximize}		\quad & \sum^N_{\aaa=1}\Gamma_{\aaa\aaa}\,, \nn \\
	\text{subject to} 	\quad & \Gamma_{00}=1 \,, \nn \\
	& \Gamma_{\aaa0}=\Gamma_{\aaa\aaa} \,, \nn \\
	& \Gamma \succeq 0 \,, \nn \\
	& \Gamma_{\aaa\bbb}=0 \quad  \text{if} \quad \aaa \sim \bbb \,.
\end{align}
Because of $a\sim a$ if and only if $a\sim 0$ the last constraint can be
restricted to $a,b \neq 0$ without changing the objective value.
Thus the SDP in Eq.~\eqref{eq:sdpgamma_pure}
is nothing other than the semidefinite program for the Lov\'asz theta number [Eq.~\eqref{eq:lovasz_SDP2}]
of the confusability graph $G$ with the vertex corresponding to the identity removed.
To see this, note that $\Gamma$ is indexed by $0,\dots, N$, while the sum in Eq.~\eqref{eq:sdpgamma_pure} starts at $1$,
corresponding to the sum over $N$ elements in Eq.~\eqref{eq:lovasz_SDP2}.
For a self-dual code to exist,
the constraint $\sum^{N}_{\aaa=0}\Gamma_{\aaa\aaa}={2^n/K}$ in Eq.~\eqref{eq:sdpgamma} must be met.
Thus the objective value of Eq.~\eqref{eq:sdpgamma_pure} must be larger or equal than $2^n-1$.
As a consequence, we have the following corollary:
\begin{corollary}\label{cor:lovasz}
	If a pure $(\!(n,1,\delta)\!)_2$ code exists, then $2^n \leq 1 + \vartheta(G')$,
	where $G'$ is the confusability graph $G$ [Definition~\ref{def:pure_conf_graph}]
	of a self-dual quantum code  with the identity vertex removed.
\end{corollary}
\begin{example}
For the confusability graph $G'$ of a $(\!(4,1,3)\!)_2$ code, the Lov\'asz theta number is $\vartheta(G')= 7$.
By Corollary~\ref{cor:lovasz}, $1+\vartheta(G') < 2^4$ implies that a $(\!(4,1,3)\!)_2$ code does not exist.
While this result was already known from the linear programming bounds enhanced with the shadow inequalities~[Eq.~\eqref{eq:lp}]~\footnote{
The weights $[1,0,0,12,3]$ are unique in being non-negative, invariant under the quantum MacWillams transformation, and satisfying the pure code constraint. However, they violate the shadow inequalities.},
we think that using Corollary~\ref{cor:lovasz} provides a conceptually simpler approach.
\end{example}

Note that $\Gamma$ from Eq.~\eqref{eq:sdpgamma_pure} has zeros
in all rows and columns with $0<\wt(E_\aaa)<\delta$.
The positivity of $\Gamma$ is then equivalent to the positivity of the same matrix without such rows and columns:
\begin{remark}\label{rmk:Gprime}
Let $G''$ be the subgraph of the confusability graph $G$ from Definition~\ref{def:pure_conf_graph},
whose vertices $\one$ and $E_\aaa\in \EE_n$ such that $\wt(E_\aaa)<\delta$ removed.
Then, $\vartheta(G'')=\vartheta(G')$.
\end{remark}

\subsection{Self-dual stabilizer codes}
	One can give the following argument on the appearance of the Lov\'asz theta number in the case of self-dual stabilizer codes.
	Recall that a stabilizer code corresponds to an abelian subgroup of the $n$-qubit Pauli group for which $-\one \not \in S$.
	A pure stabilizer code additionally satisfies $\wt(s_i)\geq \delta$ for all $s_i\in S \setminus \one$.

	Recall that $\vartheta(G)$ upper bounds the independence number $\a(G)$, i.e.,
	the maximum size of a set with pairwise disconnected vertices.
	Let $G'$ be the confusability graph of a self-dual code with the identity vertex removed.
	Then two vertices $a,b\in V$ are disconnected if and only if they satisfy all the following conditions,
	\begin{equation}
	E_\aaa^\dag E_\bbb -  E_\bbb E_\aaa^\dag=0 \quad \text{and}
	\quad \wt(E_\aaa^\dag), \wt(E_\bbb), \wt(E_\aaa^\dag E_\bbb)\geq \delta\,.
	\end{equation}
	Given $s_\aaa, s_\bbb \in S$, we also have $s_c=s_\aaa s_\bbb \in S$.
	Therefore, if $E_c=\pm E_\aaa^\dag E_\bbb$, then $a,b,c\in V$ do not share any edges.
	Denote by $H\subseteq V$ a subset of disconnected vertices of the graph $G'$ with maximum size $\a(G')$.
	Any $a,b\in H$ then satisfies
	\begin{equation}\label{eq:max_set_stab}
	  E_\aaa^\dag E_\bbb -  E_\bbb E_\aaa^\dag=0, \quad \wt(E_\aaa^\dag), \wt(E_\bbb)\geq \delta\,, \quad  \text{and} \quad  c \in H  \quad \text{if} \quad  E_c = \pm E_\aaa^\dag E_\bbb.
	\end{equation}

	The above equation allows us to state the following:
	given a self-dual stabilizer code $S$, one can construct a disconnected subset $H_S\subseteq V$ with size $2^n-1$ by identifying each vertex $a\in V$ with a Pauli error $E_a$ such that $E_a=\pm s_a\in S\setminus \one$.
	Conversely, given $H$, one can construct a subset of $\EE_n$ by identifying each vertex $a\in V$ with $E_a$. Using Eq.~\eqref{eq:max_set_stab}, it can be seen that all elements of this set has weight equal or larger than $\delta$ and up to real phases, it forms an abelian group if one includes $+\one$.
	But this is nothing else than a stabilizer group.
	Note that this stabilizer group can have at most $2^n-1$ nontrivial elements
	since the maximum size of a commuting Pauli subgroup is $2^n$.
	Therefore, a self-dual stabilizer code exists if and only if $\alpha(G')=2^n-1$.
	When $\alpha(G') < 2^n-1$, it is not possible to construct a stabilizer group with $2^n$ elements and therefore, a self-dual stabilizer code does not exist.
	The same reasoning applies to the subgraph $G''$ of Remark~\ref{rmk:Gprime}.
	
\section{Symmetry-reduction through the Terwilliger algebra}

We use the technique developed by Schrijver~\cite{1468304} and Gijswijt et al.~\cite{GIJSWIJT20061719}
for symmetry-reducing semidefinite programs that bound the parameters of binary and nonbinary codes respectively.
This is done by averaging the semidefinite program in Eq.~\eqref{eq:sdpgamma}
through operations that leave the distances between triples of Pauli strings invariant.
Then it can be block-diagonalized with the Terwilliger algebra of the Hamming association scheme.
As we will eventually deal with an alphabet comprised of the four Pauli matrices $\one, X, Y, Z$,
we are concerned with quaternary codes and thus follow the exposition of Ref.~\cite{GIJSWIJT20061719}.
We first present the general case, before explaining the averaging and symmetry-reduction of the moment matrix $\Gamma$ from Eq.~\eqref{eq:relax}.

\subsection{Automorphism groups}
Let $\by{q}=\{0,\dots,q-1\}$ be an alphabet and
$\by{q}^n$ be the set of strings of length $n$ formed from elements of $\by{q}$.
Denote the Hamming distance between two elements $\x,\y \in \by{q}^n$ by
$d(\x,\y)$ the number of positions in which $x$ and $y$ differ.
Then consider the set of automorphism $\Aut(q,n)$
of $\by{q}^n$ that preserve the Hamming distance.
It is known that $\Aut(q,n)$ consists of the set of permutations
that permute the $n$ coordinates
and that independently permute the alphabet at each coordinate~\cite{GIJSWIJT20061719}.

Let $\Aut_{0}(q,n)$ be the subset of $\Aut(q,n)$
that leaves the all-zero string invariant.
This consists of the set of permutations
that permute the $n$ coordinates
and independently permute the $q-1$ non-zero elements at each coordinate~\cite{GIJSWIJT20061719}.
The group $\Aut_{0}(q,n)$ is isomorphic to the wreath product $S_{q-1} \wr S_n$.

The equivalence class of a classical code is generated by $\Aut(q,n)$.
In turn, the equivalence class of a quantum code under local Clifford operations and permutations is generated by $\Aut_{0}(q,n)$~\footnote{Usually equivalence under local unitary operations and permutations are considered, but due to our indexing of the moment matrix, the subgroup of local Clifford operations and permutations suffices.}.
The reason is that the projector onto the code space has to remain positive-semidefinite under the equivalence operation,
and consequentially the identity matrix has to be a fixed point.

Let $R$ be a real matrix indexed by $\x,\y\in \boldsymbol{q}^n$.
Given an automorphism $\pi \in \Aut(q,n)$,
let it act on $R$ by permuting its rows and columns accordingly,
\begin{equation}
 \pi(R)_{\x\y} = R_{\pi^{-1}(\x)\pi^{-1}(\y)}\,.
\end{equation}
Note that if $R$ is positive semidefinite, then $\pi(R)$ is also positive semidefinite for any~$\pi$.
Matrices invariant under the action of $\Aut_0(q,n)$ are in the Terwilliger algebra,
while the commutative subalgebra invariant under $\Aut(q,n)$ is the Bose-Mesner algebra~\cite{GIJSWIJT20061719}.

\subsection{Nonbinary Terwilliger algebra}
Let $M_{i,j}^{t,p}$ a $q^n\times q^n$ matrix indexed by $\qn$ and defined by
\begin{equation}\label{eq:terw_basis}
	(M_{i,j}^{t,p})_{\x\y} =\begin{cases}
		1 & \text{ if } \quad
		|s(\x)|=i\,, \quad |s(\y)|=j\,, \quad \lvert s(\x)\cap s(\y)\rvert=t\,,\\
		& \quad\quad\,\, \lvert\{m \,:\, \x_{m}=\y_{m}
		\neq0\}\rvert=p \,,
		\vspace{0.2cm} \\
		0 & \text{ otherwise} \,,
	\end{cases}
\end{equation}
where $s(\x) = \{m \,|\, \x_m \neq 0 \}$ is the support of $\x$;
$|s(\x)|$ is its size (that is the number of non-zero coordinates of $x$);
and $\x_m$ is the element at coordinate $m$ of $\x$.

The range of tuples,
outside of which
$M_{i,j}^{t,p}$ necessarily vanishes, is
\begin{equation}\label{eq:I-range}
 \II(q,n)= \{ (i, j, t, p)\,:\, 0 \leq p \leq t \leq i,j \quad \text{and} \quad i+j \leq t+n\, \}\,.
\end{equation}
Ref.~\cite[Proposition 9]{GIJSWIJT20061719} shows that $M_{i,j}^{t,p}$ are pairwise orthogonal and $\left\langle M_{i,j}^{t,p},M_{i,j}^{t,p}\right\rangle=\gamma_{i,j}^{t,p}$,
where
\begin{align}\label{eq:gamma}
	\gamma_{i,j}^{t,p}
	&= (q-1)^{i+j-t}(q-2)^{t-p}\binom{n}{p,t-p,i-t,j-t}\,,
\end{align}
and
\begin{align}
\binom{n}{a_1,\dots,a_r}= \frac{n!}{a_1!\dots a_r!(n-\sum^r_{\ell=1} a_\ell)!}\,
\end{align}
is a multinomial coefficient.
The set of $M_{i,j}^{t,p}$ forms a basis for the Terwilliger algebra
$\mathcal{A}_{q,\,n}$ of the $q$-ary Hamming cube.
The Terwilliger algebra can be block-diagonalized~\cite{GIJSWIJT20061719},
allowing us to check the positivity of the blocks of an element in $\mathcal{A}_{q,\,n}$.

\begin{remark}\label{rmk:interp}
Let us give some motivation as to how the indexing with $(i,j,t,p)$ in Eq.~\eqref{eq:terw_basis}
will be used later.
Consider the three terms $E^\dag$, $F$, and $E F^\dag$
that appear in a single entry of the matrix $\Gamma$ [Eq.~\eqref{eq:relax}].
Now suppose that
$\wt(E) = i$ and $\wt(F) = j$.
Then, set $t = |s(E) \cap s(F)|$ and
$p = |\{m\,:\, (E)_m = (F)_m \neq \one \} $
the number of tensor factors (coordinates)
where both $E$ and $F$ have identical nontrivial Pauli elements.
Then the overlap between the supports of $E$ and $F$ is $t$,
and the number of places a nontrivial Pauli in $E$ cancels a nontrivial Pauli in $F$ is $p$.
The quantity $t-p$ is then the number of places that a nontrivial Pauli in $E$ does {\em not} cancel a nontrivial Pauli in $F$.
Note that if $E$ and $F$ commutes then $t-p$ is even and if $E$ and $F$ anticommutes then $t-p$ is odd.
Finally, $\wt(E F^\dagger) = i+j-t-p$.
\end{remark}

\subsection{Averaging over codewords}
\label{subsec:average_codewords}
We now consider different averaging methods, whose results lie in the Terwilliger algebra.
Recall that an automorphism $\pi \in \Aut(q,n)$
acts on a matrix $R$ by permuting its rows and columns,
$ \pi(R)_{\x\y} = R_{\pi^{-1}(\x)\pi^{-1}(\y)}$.
Now for $v \in \qn$, let $\sigma_{v}$ be a distance-preserving automorphism that additionally
satisfies $\sigma_v(v) = 0$.
With this, define $R^{v}$ as
\begin{equation}\label{eq:Rv}
	R^{v}=\frac{1}{\lvert\Aut_{0}(q,n)\rvert}
	\sum\limits_{\pi\in\Aut_{0}(q,n)}
	\pi\left(\sigma_{v}\!\left(R\right)\right)\,.
\end{equation}
The effect of Eq.~\eqref{eq:Rv} is to average $R$ over all distance-preserving automorphisms
that map $v$ to the zero string.

Let $C$ be a code with elements from $\qn$.
Then one can additionally average over the set of matrices $R^v$ where $v$ is a codeword:
\begin{equation}\label{eq:Rtilde}
	\tilde R(C) =\frac{1}{\left\lvert C \right\rvert}\sum\limits_{v \in C}R^{v}.
\end{equation}
The effect of Eq.~\eqref{eq:Rtilde} is to average $R$
over all distance-preserving automorphism that map codewords to the zero string.
Note that both $R^v$ and $\tilde R$ are invariant under
$\Aut_{0}(q,n)$.

In particular, this means we can express $\tilde R$ in a basis of the Terwilliger algebra $\mathcal{A}_{q,n}$,
\begin{equation}\label{eq:R_expansion}
	\tilde R=
	\sum\limits_{\left(i,j,t,p\right)\in \II(q,n)}x_{i,j}^{t,p}M_{i,j}^{t,p}\,,\quad\quad x_{i,j}^{t,p} \in \R\,.
\end{equation}

We now show how to write
$x_{i,j}^{t,p}$ in terms of elements from the matrix $R$.
Define $\lambda^{t,p}_{i,j} = \langle\, \tilde R,M_{i,j}^{t,p}\,\rangle$ and note that using Eq.~\eqref{eq:R_expansion}, we get that $\lambda^{t,p}_{i,j} =\gamma^{t,p}_{i,j} x^{t,p}_{i,j}$. We then develop $\lambda^{t,p}_{i,j}$ using the definition of $\tilde R$ in Eq.~\eqref{eq:Rtilde},
\begin{align}\label{eq:lambda_def_start}
	\lambda^{t,p}_{i,j}
	&=\frac{1}{\left\lvert C\right\rvert} \frac{1}{\lvert\Aut_{0}(q,n)\rvert}
	\sum\limits_{\pi\in\Aut_{0}(q,n)}
	\sum\limits_{v \in C}\Big\langle 	\pi\left(\sigma_{v}\!\left(R\right)\right),M_{i,j}^{t,p}\Big\rangle \nn\\
	% %
	&=\frac{1}{\left\lvert C\right\rvert} \left\langle 	 \sum\limits_{v \in C}\sigma_{v}\!\left(R\right),\frac{1}{\lvert\Aut_{0}(q,n)\rvert}\sum\limits_{\pi\in\Aut_{0}(q,n)}\pi^{-1}\left(M_{i,j}^{t,p}\right)\right\rangle \nn \\
	&=\frac{1}{\left\lvert C\right\rvert} \left\langle 	 \sum\limits_{v \in C}\sigma_{v}\!\left(R\right),M_{i,j}^{t,p}\right\rangle\,.
\end{align}
Here we used the fact that $M^{t,p}_{i,j}$ is invariant under the action of elements from $\Aut_0(q,n)$.
From the definition of $M^{t,p}_{i,j}$ [see Eq.~\eqref{eq:terw_basis}], one obtains
\begin{align}\label{eq:lambda_def}
	\lambda^{t,p}_{i,j}= \frac{1}{\left\lvert C\right\rvert} \!\!\!\! \sum\limits_{\substack{ \x,\y \in \by{q}^n\\ |s(\x)| = i
			\\ |s(\y)| = j
			\\ \lvert s(\x)\cap s(\y)\rvert = t
			\\ \lvert\{m\mid \x_{m}=\y_{m}\neq0\}\rvert = p \\}} \!\!\!\! \sum\limits_{v\in C} \sigma_{v}(R)_{\x\y} \,,
\end{align}
since the only elements $\sigma_{v}(R)_{\x \y}(M_{i,j}^{t,p})_{\x \y}$
not equal to zero are the ones satisfying
$|s(\x)|=i$, $|s(\x)|=j$, $\lvert s(\x)\cap s(\y)\rvert=t$ and $\lvert\{m \mid \x_{m}= \y_{m} \neq0\}\rvert=p$.

\begin{remark}
	Ref.~\cite[Proposition 9]{GIJSWIJT20061719} defines $\langle\, \tilde R,M_{i,j}^{t,p}\,\rangle =\frac{1}{|C|} \lambda^{t,p}_{i,j}$.
	However, we decide to define $\langle\, \tilde  R,M_{i,j}^{t,p}\,\rangle = \lambda^{t,p}_{i,j}$ since in the quantum case, $\tilde \Gamma$ is averaged over
	$\Aut_0(q,n)$ instead of averaging over both $\Aut_0(q,n)$ and codewords
	(see Proposition~ \ref{prop:terwilliger_main_reduction}).
	 \end{remark}
The matrix $\tilde R$ can be block-diagonalized. For this define the following coefficients~\cite[Eq.~(20) and~(27)]{GIJSWIJT20061719}:
\begin{align}\label{eq:alpha}
	\alpha(i,j,t,p,a,k)&= \beta_{i-a,j-a,k-a}^{n-a,t-a}\left(q-1\right)^{\frac{1}{2}\left(i+j\right)-t} \sum\limits_{g=0}^{p}\left(-1\right)^{a-g}\binom{a}{g}\binom{t-a}{p-g}\left(q-2\right)^{t-a-p+g}, \nn \\
	\beta_{i,j,k}^{m,t}&=\sum\limits_{u=0}^{m}\left(-1\right)^{t-u}\binom{u}{t}\binom{m-2k}{m-k-u}\binom{m-k-u}{i-u}\binom{m-k-u}{j-u}\,.
\end{align}
It can be shown that $\tilde R$ is positive semidefinite
if and only if
\begin{equation} \label{eq:psd_tildeR}
	\bigoplus_{\substack{a,k \in \N_0\\ 0\leq a\leq k\leq n+a-k}}
	\left(\sum\limits_{\substack{t,p\in \N_0 \\ 0 \leq p \leq t \leq i,j \\ i+j\leq t+n}} \alpha(i,j,t,p,a,k)x_{i,j}^{t,p}\right)_{i,j=k}^{n+a-k} \,.
\end{equation}
is positive semidefinite~\cite[Eq.~(43)]{GIJSWIJT20061719}.

\subsection{Averaging over the complement of a code}
\label{subsec:average_codewords_comp}

Similarly, one can average over the elements that are not in the code and obtain
\begin{equation}\label{eq:Rprime}
	\tilde R^c(C)=\frac{1}{q^n - \left\lvert C \right\rvert}\sum\limits_{v \in \by{q}^n \setminus \text{C}}R^{v}.
\end{equation}
The average over all $\sigma \in \Aut(q,n)$ can then be expressed as
\begin{align}\label{eq:Rbar}
	\overbar R= \frac{1}{q^n} \sum_{v\in \by{q}^n} R^v = \frac{1}{q^n} \left( |C|\tilde R  + (q^n-|C|)\tilde R^c \right) \,.
\end{align}
Since $\overbar R$ is invariant under all permutations $\sigma \in \Aut(q,n)$,
the matrix $\overbar{R}$ is an element of the Bose–Mesner algebra~\cite[Eq.~(40)]{GIJSWIJT20061719}.
Therefore, it can be written as
\begin{align}\label{eq:Rbar_BoseMesner}
	\overbar{R}= \sum^n_{k=0} y_k N_k \,,
\end{align}
where $y_k \in \R$ and $\{N_k\}_{k=0}^n$ is a basis for the Bose-Mesner algebra given by
\begin{equation}
 (N_k)_{\x\y} = \begin{cases}
               1 \quad \text{if }\,\, d(\x,\y) = k\,, \\
               0 \quad \text{otherwise}\,.
              \end{cases}
\end{equation}
Here $d(\x,\y)$ is the Hamming distance between $\x$ and $\y$.
Note that that the Terwilliger algebra $\mathcal{A}_{q,n}$ contains the Bose-Mesner algebra as a commutative subalgebra, with
\begin{align}\label{eq:bose_bases}
	N_k=\sum_{\substack{(i,j,t,p)\in \II(q,n)\\i+j-t-p=k}} M^{t,p}_{i,j}\,.
\end{align}

Suppose that $\tilde R^c_{\x0}=0$ for all $\x$.
This happens for example when $R=\chi \chi^T$ where $\chi$ is the incidence vector of a classical code.
Ref.~\cite[Proposition 8]{GIJSWIJT20061719} shows that
\begin{align}\label{eq:rtilde_c}
	\tilde R^c &=
	\frac{1}{q^n-|C|}(q^n \bar R - |C|\tilde R) \nn \\
	&=\frac{|C|}{q^n-|C|}\sum_{(i,j,t,p) \in \II(q,n)} (x^{0,0}_{i+j-t-p,0}-x^{t,p}_{i,j}) M^{t,p}_{i,j} \,.
\end{align}
Then, in analogy to Eq.~\eqref{eq:psd_tildeR},
$\tilde R^c \succeq 0$ is equivalent to
\begin{equation}\label{eq:psd_tildeR_c}
	\bigoplus_{\substack{a,k \in \N_0\\ 0\leq a\leq k\leq n+a-k}} \left(\sum\limits_{\substack{t,p\in \N_0 \\ 0 \leq p \leq t \leq i,j \\ i+j\leq t+n}}\alpha(i,j,t,p,a,k)(x_{i+j-t-p,0}^{0,0}-x_{i,j}^{t,p})\right)_{i,j=k}^{n+a-k} \succeq 0\,.
\end{equation}

When $R=\chi \chi^T$,
it follows from the diagonalization of $\overbar R$ that the $B_i$ enumerators are non-negative.
This, together with the non-negativity of the $A_i$ enumerator, gives the Delsarte bound~[see Eq.~\eqref{eq:lp_Delsarte_class}]
~\cite{10.1007/978-94-010-1826-5_7,GIJSWIJT20061719}.

\subsection{Symmetry reduction of $\Gamma$}\label{ref:sym_red}

A natural mapping identifying the matrices $I,Z,X,Y$ with $0,1,\a,\a^2 \in \GF(4)$ respectively, maps,
up to a phase, the multiplication of Pauli matrices to addition over $\GF(4)$~\cite{681315}.
That is, $E_{x+y} = \omega_{xy} E_x E_y$ with $\omega_{xy} \in \{\pm i, \pm 1\}$, and where
$E_{-x} = E_x^\dag$.
Therefore, we map $\EE_n$ to $\GF(4)^n$.
In particular, $E_0  = I^{\otimes n}$, where we write $E_0$ for $E_{0^n}$.
Indexing the matrix $\Gamma$ from Eq.~\eqref{eq:relax} with all elements from $\GF(4)^n$
(and eventually ${\by q}^n$ with $q=4$),
our aim is now to follow the averaging methods from
Sections~\ref{subsec:average_codewords} and \ref{subsec:average_codewords_comp}.

For any $v \in \GF(4)^n$,
define $\sigma_{v}$ to be a distance-preserving automorphism on $\GF(4)^n$.
It acts on $\Gamma$ as
\begin{align}\label{eq:action_sigma_v_on_code}
	\sigma_{v}\!\left(\Gamma\right)_{\x,\y} &=\Gamma_{{\x-v },{\y-v}} \nn\\
	&= \av{E_{\x-v}^\dagger} \av{E_{\y-v}} \av{E_{\x-v} E^\dagger_{\y-v}}  \nn\\
	&= 	\av{(E_{\x}  E_{v}^\dagger)^\dagger} \av{E_{\y}E_{v}^\dagger} \av{E_{\x} E_v^\dagger E_v E_{\y}^\dagger} \nn\\
	&= \av{E_{v} E_{\x}^\dagger} \av{E_{\y}E_{v}^\dagger} \av{E_{\x} E_{\y}^\dagger} \,,
\end{align}
where the third equation follows from $\bar{\omega}_{x(-v)} \omega_{y(-v)} \omega_{x(-v)} \bar{\omega}_{y(-v)} = 1$.
This maps the row and column indexed by $v$ to the one indexed by $0^n$ respectively.
%//
We can then average $\Gamma$ over a subset of distance-preserving automorphisms on $\GF(4)^n$
such that it can be expanded in the Terwilliger algebra.

Note that the averaging of $\Gamma$
over distance-preserving automorphisms of $\GF(4)^n$ has the same effect as
averaging over distance-preserving automorphisms of
$\by{q}^n$ with $q=4$.
This is because the weight of $E_x^\dag E_y$
(corresponding to the distance between $x$ and $y$)
is not affected by any group structure of the underlying alphabet.
More formally, the weight is characterized by $\wt(E_x E_y^\dag) = i + j - t - p$ with $i,j,t,p$ defined in Remark~\ref{rmk:interp},
which only depends on the alphabet and not the group structure.
Thus, we can map $\GF(4)^n$ to $\by{4}^n$.

With some abuse of notation, we define in analogy to Eq.~\eqref{eq:Rv}
the matrix $\Gamma^{v}$ as
\begin{equation}\label{eq:Gv}
	\Gamma^{v}=\frac{1}{\lvert\Aut_{0}(4,n)\rvert}\sum\limits_{\pi\in\Aut_{0}(4,n)}\pi\left(\sigma_{v}\!\left(\Gamma\right)\right)\,.
\end{equation}
Similar to Eq.~\eqref{eq:Rtilde}, we also define
\begin{equation}\label{eq:Gtilde}
	\tilde \Gamma(Q) =\frac{1}{|Q|}\sum\limits_{v \in Q}\Gamma^{v}\,,
\end{equation}
where $Q$ is a subset of $\by{4}^n$. The effect of Eq.~\eqref{eq:Gtilde} is to average $\Gamma$
over all distance-preserving automorphisms that map an element of $Q$ to $0^n$.
Likewise, in analogy to Eq.~\eqref{eq:Rbar}
define the averaging over all elements from $\Aut(4,n)$ as
\begin{align}\label{eq:gamma_bar}
	\overbar\Gamma=& \frac{1}{4^n} \sum_{v\in \by{4}^n} \Gamma^v\,.
\end{align}

We now consider averaging over different subsets $Q$. We first take $Q=\by{4}^n$.
This is equivalent to averaging $\Gamma$ over all elements from $\Aut(4,n)$, which will lead to the quantum Delsarte bound.

\section{Quantum Delsarte bound}
\label{sec:Delsarte}

For classical codes,
the symmetrization over $\Aut(q,n)$ of the semidefinite program SDP1 for the Lov\'asz theta number [Eq.~\eqref{eq:lovasz_SDP1}], along
with the non-negativity of the matrix entries as an extra constraint, leads to the Delsarte bound
\cite{bachoc2010applicationssemidefiniteprogrammingcoding}.
The matrix used for the Lov\'asz theta bound is defined as $R=\chi \chi^T$, where $\chi$ is the indexing vector of a classical code.
Our aim is to recover a similar result for the Lov\'asz theta number for self-dual quantum codes.
For simplicity in the presentation, we here use the semidefinite program SDP2 [see Eq.~\eqref{eq:lovasz_SDP2}].
A key difference to the classical case is that the entries of the semidefinite variable from the quantum Lov\'asz theta bound
[see Corollary~\ref{cor:lovasz}] can be negative.
Due to this fact, we require a different extra constraint to the quantum Lov\'asz theta bound to recover the quantum Delsarte bound [see Eq.~\eqref{eq:lp_Delsarte}] after the symmetrization.

Suppose $\Gamma$ arises from a state.
Our aim is to show that the $A_j(\rho)$ enumerators emerge as elements from $\overbar\Gamma$ expressed in the Bose-Mesner algebra basis $\{N_k\}_{k=0}^n$ [see Eq.~\eqref{eq:bose_bases}].
Diagonalizing $\overbar \Gamma$, one then obtains the dual enumerator $B_j(\rho)$ by means of the quantum MacWilliams identity.
As in Eq.~\eqref{eq:Rbar_BoseMesner}, the matrix $\overbar \Gamma$ is part of the Bose-Mesner algebra and can be expanded
in terms of the basis $\{M^{t,p}_{i,j}\}_{(i,j,t,p) \in \II(4,n)}$ of the Terwilliger algebra,
\begin{align}\label{eq:gamma_bar_bose}
	\overbar{\Gamma}= \sum^n_{k=0} y_k N_k
	=  \quad \sum^n_{k=0} y_k  \sum_{\substack{(i,j,t,p)\in \II(4,n)\\k=i+j-t-p}} M^{t,p}_{i,j}
	\quad = \sum_{(i,j,t,p)\in \II(4,n)} y_{i+j-t-p} M^{t,p}_{i,j}\,,
\end{align}
where we have used the decomposition in Eq.~\eqref{eq:bose_bases}.
Given $\overbar \Gamma$, the goal is now to obtain $y_k$ in terms of elements from $\Gamma$.

Recall the matrices $M^{t,p}_{i,j}$ form an orthogonal basis, and that
$\left\langle M_{i,j}^{t,p},M_{i,j}^{t,p}\right\rangle=\gamma_{i,j}^{t,p}$,
where $\gamma_{i,j}^{t,p}$ is given by Eq.~\eqref{eq:gamma}.
Since all $M_{i,j}^{t,p}$ with $k=i+j-t-p$ have the same coefficient $y_k$,
it is enough to compute the inner product with a representative, e.g., $M_{k,0}^{0,0}$,
\begin{align}\label{eq:yk}
y_k
&=
\frac{1}{\gamma^{0,0}_{k,0}}\av{\overbar\Gamma, M^{0,0}_{k,0}}=\frac{1}{\gamma^{t,p}_{i,j}}\av{\overbar\Gamma, M^{t,p}_{i,j}} \quad\quad \text{for all} \quad k=i+j-t-p \,.
\end{align}
To determine $y_k$ it is therefore sufficient to compute $\av{\overbar\Gamma, M^{0,0}_{k,0}}$.
Following the procedure of Eq.~\eqref{eq:lambda_def_start} and \eqref{eq:lambda_def} and using the definition of $\overbar\Gamma$ in Eq.~\eqref{eq:gamma_bar}, one obtains
\begin{align}\label{eq:basis_gammaprimeprime}
	\av{\overbar\Gamma, M^{0,0}_{k,0}}
	&=\frac{1}{4^n\lvert\Aut_{0}(4,n)\rvert}
	\sum\limits_{\pi\in\Aut_{0}(4,n)}
	\sum\limits_{v \in \by{4}^n}\left\langle 	\pi\left(\sigma_{v}\!\left(\Gamma\right)\right),M_{k,0}^{0,0}\right\rangle \nn\\
	&= \left\langle \frac{1}{4^n}\sum\limits_{v \in \by{4}^n}\sigma_{v}\!\left(\Gamma\right),\frac{1}{\lvert\Aut_{0}(4,n)\rvert}\sum\limits_{\pi\in\Aut_{0}(4,n)}\pi^{-1}\left(M_{k,0}^{0,0}\right)\right\rangle \nn \\
	&= \left\langle 	 \frac{1}{4^n}\sum\limits_{v \in \by{4}^n}\sigma_{v}\!\left(\Gamma\right),M_{k,0}^{0,0}\right\rangle \nn \\
	& = \frac{1}{4^n}\sum\limits_{\substack{x\in \by{4}^n \\ \wt(\x)=k}}  \sum\limits_{v\in \by{4}^n}  \sigma_{v}(\Gamma)_{\x 0} \,.
\end{align}

We now show that $y_k$ is proportional to the $A_k(\rho)$ enumerator.
Recall the action of $\sigma_v$ on $\Gamma$ in Eq.~\eqref{eq:action_sigma_v_on_code}.
To this end, let us expand
for a quantum code $\varrho = \Pi /K$
the inner sum of Eq.~\eqref{eq:basis_gammaprimeprime},
\begin{align}\label{eq:av_0b}
\sum_{{v}\in \by{4}^n}\sigma_{v}\!\left(\Gamma\right)_{\x0} &= \sum_{E_{v}\in \EE_n} \Big( \tr(E_vE^\dagger_\x\varrho) \tr(E^\dagger_v\varrho) \Big)\tr( E_\x\varrho) \nn \\
&=  \tr\Big(\sum_{E_{v}\in \EE_n} (E_vE^\dagger_\x  \ot E^\dagger_v)  (\varrho \ot \varrho) \Big) \tr( E_\x \varrho) \nn \\
&= 2^n\tr\Big( (1,2) (E^\dagger_\x\varrho \ot \varrho) \Big) \tr( E_\x \varrho)   \nn\\
&= 2^n\tr(E^\dagger_\x \varrho^2) \tr(E_\x \varrho) \nn\\
&= \frac{2^n}{K}\tr(E^\dagger_\x \varrho) \tr(E_\x \varrho) =\frac{2^n}{K} \Gamma_{\x0}\,.
\end{align}
Here we used that $\varrho^2 = \varrho/K$ and $\tr\big((1,2) \varrho^{\ot 2}\big)= \tr(\varrho^2)$ where the swap operator is
$(1,2)=2^{-n}\sum_{E_v\in \EE_n} E_v \ot E^\dagger_v$.
Note that Eq.~\eqref{eq:av_0b} can also be seen from combining the constraints of Eqs.~\eqref{eq:structure_constraint}
and ~\eqref{eq:projconst_lvl} in the main semidefinite programming relaxation.

With this we can write Eq.~\eqref{eq:yk} as
\begin{align}\label{eq:y_k_A}
	y_k &=
	\frac{1}{4^n\gamma^{0,0}_{k,0}}
	\av{\overbar\Gamma, M^{0,0}_{k,0}} \nn\\
	&= \frac{1}{4^n\gamma^{0,0}_{k,0}} \sum\limits_{\substack{x\in \by{4}^n \\ \wt(\x)=k}}  \sum\limits_{v\in \by{4}^n}  \sigma_{v}(\Gamma)_{\x0} \nn\\
	&= \frac{1}{4^n\gamma^{0,0}_{k,0}}
	\frac{2^n}{K} \sum_{
	\substack{ E_\x \in \EE_n\\ \wt(E_\x)=k}}\Gamma_{\x0} \nn\\
 	&= (2^nK\gamma^{0,0}_{k,0})^{-1} A_k(\varrho) \,,
\end{align}
 since  $\sum_{x\in \by{4}^n, \,|s(x)|=k} \Gamma_{\x0} =A_k(\varrho)$ and we used Eq.~\eqref{eq:av_0b}.

We now diagonalize $\overbar\Gamma$.
Ref~\cite{10.1007/978-94-010-1826-5_7} shows that the eigenvalues of $N_k$ can be derived from the Krawtchouk polynomials $K_k(i;n)$ [see Eq.~\eqref{eq:krawpoly}]. Since $\{N_k\}^n_{k=0}$ is a commuting basis it can be simultaneous diagonalized, thus the eigenvalues of $\overbar\Gamma$ are $\sum^n_{k=0} y_k K_k(i;n)$.
The quaternary Krawtchouk polynomials satisfy the following relation~\cite[Theorem 17]{macwilliams1977theory},
\begin{align}
	3^i \binom{n}{i}K_k(i;n)=3^k \binom{n}{k}K_i(k;n) \,.
\end{align}
With Eq.~\eqref{eq:gamma}, one can express this as $\gamma^{0,0}_{i,0} K_k(i;n)= \gamma^{0,0}_{k,0}  K_i(k;n)$.
We then develop the eigenvalues of $\overbar\Gamma$ as follows
\begin{align}\label{eq:quantum_Delsarte_Bk}
	\sum^n_{k=0} y_k K_k(i;n) &= \frac{1}{\gamma^{0,0}_{i,0}}\sum^n_{k=0}   y_k \gamma^{0,0}_{k,0}K_i(k;n)  \nn\\
 	&=(2^nK\gamma^{0,0}_{k,0})^{-1} \sum^n_{k=0} K_i(k;n) A_k(\varrho) \,,
\end{align}
where we have used Eq.~\eqref{eq:y_k_A}.

Recall that  $\sum_{x\in \by{4}^n, \,|s(x)|=k} \Gamma_{\x0} =A_k(\varrho)$ for a state $\varrho = \Pi / K$.
By the quantum MacWilliams identity in Eq.~\eqref{eq:macexp},
it can be seen that Eq.~\eqref{eq:quantum_Delsarte_Bk}
is proportional to the dual enumerator $B_i(\varrho)$.
Thus both the $A_j$ and $B_j$ enumerators from averaging $\Gamma$ over $\Aut(q,n)$,
and diagonalizing the resulting $\overbar \Gamma$ in the basis of the Bose-Mesner algebra $\{N_k\}_{k=0}^n$.

However, we want to point out that Eq.~\eqref{eq:av_0b}
assumes that $\Gamma$ has been constructed from a state $\varrho = \Pi/K$, for which Eq.~\eqref{eq:av_0b} applies.
We impose this extra constraint on the Lov\'asz bound to obtain the quantum Delsarte bound:

\begin{theorem}\label{thm:Lovasz2Delsarte}
 For self-dual quantum codes, the averaging of the Lov\'asz bound from Corollary~\ref{cor:lovasz}
 over $\Aut(q,n)$ with the condition
 \begin{align}\label{eq:Delsarte_proof_extra_constraint}
	\sum_{v\in \by{4}^n} \sigma_{v}(\Gamma)_{x0} = 	 2^n \Gamma_{\x0} \quad \forall x \in \by{4}^n\,,
\end{align}
 yields the quantum Delsarte bound [Eq.~\eqref{eq:lp_Delsarte}].
\end{theorem}
\begin{proof}
Note that if $\Gamma$ is constructed from a code $\varrho$, then Eq.~\eqref{eq:av_0b} shows that $\sum_{v\in \by{4}^n} \sigma_{v}(\Gamma)_{x0}=2^n \Gamma_{\x0} $ for all $x\in \by{4}^n $.
By including $\sum_{v\in \by{4}^n} \sigma_{v}(\Gamma)_{x0} =  2^n \Gamma_{\x0}$ in the program for the Lov\'asz theta number for self-dual quantum codes from Corollary~\ref{cor:lovasz}, we obtain,
\begin{align}\label{eq:lovasz_modf}
	\text{max} 		\quad & \sum_{x \in \EE_n \backslash \one}\Gamma_{xx} \nn\\
	\text{subject to} 	\quad
	&  \Gamma \succeq 0, \quad \Gamma_{00}=1\,,\quad \Gamma_{x0}=\Gamma_{xx} \,, \nn \\
	&  \Gamma_{xy}=0 \quad  \text{if} \quad x \sim y\,, \nn\\
	&\sum_{v\in \by{4}^n} \sigma_{v}(\Gamma)_{x0} = 	 2^n \Gamma_{\x0} \,.
\end{align}
Note that $\Gamma_{xx}=\Gamma_{x0}\geq 0$ since $\Gamma \succeq 0$.
Now average $\Gamma$ from  Eq.~\eqref{eq:lovasz_modf} over $\Aut(4,n)$ to obtain $\overbar \Gamma=\sum^n_{k=0} y_k N_k$.
Using Eq.~\eqref{eq:yk},
Eq.~\eqref{eq:basis_gammaprimeprime},
and the constraint
Eq.~\eqref{eq:Delsarte_proof_extra_constraint}
write $y_k$ as
\begin{align}
	y_k &=
\frac{1}{\gamma^{0,0}_{k,0}}
\av{\overbar\Gamma, M^{0,0}_{k,0}} \nn\\
&= \frac{1}{4^n\gamma^{0,0}_{k,0}} \sum\limits_{\substack{x\in \by{4}^n \\ \wt(\x)=k}}  \sum\limits_{v\in \by{4}^n}  \sigma_{v}(\Gamma)_{\x0} \nn \\
&= (2^n\gamma^{0,0}_{k,0})^{-1} \sum\limits_{\substack{x\in \by{4}^n \\ \wt(\x)=k}} \Gamma_{\x0} \,.
\end{align}
Define $a_j = \sum_{x\in \by{4}^n \,, \wt(\x)=j} \Gamma_{\x0}$,
which corresponds to the $A_j$ enumerator.
Note that $\bar\Gamma$ only depends on $a_j$.
The $\{a_j\}^n_{j=0}$ satisfy:
\begin{enumerate}
\item $a_0=\Gamma_{00}=1$,
\item $a_j = \sum_{x\in \by{4}^n \,, \wt(\x)=j} \Gamma_{\x0} \geq 0$, since $\Gamma_{x0}\geq 0$ \,,
\item $a_j = 0$ if $0 \leq j < \delta $, since $\Gamma_{\x0}=0$ if $0<\wt(\x)<\delta$.
\item $\sum^n_{j=0} a_j = \sum_{x \in \by{4}^n}\Gamma_{\x\x}=\sum_{x \in \by{4}^n}\Gamma_{\x0}$.
\end{enumerate}
It remains to use the fact that $\Gamma \succeq 0$.
Recall that $\Gamma \succeq 0$ implies $\overbar \Gamma \succeq 0$.
By Eq.~\eqref{eq:quantum_Delsarte_Bk}, $\overbar \Gamma$ has as eigenvalues,
\begin{align}
	(2^n\gamma^{0,0}_{j,0})^{-1} \sum^n_{j=0}K_i(j;n) A_j \,.
\end{align}
Therefore, $\overbar \Gamma \succeq 0$ if and only if $\sum^n_{j=0}K_i(j;n) A_j \geq0$ for all $0 \leq i \leq n$.
The combination with (1) - (4) above yields the quantum Delsarte bound for self-dual codes [Eq.~\eqref{eq:lp_Delsarte}],
\begin{align}
	\eta = \text{max} 	\quad & \sum^n_{j=0} a_j \nn\\
	\text{subject to} \quad   & a_0 = 1\,, \nn\\
	&a_j \geq 0 \quad \text{with equality for } \,\, 1<j<\delta \,,\nn\\
	&\sum^n_{i=0} K_j(i;n)a_i \geq 0 \quad \text{for}  \quad 0 \leq j \leq n\,.
\end{align}
 If $\eta < 2^n$ then a code with parameters $(\!(n,1,\delta)\!)_2$ does not exist.
 This ends the proof.
\end{proof}

\begin{remark}
 Naturally, this bound can be made stronger by formulating it
 as a feasibility program that includes the constraint $\sum_{j=0}^n a_j = 2^n$.
 We use this for excluding the existence of the $(\!(7,1,4)\!)_2$ code in Section~\ref{sec:7qb}.
\end{remark}

\section{Symmetry-reduced semidefinite programming bounds}
\label{sec:sym_red_Terwilliger}

We now aim to symmetry-reduce $\Gamma$
while retaining as many constraints from the original program [Eq.~\eqref{eq:sdpgamma}] as possible.
This is achieved by averaging $\Gamma$ over $\Aut_0(4,n)$,
which corresponds to Eq.~\eqref{eq:Gtilde} with $Q=\{0^n\}$. We therefore define
\begin{align}\label{eq:gammatilde_any}
	\tilde\Gamma = \Gamma(0^n) = \frac{1}{\lvert\Aut_{0}(4,n)\rvert}\sum\limits_{\pi\in\Aut_{0}(4,n)}\pi\left(\Gamma\right)\,.
\end{align}
In analogy to Eq.~\eqref{eq:Rbar}, define $\tilde \Gamma^c$ as
\begin{align}\label{eq:gammatildec_any}
	\tilde\Gamma^c =\frac{1}{4^n-2^n/K}\left(4^n \bar \Gamma - \frac{2^n}{K}\tilde \Gamma\right)\,,
\end{align}
where $\overbar\Gamma$ is stated in  Eq.~\eqref{eq:gamma_bar}.
We choose $2^n/K$ in the normalization of $\tilde \Gamma$,
which for stabilizer codes corresponds to the number of stabilizer elements.
In general however,
the matrix  $\tilde\Gamma^c$ does not necessarily come from averaging $\Gamma$ over the codewords of a code,
i.e., $\tilde\Gamma^c \neq \frac{1}{4^n - \left\lvert Q \right\rvert}\sum_{v \notin Q}\Gamma^{v}$ for some $Q \subseteq \by{4}^n$.
It is thus not immediate that $\tilde \Gamma^c$ is positive semidefinite.
We now show that $\tilde\Gamma^c \succeq 0$ for all quantum codes.
\begin{proposition}\label{prop:pos_cond_av}
For every quantum code $\tilde\Gamma\succeq0$ and $\tilde \Gamma^c \succeq 0$.
\end{proposition}
\begin{proof}
The matrix $\tilde\Gamma$ is positive semidefinite by construction in Eq.~\eqref{eq:gammatilde_any},
being a sum over positive semidefinite matrices.

We now consider $\tilde \Gamma^c$.
We first show that a matrix with the entries
	\begin{align}\label{eq:pos_ent}
		L_{xy} =
		\big[\tr(E_\y  E^\dagger_\x \varrho) - \tr( E^\dagger_\x  \varrho)
		\tr(E_\y  \varrho) \big]
		\tr(E_\x  E^\dagger_\y \varrho) \,,
	\end{align}
	is positive semidefinite.
	To see this, write $L = M \circ G$ where $\circ$ is the Schur product, that is,
	$L_{xy} = M_{xy} G_{xy}$, and where
	\begin{align}
		M_{xy} &=\tr(E_\y  E^\dagger_\x \varrho) - \tr( E^\dagger_\x  \varrho)
		\tr(E_\y  \varrho) \,,\nn\\
		G_{xy} &= \tr(E_\x  E^\dagger_\y \varrho) \,.
	\end{align}
	The Schur product theorem states that if $A \succeq 0$ and $B\succeq 0$,
	then $A\circ B \succeq 0$.

	As $G$ has the form of a moment matrix it is clear that $G \succeq 0$.
	Let us show that also $M\succeq 0$. This can be seen by noting that $M$ is a covariance matrix,
	\begin{align}
		M_{xy} &= \big\langle \big( E_y - \av{E_y}_{\varrho} \big) \big( E^\dagger_x - \av{E^\dagger_x}_{\varrho} \big) \big\rangle_\varrho \nn \\
		&= \av{ E_y E^\dagger_x}_\varrho - \av{ E_y} \av{E^\dagger_x}_\varrho \,,
	\end{align}
	where we used that all observables are Hermitian.
	Thus $M\succeq 0$ and also $L = M \circ G\succeq 0$.

	We now show that
	\begin{align}\label{eq:L_average}
		(2^nK)\overbar\Gamma - \tilde\Gamma = \frac{1}{\lvert\Aut_{0}(4,n)\rvert}\sum_{\pi \in \Aut_0(4,n)} \pi(L)\,.
	\end{align}
	The expression in Eq.~\eqref{eq:L_average}, being a sum over positive semidefinite matrices,
	is positive semidefinite and thus proves that $\tilde \Gamma^c \succeq 0$.
	To see that Eq.~\eqref{eq:L_average} holds,
	write $\overbar\Gamma$ using Eq.~\eqref{eq:gamma_bar} together with Eq.~\eqref{eq:Gv} as
	\begin{align}\label{eq:proof1}
		\overbar\Gamma&= \frac{1}{4^n\lvert\Aut_{0}(4,n)\rvert}
		\sum_{v\in \by{4}^n}
		\sum\limits_{\pi\in\Aut_{0}(4,n)}
		\pi\left(\sigma_{v}\!\left(\Gamma\right)\right)\,\nn \\
		&=\frac{1}{4^n\lvert\Aut_{0}(4,n)\rvert}
		\sum\limits_{\pi\in\Aut_{0}(4,n)}\pi
		\left(\sum_{v\in \by{4}^n}\sigma_{v}\!\left(\Gamma\right)\right)\,.
	\end{align}
	Because $\overbar\Gamma$ arises from a quantum code $\varrho = \Pi/K$,
	we can develop the inner sum of Eq.~\eqref{eq:proof1}
	as in Eq.~\eqref{eq:av_0b},
	\begin{align}\label{eq:proof2}
		\sum_{v\in \by{4}^n}\sigma_{v}\!\left(\Gamma\right)_{\x\y}  &= \sum_{E_{v}\in \EE_n} \Big( \tr(E_vE^\dagger_\x \varrho) \tr(E_\y E^\dagger_v\varrho)  \Big)\tr(E_\x  E^\dagger_\y \varrho) \nn \\
		&= \frac{2^n}{K}\tr(E_\y  E^\dagger_\x \varrho) \tr(E_\x  E^\dagger_\y \varrho) \,.
	\end{align}
	Using Eq.~\eqref{eq:proof1} together with Eq.~\eqref{eq:proof2}, the matrix $\bar\Gamma$ has entries
	\begin{align}\label{eq:proof3}
		\overbar\Gamma_{xy}&=\frac{1}{2^nK\lvert\Aut_{0}(4,n)\rvert}
		\sum\limits_{\pi\in\Aut_{0}(4,n)}\pi
		\left(S\right)\,,
	\end{align}
	where 	$S_{xy}= \tr(E_\y  E^\dagger_\x \varrho)\tr(E_\x  E^\dagger_\y \varrho)$.
	With this notation, the matrix $L$ can be expressed as $L=S-\Gamma$.
	Using the definition Eq.~\eqref{eq:gammatilde_any} of $\tilde\Gamma$ and the identity Eq.~\eqref{eq:proof3} for  $\bar\Gamma$,
	proves that Eq.~\eqref{eq:L_average} holds:
	\begin{align}\label{eq:L_average2}
		&\frac{1}{\lvert\Aut_{0}(4,n)\rvert}\sum_{\pi \in \Aut_0(4,n)} \pi(L)  \nn\\
		&= 	\frac{1}{\lvert\Aut_{0}(4,n)\rvert}\sum_{\pi \in \Aut_0(4,n)} \pi(S)
		\quad - \quad
		\frac{1}{\lvert\Aut_{0}(4,n)\rvert}\sum_{\pi \in \Aut_0(4,n)} \pi(\Gamma) \nn \\
		&= (2^nK)\overbar\Gamma - \tilde\Gamma.
	\end{align}
	Recall that since $L$ is positive semidefinite,
	then $\pi(L)$ is also positive semidefinite for any $\pi \in \Aut(4,n)$.
	Eq.~\eqref{eq:L_average2} is a sum over positive semidefinite matrices and thus
	$(2^nK)\overbar\Gamma - \tilde\Gamma\succeq 0$.
	A multiplication by $\big(K(2^n - 1/K)\big)^{-1}$ yields Eq.~\eqref{eq:gammatildec_any}. This proves the claim.
\end{proof}

\subsection{Positive semidefinite constraints in the Terwilliger algebra}
To make use of Proposition~\ref{prop:pos_cond_av} in a symmetry-reduced semidefinite program,
we need to express the constraints in terms of the Terwilliger algebra.
\begin{proposition}\label{prop:terwilliger_main_reduction}
	The matrices $\tilde\Gamma$ and $\tilde\Gamma^c$ expand in the Terwilliger algebra as
	\begin{align}
		\tilde\Gamma& =\sum_{(i,j,t,p)\in \II(q,n)} x^{t,p}_{i,j} M^{t,p}_{i,j} \label{eq:gammatilde_anycode} \,,\\
		\tilde\Gamma^c &=\frac{2^n}{K}\sum_{(i,j,t,p)\in \II(q,n)} (x^{0,0}_{i+j-t-p}-x^{t,p}_{i,j}) M^{t,p}_{i,j} \label{eq:gammatilde_complement_anycode} \,.
	\end{align}
	Here $x^{t,p}_{i,j}=\frac{\lambda^{t,p}_{i,j}}{\gamma^{t,p}_{i,j}}$ with $\gamma^{t,p}_{i,j}$ defined in Eq.~\eqref{eq:gamma} and
	\begin{align}\label{eq:lambda_anycode}
		\lambda^{t,p}_{i,j}= \!\!\!\! \sum\limits_{
			\substack{x,y \in \by{4}^n \\ |s(\x)| = i
				\\ |s(\y)| = j
				\\ \lvert s(\x)\cap s(\x)\rvert = t
				\\ \lvert\{m \,\mid\, \x_{m}=\y_{m}\neq0\}\rvert = p \\}} \!\!\!\!\!\! \Gamma_{\x\y}  \,.
	\end{align}
\end{proposition}

\begin{proof}
	The matrix $\tilde\Gamma$ is invariant under $\Aut_0(q,n)$,
	and consequentially it can be expanded in the
	basis elements $M^{t,p}_{i,j}$ of the Terwilliger algebra.
	With the definition from Eq.~\eqref{eq:lambda_def}, the parameters $\lambda^{t,p}_{i,j}$
	are then given by Eq.~\eqref{eq:lambda_anycode} and $x^{t,p}_{i,j}=\frac{\lambda^{t,p}_{i,j}}{\gamma^{t,p}_{i,j}}$.

	We can express $\tilde\Gamma^c$ from Eq.~\eqref{eq:gammatildec_any}
	using Eq.~\eqref{eq:gamma_bar_bose},
	\begin{align}\label{eq:gammatilde_complementary_proof_anycode}
		\tilde\Gamma^c =\sum_{(i,j,t,p)\in \II(4,n)} (y_{i+j-t-p}-\frac{2^n}{K}x^{t,p}_{i,j}) M^{t,p}_{i,j}\,.
	\end{align}
	With this, Eq.~\eqref{eq:y_k_A} with $\Gamma_{xx}=\Gamma_{x0}$ shows that
	\begin{align}\label{eq:y_k_proof_anycode}
		y_k=(2^nK\gamma^{0,0}_{k,0})^{-1}\sum_{\substack {x\in \by{4}^n \\ |s(x)|=k}} \Gamma_{x0}=
		\frac{2^n}{K} x^{0,0}_{k,0}\,.
	\end{align}
	The combination of Eq. ~\eqref{eq:gammatilde_complementary_proof_anycode} with Eq. ~\eqref{eq:y_k_proof_anycode} results in Eq. ~\eqref{eq:gammatilde_complement_anycode}. This ends the proof.
\end{proof}

\subsection{Stabilizer codes}

Consider the case of a stabilizer code with stabilizer group $S$. Then we can recover the Propositions~\ref{prop:pos_cond_av} and~\ref{prop:terwilliger_main_reduction} in an alternative fashion for the matrices
\begin{align}
	\tilde \Gamma(\mathcal{S}) &=
	\frac{1}{|\mathcal{S}|}\sum\limits_{v\in \mathcal{S}}\Gamma^{v}  \label{eq:gamma_stab_first}\,, \\
	\tilde\Gamma^c(\mathcal{S})&=
	\frac{1}{4^n-|\mathcal{S}|}
	\sum\limits_{v \in \mathcal{S}\setminus \by{4}^n}\Gamma^{v}
	= \frac{1}{4^n-|\mathcal{S}|}
	\Big(4^n\overbar\Gamma-|\mathcal{S}|\tilde \Gamma(\mathcal{S})\Big) \label{eq:gamma_c_stab_first} \,,
\end{align}
that is, by replacing $Q=\{0^n\}$ with $Q=\mathcal{S}$
where $\mathcal{S}$ contains the unsigned Pauli strings of $S$ mapped to $4^n$.

Then as in the classical case discussed in Sections~\ref{subsec:average_codewords} and~\ref{subsec:average_codewords_comp}, these matrices are both positive semidefinite and invariant under $\Aut_{0}(4,n)$.
They can thus be expressed in a basis of the Terwilliger algebra $\mathcal{A}_{4,n}$,
recovering the same constraints on
$x^{t,p}_{i,j}$ as in the general case.
This follows from the fact that by the definition of $\Gamma^v$ in Eq.~\eqref{eq:Gv},
it holds that $\Gamma^v=\Gamma^{0}$ for any $v\in \mathcal{S}$ leading to $\tilde\Gamma=\Gamma^0$,
and that for a stabilizer code $|S|=|\mathcal{S}|=2^n/K$.

Let us explain this in more detail:
Consider Eq.~\eqref{eq:Gv},
\begin{align}
	\tilde \Gamma(\mathcal{S}) &=
\frac{1}{|\mathcal{S}|}\sum\limits_{v\in \mathcal{S}}\Gamma^{v}
=\frac{1}{|\mathcal{S}|} \frac{1}{\lvert\Aut_{0}(4,n)\rvert}\sum\limits_{v\in \mathcal{S}}	\sum\limits_{\pi\in\Aut_{0}(4,n)}\pi\left(\sigma_{v}\!\left(\Gamma\right)\right)\,.
\end{align}
Because $s\Pi = \Pi s=\Pi$ for all $s\in S$, we have for any $v \in \mathcal{S}$ that
\begin{align}
	\sigma_v(\Gamma)_{\x\y} &= \tr(E^\dagger_\x\varrho E_v) \tr(E_\y E^\dagger_v\varrho)\tr(E_\x E^\dagger_\y \varrho ) \nn \\
	&= \tr(E^\dagger_\x \varrho ) \tr(E_\y \varrho)\tr(E_\x E^\dagger_\y \varrho ) =\Gamma_{\x\y} \,.
\end{align}
This leads to
\begin{align}
\tilde \Gamma(\mathcal{S})
= \frac{1}{\lvert\Aut_{0}(4,n)\rvert} 	\sum\limits_{\pi\in\Aut_{0}(4,n)}\pi\left(\Gamma \right)
= \tilde \Gamma(0^n)\,,
\end{align}
which recovers the general case [Eq.~\eqref{eq:gammatilde_any}].
In analogy to Proposition~\eqref{prop:terwilliger_main_reduction},
\begin{align}
	\tilde\Gamma(\mathcal{S})&=\sum_{(i,j,t,p) \in \II(4,n)} x^{t,p}_{i,j} M^{t,p}_{i,j} \label{eq:gamma_stab} \,, \\
	\tilde\Gamma^c (\mathcal{S})&=\frac{|\mathcal{S}|}{4^n-|\mathcal{S}|}
	\sum_{(i,j,t,p) \in \II(4,n)} (x^{0,0}_{i+j-t-p}-x^{t,p}_{i,j}) M^{t,p}_{i,j} \label{eq:gamma_c_stab_second} \,,
\end{align}
since for a stabilizer code $|S|=|\mathcal{S}|=2^n/K$.

\subsection{Symmetry-reduced constraints}\label{subsect:sdp_aver}
We now translate all constraints on
$\Gamma$ into constraints on $x^{t,p}_{i,j}$.
Similar to the classical case in Ref.~\cite[Eq.~(49)]{GIJSWIJT20061719}, one has:
\begin{proposition} \label{prop:Terwilliger_extra_constraints}
With $\tilde \Gamma$ and $\tilde \Gamma^c$ expanded as in Proposition~\ref{prop:terwilliger_main_reduction},
the $x^{t,p}_{i,j}$ satisfy:
\begin{alignat}{2}
	\operatorname{(i)}& \quad x^{00}_{00}=1\,, \nn \\
	\operatorname{(ii)}& \quad x^{t,p}_{i,j}=0 \quad\quad\,\, \text{if} \quad t-p \quad \text{is odd}\,, \nn\\
	\operatorname{(iii)}& \quad x^{t,p}_{i,j}=x^{t',p'}_{i',j'} \quad \text{if} \quad  t-p=t'-p' \quad \text{is even} \quad \text{and}\nn \\
	&\quad\quad\quad\quad\quad\quad\,\, (i,j,i+j-t-p)\quad \text{a permutation of} \quad(i',j',i'+j'-t'-p')\,,\nn \\
	\operatorname{(iv)}& \quad \sum^{n}_{i=0}\gamma^{0,0}_{i,0}\x^{0,0}_{i,0} =\frac{2^n}{K} \,,
	\nn \\ 
	\operatorname{(v)}& \sum_{\substack{(i,j,t,p)\in \II(4,n)\\k=i+j-t-p}} \!\!\! \gamma^{t,p}_{i,j}x^{t,p}_{i,j} =\frac{2^n}{K} \gamma^{0,0}_{k,0} x^{0,0}_{k,0}.
\end{alignat}
\end{proposition}
\begin{proof}
Recall from
Eq.~\eqref{eq:lambda_anycode} that
\begin{align}\label{eq:x_anycode}
	\gamma^{t,p}_{i,j}x^{t,p}_{i,j}= \!\!\!\!\!\!\!\!\!\!\!\!
	\sum\limits_{\substack{x,y \in \by{4}^n \\ |s(\x)| = i
			\\ |s(\y)| = j
			\\ \lvert s(\x)\cap s(\x)\rvert = t
			\\ \lvert\{m \,\mid\, \x_{m}=\y_{m}\neq0\}\rvert = p \\}} \!\!\!\!\!\! \Gamma_{\x\y}  \,.
\end{align}
\noindent
(i) From $\Gamma_{00}=1$, $(M_{0,0}^{0,0})_{00} = 1$ and is zero otherwise, and from $\gamma^{0,0}_{0,0}=1$ it follows that $x^{0,0}_{0,0}=1$.

\noindent
(ii) Remark~\ref{rmk:interp} shows that, given $E_\x,E^\dagger_\y\in\EE_n$, the value $t-p$ counts the number of places in which
a nontrivial Pauli element in $E_\x$ does not cancel a nontrivial Pauli element in $E^\dagger_\y$.
As a consequence,
if $t-p$ is even, then $E_\x$ and $E^\dagger_\y$ commute;
if $t-p$ is odd, then $E_\x$ and $E^\dagger_\y$ anticommute.
Since $\Gamma_{\x\y}=0$ if $E_\x E^\dagger_\y+E^\dagger_\y E_\x=0$,
Eq.~\eqref{eq:x_anycode} then leads to $x^{t,p}_{i,j}=0$ if $t-p$ is odd. This proves condition (ii).

\noindent
(iii)
Eq.~\eqref{eq:structure_constraint} shows that
\begin{align}\label{eq:condition}
	\Gamma_{\x\y}=\Gamma_{\x'\y'}\,, \quad \quad\quad  &\text{if} \quad\quad  E_\x E^\dagger_\y  - E^\dagger_\y E_\x = 0 \quad \text{and}
	\nn\\&
	(E^\dagger_\x,E_\y,\omega_{\x\y}E_\x E^\dagger_\y)\quad \text{a permutation of} \quad(E^\dagger_{\x'},E_{\y'},\omega_{\x'\y'}E_{\x'}E^\dagger_{\y'})\,,
\end{align}
Recall from Remark~\ref{rmk:interp} that $\wt(E_{\x}E^\dagger_{\y})=i+j-t-p$ and note that $t-p=t'-p'$ for any permutation from Eq.~\eqref{eq:condition}.
Therefore if $t-p=t'-p'$ is even and $(i,j,i+j-t-p)$ is a permutation of $(i',j',i'+j'-t'-p')$, then
\begin{align}
\gamma^{t,p}_{i,j}x^{t,p}_{i,j}\quad=
\sum\limits_{\substack{ \x,\y \in \by{4}^n \\ |s(\x)| = i
		\\ |s(\y)| = j
		\\ \lvert s(\x)\cap s(\x)\rvert = t
		\\ \lvert\{m \,\mid\, \x_{m}=\y_{m}\neq0\}\rvert = p \\}} \!\!\!\!\!\! \Gamma_{\x\y} = \sum\limits_{\substack{\x',\y' \in \by{4}^n \\ |s(\x')| = i'
		\\ |s(\y')| = j'
		\\ \lvert s(\x')\cap s(\x')\rvert = t'
		\\ \lvert\{m \,\mid\, \x'_{m}=\y'_{m}\neq0\}\rvert = p' \\}} \!\!\!\!\!\! \Gamma_{\x'\y'}
		\quad =\quad  \gamma^{t',p'}_{i',j'}x^{t',p'}_{i',j'}\,.
\end{align}
Furthermore, one can check [Eq.~\eqref{eq:gamma}]
that $\gamma^{t,p}_{i,j}=\gamma^{t',p'}_{i',j'}$ which leads to condition (iii).

For the last two conditions we derive the projector constraints in Eq.~\eqref{eq:projconst_lvl} in terms of $x^{t,p}_{i,j}$ using that $\Gamma_{\x0}=\Gamma_{\x\x}$:

\noindent
(iv) Recall that $\sum_{\x \in \by{4}^n}\Gamma_{\x0}=\frac{2^n}{K}$ as shown in Eq.~\eqref{eq:projconst_lvl}. With this, Eq.~\eqref{eq:x_anycode}
leads to condition (iv).

\noindent
(v) Here we use the constraint $\sum_{\x \in \by{4}^n}\Gamma_{\x\y}=\frac{2^n}{K}\Gamma_{\x0}$ [Eq.~\eqref{eq:projconst_lvl}] in combination of conditions (ii)-(iii).
First, we write such constraint in terms of $x^{t,p}_{i,j}$ by using Eq.~\eqref{eq:condition}. This leads to
\begin{align}\label{eq:proofxijtp}
	\sum_{\substack{(i,j,t,p)\in \II(4,n) \\ i=k}} \gamma^{t,p}_{i,j} x^{t,p}_{i,j}=\frac{2^n}{K}\gamma^{0,0}_{k,0}x^{0,0}_{k,0}\,.
\end{align}
Condition (ii)-(iii), show that  $x^{t,p}_{i,j} = x^{t',p'}_{i',j'}$
if  $t-p=t'-p'$, $i=i'+j'-t'-p'$, and $j=j'$.
For this case also $\gamma^{t,p}_{i,j}  = \gamma^{t',p'}_{i',j'}$
by Eq.~\eqref{eq:gamma} and thus,
$\gamma^{t,p}_{i,j} x^{t,p}_{i,j} = \gamma^{t',p'}_{i',j'} x^{t',p'}_{i',j'}$. This allows us to write Eq.~\eqref{eq:proofxijtp} as
\begin{align}
	\sum_{\substack{(i,j,t,p)\in \II(4,n) \\ i=k}} \gamma^{t,p}_{i,j} x^{t,p}_{i,j} = \sum_{\substack{(i',j',t',p')\in \II(4,n) \\ i'+j'-t'-p'=k}} \gamma^{t',p'}_{i',j'} x^{t',p'}_{i',j'} =\frac{2^n}{K}\gamma^{0,0}_{k,0}x^{0,0}_{k,0} \,,
\end{align}
leading to condition (v).
This ends the proof.
\end{proof}

The $x^{t,p}_{i,j}$ is further constrained by the Knill-Laflamme conditions:.
\begin{proposition}\label{prop:kl_cond_sdp}
The Knill-Laflamme conditions as stated in Eq.~\eqref{eq:klred} imply that
	\begin{align}\label{eq:klred_xijtp}
\frac{K}{2^n}\sum^n_{i=0} K_j(i;n) \gamma^{0,0}_{i,0} x^{0,0}_{i,0}\quad \geq \quad \gamma^{0,0}_{j,0} x^{0,0}_{j,0} \quad  \text{with equality for} \quad  0<j < \delta \,,
	\end{align}
According to Eq.~\eqref{eq:extra_zeros}, pure codes additionally satisfy
	 \begin{align}\label{eq:klred_pure_xijtp}
	 	x^{t,p}_{i,j}=0 \quad \text{if} \quad \{i,j,i+j-t-p\} \cap \{1,\cdots, \delta-1\} \neq \emptyset\,,
	 \end{align}
which is the condition from Ref.~\cite[Eq.~(49)(iv)]{GIJSWIJT20061719}.
\end{proposition}
\begin{proof}
Eq.~\eqref{eq:x_anycode} shows that $\gamma^{0,0}_{j,0} x^{0,0}_{j,0} =\sum_{\x \in \by{4}^n,\, \wt(E_\x )=j} \Gamma_{\x0}$. This allows to write Eq.~\eqref{eq:klred} as Eq.~\eqref{eq:klred_xijtp}. By Eq.~\eqref{eq:x_anycode}, the condition for pure codes stated in Eq.~\eqref{eq:extra_zeros} can be written as $x^{t,p}_{i,j}=0$ for $0<i<\delta$ or $0<j<\delta$ or $0<i+j-t-p<\delta$. But this is nothing else than Eq.~\eqref{eq:klred_pure_xijtp}. This ends the proof.
\end{proof}
Note that, the constraint $\sum^{n}_{i=0}\gamma^{0,0}_{i,0}x^{0,0}_{i,0} =\frac{2^n}{K}$ is equivalent to $\frac{K}{2^n}\sum^n_{i=0} K_0(i;n) \gamma^{0,0}_{i,0} x^{0,0}_{i,0}=\gamma^{0,0}_{0,0} x^{0,0}_{0,0}$, since $K_0(i;n)=1$ for all $i$ and $\gamma^{0,0}_{0,0} x^{0,0}_{0,0}=1$.
It is thus possible to include either one or the other, according to convenience.

\begin{remark}\label{rmk:stab_codes_constraint}
	For stabilizer codes, one also has the condition in Ref.~\cite[Eq.~(49)(ii)]{GIJSWIJT20061719} which holds for classical codes,
	\begin{equation}
		0 \leq x_{i,j}^{t,p} \leq x_{i,0}^{0,0}\,.
	\end{equation}
\end{remark}
\begin{proof}

	Eq.~\eqref{eq:stab_code_Gamma} shows that stabilizer codes satify $ 0\leq\Gamma_{xy} \leq \Gamma_{xx}$.
	Recall that $\tilde\Gamma \propto \sum_{\pi\in \Aut(4,n)} \pi(\Gamma)$ [see Eq.~\eqref{eq:gammatilde_any}].
	Thus the same relation holds after averaging and $0 \leq \tilde\Gamma_{xy}\leq \tilde\Gamma_{xx}$
	since $\pi$ permutes rows and columns of $\Gamma$.

	Because $\tilde \Gamma$ is spanned by $M_{i,j}^{t,p}$, which have entries in $\{0,1\}$,
	is follows that $0\leq x^{t,p}_{i,j} \leq {x^{i,i}_{i,i}}$.
	According to condition (iii) from Proposition~\ref{prop:Terwilliger_extra_constraints}, $x^{i,i}_{i,i}=x^{0,0}_{i,0}$ and thus $0 \leq \x^{t,p}_{i,j} \leq x^{0,0}_{i,0}$.
\end{proof}

\subsection{Symmetry-reduced SDP bound}\label{subsec:sym_red_sdp}

We now state the symmetry reduced version of the main semidefinite programming bound in Eq.~\eqref{eq:sdpgamma}.
It includes all constraints on $x^{t,p}_{i,j}$
derived in the previous sections [Propositions~\ref{prop:pos_cond_av} to \ref{prop:kl_cond_sdp}].

\begin{theorem}\label{prop:sym_red_sdp}
If a $(\!(n,K,d)\!)_2$ quantum code exists, then the following semidefinite program is feasible:
\begin{align}\label{eq:sdpx}
	\textnormal{find} 		\quad & x^{t,p}_{i,j}  \nn \\
	\textnormal{subject to} 	\quad
	&  x^{0,0}_{0,0}=1, \nn \\ &x^{t,p}_{i,j}=0 \quad\quad\quad \text{if}\quad  t-p \quad  \text{is odd}\,,\nn \\
	& x^{t,p}_{i,j}=x^{t',p'}_{i',j'}\,\,\,\, \quad \text{if} \quad  t-p=t'-p' \quad \text{is even} \quad \text{and} \nn \\
	&\quad\quad\quad\quad\quad\quad\,\, (i,j,i+j-t-p) \quad\text{a permutation of} \quad(i',j',i'+j'-t'-p')\,, \nn \\
	& \sum^n_{i=0} \gamma^{0,0}_{i,0} x^{0,0}_{i,0}=\frac{2^n}{K}, \quad \sum_{\substack{(i,j,t,p)\in \II(4,n) \\ k=i+j-t-p}}  \gamma^{t,p}_{i,j}x^{t,p}_{i,j}=\frac{2^n}{K}  \gamma^{0,0}_{k,0}x^{0,0}_{k,0}  \,, \nn\\
	& K2^{-n}\sum^n_{i=0} K_j(i;n) \gamma^{0,0}_{i,0} x^{0,0}_{i,0}\quad \geq \quad \gamma^{0,0}_{j,0} x^{0,0}_{j,0} 
	\quad \text{with equality for}  \quad 0 <j < \delta  \,, \nn \\
	&\bigoplus_{\substack{a,k \in \N_0\\ 0\leq a\leq k\leq n+a-k}} \left(\sum\limits_{\substack{t,p\in \N_0 \\ 0 \leq p \leq t \leq i,j \\ i+j\leq t+n}}\alpha(i,j,t,p,a,k)x_{i,j}^{t,p}\right)_{i,j=k}^{n+a-k}   \succeq 0\,,\nn \\
	&\bigoplus_{\substack{a,k \in \N_0\\ 0\leq a\leq k\leq n+a-k}} \left(\sum\limits_{\substack{t,p\in \N_0 \\ 0 \leq p \leq t \leq i,j \\ i+j\leq t+n}}\alpha(i,j,t,p,a,k)(x_{i+j-t-p,0}^{0,0}-x_{i,j}^{t,p})\right)_{i,j=k}^{n+a-k} \succeq 0\,,
\end{align}
where $\alpha(i,j,t,p,a,k)$ is given by Eq.~\eqref{eq:alpha}.
\end{theorem}
As with the previous results, Theorem~\ref{prop:sym_red_sdp}
applies equally to stabilizer and non-stabilizer codes.
Furthermore, it is at least as strong as the linear programming bounds from Eq.~\eqref{eq:lp} without the shadow inequalities $S_j(\rho) \geq 0$.
Note that $\gamma_{j,0}^{0,0} x_{j,0}^{0,0}$ approximates $A_j(\varrho)$.
This allows us to add the rest of linear programming constraints as done in Eq.~\eqref{eq:sdpgamma}.
The SDP above can then be further restricted by adding:
\begin{enumerate}
	\item[(i)] The shadow inequalities [Eq.~\eqref{eq:shadowexp}] leads to
	\begin{equation}
2^{-n}\sum^n_{i=0} (-1)^i K_j(i;n) \gamma_{i,0}^{0,0} x_{i,0}^{0,0} \quad \geq \quad 0\,,
	\end{equation}
	with equality for $K=1$ and $n-j$ odd~\cite[Theorem 12]{796376}.

	\item[(ii)] For stabilizer codes, $0 \leq x_{i,j}^{t,p} \leq x_{i,0}^{0,0}$ [see Remark~\ref{rmk:stab_codes_constraint}] and
	\begin{align}
		\sum_{j\geq 0} \gamma_{2j,0}^{0,0} x_{2j,0}^{0,0}=
		\begin{cases}
			2^{n-\log_2(K)-1} \quad \text{(Type I)\,,} \\  2^{n-\log_2(K)} \quad \quad \!\text{(Type II)\,.}
		\end{cases}
	\end{align}
	\item [(iii)] For pure codes [see Eq.~\eqref {eq:klred_pure_xijtp}],
	\begin{align}
		x^{t,p}_{i,j}=0 \quad \text{if} \quad \{i,j,i+j-t-p\} \cap \{1,\cdots, \delta-1\}\neq \emptyset\,,
	\end{align}

\end{enumerate}

Therefore, SDP~\eqref{eq:sdpx}
is at least as strong as the linear programming bounds in Eq.~\eqref{eq:lp_Delsarte}.

We conclude this section by pointing out that the main semidefinite bound in Eq.~\eqref{eq:sdpgamma} with scaling of $O(4^n)$ is now symmetry reduced to Eq.~\eqref{eq:sdpx} with scaling of $O(n^4)$.
This fact allows us to handle codes with larger $n$.

\section{Non-existence of quantum codes}
\label{sec:appli_codes}
\subsection{The Lov\'asz bound refutes a $(\!(7,1,4)\!)_2$ code}
\label{sec:7qb}

Using the feasibility version of the symmetry reduced  Lov\'asz theta number SDP for self-dual quantum codes [Eq.~\eqref{eq:sdpgamma_pure}],
we prove the non-existence of a $(\!(7,1,4)\!)_2$ code.
The feasibility version includes as an extra condition the constraint $\sum_{x \in \by{4}^n} \Gamma_{xx} = 2^n$:
\begin{align}\label{eq:sdpgamma_pure_feas}
	\text{find}		\quad & \Gamma \nn \\
	\text{subject to} 	\quad & \Gamma_{00}=1 \,, \nn \\
	& \Gamma_{\x0}=\Gamma_{\x\x} \,, \nn \\
	& \Gamma \succeq 0 \,, \nn \\
	& \Gamma_{\x\y}=0 \quad  \text{if} \quad \x \sim \y 
	\,, \nn \\
	& \tr(\Gamma)=2^n\,.
\end{align}
Following Section~\ref{subsec:sym_red_sdp} we symmetry reduce the SDP in Eq.~\eqref{eq:sdpgamma_pure_feas}
as done for the main semidefinite programming bound [Eq.~\eqref{eq:sdpgamma}].
First apply conditions (i) and (iii) from Proposition~\eqref{prop:Terwilliger_extra_constraints}
to map $\Gamma_{00}=1$ and $\Gamma_{\x0}=\Gamma_{\x\x}$ to $x^{0,0}_{i,0}$ and $x^{0,0}_{i,0}=x^{i,i}_{i,i}$, respectively.
Then consider the block-diagonalization of  $\tilde\Gamma \succeq 0$ as done with $\tilde R$ [Eq.~\eqref{eq:psd_tildeR}].
Recall that $\Gamma_{xy}=0$ if $x\sim y$.
This implies that $x^{t,p}_{i,j}=0$
if $t - p$ is odd or $\{i,j,i+j-t-p\} \cap \{1,\cdots, \delta-1\} \neq \emptyset$.
Here we have used the condition for pure codes [Eq.~\eqref{eq:klred_pure_xijtp}]
and anti-commutator condition (ii) of Proposition~\eqref{prop:Terwilliger_extra_constraints}.
Finally, $\tr(\Gamma)=1$ is equivalent to $\sum_{x \in \by{4}^n} \Gamma_{x0} = 2^n$
due to $\Gamma_{\x0}=\Gamma_{\x\x}$.
This implies $\sum^{n}_{i=0}\x^{0,0}_{i,0} \gamma^{0,0}_{i,0}=\frac{2^n}{K}$
by condition (iv) in Proposition~\ref{prop:Terwilliger_extra_constraints}.
Thus the symmetry reduced version of Eq.~\eqref{eq:sdpgamma_pure_feas} reads:
\begin{align}\label{eq:red_lovasz}
	\text{find} 		\quad & x^{t,p}_{i,j}
	\nn \\
	\text{subject to} 	\quad
	& x^{0,0}_{0,0}=1\,, \nn \\
	& x^{0,0}_{i,0}= x^{i,i}_{i,i}\,, \nn \\
	& \bigoplus_{\substack{a,k \in \N_0\\ 0\leq a\leq k\leq n+a-k}}\left(\sum\limits_{\substack{t,p\in \N_0 \\ 0 \leq p \leq t \leq i,j \\ i+j\leq t+n}}\alpha(i,j,t,p,a,k)x_{i,j}^{t,p}\right)_{i,j=k}^{n+a-k} \succeq 0 \nn \\
	 &  x^{t,p}_{i,j}=0 \quad \text{if}\quad  t-p \quad  \text{is odd} \nn \\
	 & \hspace{1.75cm} \quad \text{or} \quad \{i,j,i+j-t-p\} \cap \{1,\cdots, \delta-1\} \neq \emptyset \,, \nn \\
	&\sum^n_{i=0} \gamma^{0,0}_{i,0} x^{0,0}_{i,0}=2^n \,,
\end{align}
where $\alpha(i,j,t,p,a,k)$ is given by Eq.~\eqref{eq:alpha}.
\begin{remark}
 Here we only took $\tilde \Gamma\succeq 0$ into account, leaving out the constraint $\tilde \Gamma^c\succeq 0$.
\end{remark}

Weak duality allows to prove infeasibility of a primal problem from its dual.
Appendix~\ref{app:dual}, Proposition~\ref{prop:lovasz_dual} shows that the program dual to that of Eq.~\eqref{eq:weak} is:
\begin{align}\label{eq:dual_sol}
	\alpha \quad = \quad \max_{Y^{(a,k)}\,,\, {w}} \quad
	&   (2^n-1)w - y_{0,0}^{0,0} \nn \\
	\text{subject to}  \quad & Y^{(a,k)} \succeq 0\,, \nn \\
	& y^{t,p}_{i,j} = \frac{1}{\gamma^{t,p}_{i,j}}\sum^{\min(i,j)}_{k=0} \sum^k_{a=\max(i,j)+k-n} \alpha(i,j,t,p,a,k) Y^{(a,k)}_{i-k,j-k}\,, \nn \\
	& y_{i,i}^{i,i} +  w + 2 y_{i,0}^{0,0} =0
	\quad\quad\quad\text{if} \quad \delta \leq i \leq n \,, \nn \\
	&  y^{t,p}_{i,j}=0
	\hspace{8em}  \text{if} \quad i,j\neq 0\,, \quad t-p \quad  \text{is even } \quad  \nn \\
	&\hspace{11.6em}\quad \text{and} \quad i + j - t - p \geq \delta \,.
\end{align}
Here the dual variables are real matrices $Y^{(a,k)}$ of size $n+a-2k$ with $ 0\leq a\leq k\leq n+a-k$.

If $\alpha>0$, then there do not exist a set of
$x^{t,p}_{i,j}$ satisfying the constraints from Eq.~\eqref{eq:red_lovasz},
and thus no corresponding code exists.
This is the case for the $(\!(7,1,4)\!)_2$ code:
\begin{observation}\label{prop:infeas_cert}
	The matrices $Y^{(a,k)}$ written in Appendix~\ref{app:inf} are a solution for Eq.~\eqref{eq:dual_sol} for $n=7$ and $\delta=4$ with $\alpha =0.58>0$,  maximum constraint violation of
	$\mathcal{O}(10^{-16})$, and minimum eigenvalue $\lambda_{\min}\approx 0.013$.
	This provides an infeasibility certificate for the primal program and shows that a $(\!(7,1,4)\!)_2$ quantum code does not exist.
\end{observation}
Note that the dual variable $w$ is determined by the matrices $Y^{(a,k)}$ since $w=-y^{i,i}_{i,i}-y^{0,0}_{i,0}$ if $\delta \leq i\leq n$.
In particular, $w=-y^{n,n}_{n,n}-y^{0,0}_{n,0}$\,.

\subsection{SDP refutations for the $(\!(8,9,3)\!)_2$ and $(\!(10,5,4)\!)_2$ codes}
\label{sec:sym_red_numerics}

We first remark the following:
in principle, adding constraints in the primal program makes a semidefinite program stronger.
From a numerical perspective however this can be disadvantageous:
the more constraints appear in the primal program,
the larger is the span of the dual feasible region.
As a consequence, the dual program can numerically be more costly to solve.
For obtaining a nice infeasibility certificate it can thus help to reduce the number of constraints.

For refuting the $(\!(8,9,3)\!)_2$ and $(\!(10,5,4)\!)_2$ codes,
the following relaxation of the SDP in Eq.~\eqref{eq:sdpx} of Theorem~\ref{prop:sym_red_sdp} suffices:
\begin{align}\label{eq:sdpx_relax}
	\textnormal{find} 		\quad & x^{t,p}_{i,j}  \nn \\
	\textnormal{subject to} 	\quad
	&  x^{0,0}_{0,0}=1, \nn \\ &x^{t,p}_{i,j}=0 \quad\quad\quad \text{if}\quad  t-p \quad  \text{is odd}\,,\nn \\
	& x^{t,p}_{i,j}=x^{t',p'}_{i',j'}\,, \quad \text{if} \quad  t-p=t'-p' \quad \text{is even} \quad \text{and} \nn \\
	&\quad\quad\quad\quad\quad\quad\,\, (i,j,i+j-t-p) \quad\text{a permutation of} \quad(i',j',i'+j'-t'-p')\,, \nn \\
	& \sum^n_{i=0} \gamma^{0,0}_{i,0} x^{0,0}_{i,0}=\frac{2^n}{K}, \quad \sum_{\substack{(i,j,t,p)\in \II(4,n) \\ k=i+j-t-p}}  \gamma^{t,p}_{i,j} x^{t,p}_{i,j}=\frac{2^n}{K} \gamma^{0,0}_{k,0} x^{0,0}_{k,0}   \,, \nn\\
	& K2^{-n}\sum^n_{i=0} K_j(i,n) \gamma^{0,0}_{i,0} x^{0,0}_{i,0}\quad = \quad \gamma^{0,0}_{j,0} x^{0,0}_{j,0} \quad\quad \text{if} \quad 0<j<\delta \,, \nn \\
	&\bigoplus_{\substack{a,k \in \N_0\\ 0\leq a\leq k\leq n+a-k}} \left(\sum\limits_{\substack{t,p\in \N_0 \\ 0 \leq p \leq t \leq i,j \\ i+j\leq t+n}}\alpha(i,j,t,p,a,k)x_{i,j}^{t,p}\right)_{i,j=k}^{n+a-k}   \succeq 0\,,
\end{align}
where	$\alpha(i,j,t,p,a,k)$ is given by Eq.~\eqref{eq:alpha}.
The SDP of Eq.~\eqref{eq:sdpx_relax} is obtained from the
SDP in Eq.~\ref{prop:sym_red_sdp}
by removing the constraint corresponding to $\tilde \Gamma^c \succeq 0$ and
to only take the equality conditions
in $K2^{-n}\sum^n_{i=0} K_j(i,n) \gamma^{0,0}_{i,0} x^{0,0}_{i,0} \, \geq \,\gamma^{0,0}_{j,0} x^{0,0}_{j,0}$.

Appendix~\ref{app:dual}, Proposition~\ref{prop:dual_all} shows that the dual program of Eq.~\eqref{eq:sdpx_relax} is:
\begin{align}\label{eq:dual_solall}
	\alpha \quad = \quad \max_{Y^{(a,k)}, Q_i, C_i} \quad &   (KD_0-C_0) + \left(\frac{2^n}{K} -1\right)Q_0 - {y^{0,0}_{0,0}} \quad &  \nn \\
	\text{subject to}  \quad & Y^{(a,k)} \succeq0,\nn \\
	& y^{t,p}_{i,j}\ = \frac{1}{\gamma^{t,p}_{i,j}} \sum^{\min(i,j)}_{k=0} \sum^k_{a=\max(i,j)+k-n} \alpha(i,j,t,p,a,k) Y^{(a,k)}_{i-k,j-k}\,, \nn \\
	&D_i = 2^{-n} \sum^{\delta-1}_{j=0} K_j(i,n)  C_{j} \,, \nn \\
	&(KD_i-C_i)+ \left(\frac{2^n}{K}-2\right)Q_i - 2y^{0,0}_{i,0}
	-Q_0  -y^{i,i}_{i,i}=0
	\quad \text{if} \quad 0 <i <\delta \,,\nn \\
	& KD_i + \left(\frac{2^n}{K}-2\right)Q_i - 2y^{0,0}_{i,0}
	-Q_0  -y^{i,i}_{i,i}=0
	\quad\quad\quad\quad \text{if} \quad \delta \leq i \leq n \,, \nn \\
	& \sum_{\substack{ (i',j',t',p')\in \II(4,n) \\ (t-p)=(t'-p') \\ (i',j',i'+j'-t'-p') \, \text{a} \\
			\, \text{permutation of} \,  (i,j,i+j-t-p) \\ }} \!\!\!\!\!\!\!\!\!\!  
	\left(y^{t',p'}_{i',j'}+ Q_{i'+j'-t'-p'}\right) = 0 \quad  \hspace{3.3em} \text{if} \quad (t-p) \quad \text{even}  \nn \\ &  \hspace{6cm}  \quad \text{and} \quad i,j,i+j-t-p \neq 0 \,.
\end{align}
Here the dual variables are real matrices $Y^{(a,k)}$ of size $n+a-2k$ with $ 0\leq a\leq k\leq n+a-k$ and
$Q_i, C_i \in \R$ with $0\leq i\leq n$.

\begin{observation}\label{prop:infeas_cert2}
The SDP from Eq.~\eqref{eq:dual_solall} has solutions with $\alpha > 0$
for the parameters $(\!(8,9,3)\!)_2$ and $(\!(10,5,4)\!)_2$. In particular,
for the $(\!(8,9,3)\!)_2$, there is a dual solution with $\alpha \approx 0.989$, maximum constraint violation of $\mathcal{O}(10^{-13})$,
and minimum eigenvalue $\lambda_{\min} \approx 0.0010$.
For $(\!(10,5,4)\!)_2$, there is a dual solution with $\alpha\approx1$, maximum constraint violation of $\mathcal{O}(10^{-13})$ and minimum eigenvalue $\lambda_{\min}\approx 0.050$.
As a consequence, codes with parameters $(\!(8,9,3)\!)_2$ and $(\!(10,5,4)\!)_2$ do not exist. The infeasibility certificates can be found online at \url{https://github.com/ganglesmunne/SDP_bounds_for_quantum_codes} [Ref. \cite{Munné2024github}].
\end{observation}

\section{Discussion}

We provided a complete semidefinite programming hierarchy for the existence of quantum codes,
and relaxations that recover a Lov\'asz bound for self-dual quantum codes, a quantum Delsarte bound,
and the quantum version of Schrijver and Gijswijt's semidefinite programming bounds that use the Terwilliger algebra.
With this we show that $(\!(7,1,4)\!)_2$, $(\!(8, 9, 3)\!)_2$, and $(\!(10, 5, 4)\!)_2$ codes do not exist.
While the non-existence of $(\!(7,1,4)\!)_2$ was known previously through analytical methods~\cite{PhysRevLett.118.200502},
to our knowledge the latter two bounds are new~\cite[Table III]{Rigby_2019}.

An interesting aspect are
the differences between the classical and quantum versions of these semidefinite programming bounds.
These are perhaps best illustrated in the Lov\'asz bound [Eqs.~\eqref{eq:lovasz_SDP1}, ~\eqref{eq:lovasz_SDP2}]:
in the classical case the Hamming distance indicates which entries in $M = \chi \chi^T$ are set to zero.
For quantum codes the anti-commutation relations between pairs of observables yield additional constraints:
$\Gamma_{ab}$ is set to zero
if the Pauli strings $E_a$, $E_b$ anticommute, which can be captured by a graph with loops.
In particular, this leads to different quantum Lov\'asz and quantum Delsarte bounds
than in the classical case of quaternary codes.

A further interesting aspect is the use of the quantum MacWilliams identity to recover
expressions of the form $\tr(E \rho E^\dag \varrho)$
that, in principle, are outside of the state polynomial optimization framework.
One can hope that this opens the door to the application of state polynomial optimization
to optimize over arbitrary expressions in observables and states.
A possible approach is to use higher-order generalizations of quantum MacWilliams identities
such as presented in Ref.~\cite{TRINKER20112283}.

We note that our hierarchy also applies to the existence of absolutely maximally entangled and $m$-uniform states, these corresponding to self-dual codes~\cite{Rather_2022, 1751-8121-51-17-175301}.
With small modifications, the hierarchy also lower bounds the average purities of $n$-qubit states,
such as found in the Meyer-Wallach~\cite{doi:10.1063/1.1497700, DBLP:journals/qic/Brennen03, PhysRevA.69.052330},
potential of multipartite entanglement~\cite{Facchi_2010},
and concentratable entangled measures~\cite{PhysRevResearch.6.023019}.

Finally, we want to highlight that our approach is formally dimension-free.
That is, on can target scenarios where the Hilbert space is not known.
This could be interesting for the study of quantum version of the quantum zero-error communication capacity of a quantum graph~\cite{Duan_2013}.

 We end with some open questions:
 \begin{itemize}

 \item Can one extend the Lov\'asz and Delsarte bounds for self-dual quantum codes to codes with $K>1$?
 This can possibly achieved by purification of the code with a reference system.

 \item The shadow bounds provide powerful constraints, e.g., these alone refute the existence of a $(\!(4,1,3)\!)_2$ code.
 Is there a formulation of the generalized shadow bounds that can be useful as part of an SDP hierarchy~\cite{817508}?

 \item Does our main SDP relaxation imply the shadow bounds?
		If not, are they included at some finite level of the hieararchy?

 \item Can our SDP bounds be strengthened using the techniques developed by Laurent~\cite{Laurent2007},
 the split Terwilliger algebra~\cite{Tseng2023}, or by 4-point bounds~\cite{6142090, Litjens2017}?

 \item Can similar methods be established for other types of codes~\cite{ErrorCorrectionZoo},
 e.g. subsystem codes~\cite{Aly2006},
 hybrid quantum-classical codes~\cite{8006823, 9320518},
 quantum MDS codes~\cite{Huber_2020}
 codes over arbitrary quantum channels~\cite{9741787},
 or entanglement-assisted codes~\cite{lai2023semidefinite}?
 Likewise, can LP and SDP bounds be derived for quantum spherical codes~\cite{Jain2024} by analyzing the Johnson scheme as in the classical case~\cite{Delsarte1977, Bachoc2008}?

 \item The Hamming and Singleton bounds can be recovered from the both the classical Delsarte bound~\cite{Delsarte1972}.
 In the quantum case, the quantum Singleton bound can be obtained,
 as well as a quantum Hamming bound for non-degenerate codes~\cite{761270}.
 It is an open question whether a quantum Hamming bound applies also to degenerate codes,
 with recent results suggesting that it does~\cite{10083270, nemec2023hamminglikebounddegeneratestabilizer}.
 Can the SDP techniques introduced in this paper be used to prove that the quantum Hamming bound must apply to any quantum code?

 \item In the case of pure codes
 the Knill-Laflamme constraints on the variables $x_{i,j}^{t,p}$ in Proposition~\ref{prop:kl_cond_sdp}, Eq.~\eqref{eq:klred_pure_xijtp}
 appear to be stronger than those of the linear programming bounds that make use of the quantum MacWilliams identity.
 Does there exist a quantum MacWilliams-type identity also for the $x_{i,j}^{t,p}$
 (e.g. similar to Ref.~\cite{TRINKER20112283}),
 allowing for finer-grained constraints on the minimum distance of impure codes?

 \item How can this be generalized to qudit codes;
 can a characterization of quantum algebraic tori be used in place of the result by Gastineau-Hills~\cite{panov1998quantum}?

 \end{itemize}

\subsection{Related work}
Lai et al.~\cite{lai2023semidefinite} have introduced an SDP relaxation in the context of entanglement-assisted codeword stabilized codes. When restricted to the case with no pre-shared entanglement, their bounds showed no advantage over the LP bounds.
Berta et al.~\cite{BertaBorderiFawziScholz2021}
presented a SDP hierarchy to determine the existence of a quantum code
based on the symmetric extension hierarchy~\cite{PhysRevLett.88.187904};
see also the follow up work~\cite{kossmann2025approximatequantumerrorcorrection}.
This hierarchy is also complete but does not easily allow for numerically tractable relaxations.
In the classical case, a complete hierarchy for the existence of codes can be
formulated using $0-1$ programming
for maximum independent sets~\cite{10.1137/0801013,lasserre2001explicit}.

\subsection*{Data and code availability}

To solve the SDPs from Eq.~\eqref{eq:dual_sol} and Eq.~\eqref{eq:dual_solall}, we used the Python API  PICOS~\cite{PICOS} together with the SDP solver CVXOPT~\cite{cvx}.
The program and the SDP solutions can be found online in \url{https://github.com/ganglesmunne/SDP_bounds_for_quantum_codes} [Ref. \cite{Munné2024github}].
To obtain solutions with some rational values we run the SDP iteratively,
where in each iteration an an entry of $Y^{(a,k)}$ was forced to be equal to a rounded previous entry.
This procedure leads to the rational solution in Appendix~\ref{app:inf}.

\subsection*{Conflict of interest}

The authors declare no conflict of interest.
		
\bibliographystyle{abbrv}
\bibliography{current_bib}

\appendix

\section{Dual programs}\label{app:dual}
	
	Here we derive the dual of the primal programs from Section~\ref{sec:appli_codes}
	to obtain infeasibility certificates for codes with parameters $(\!(7,1,4)\!)_2$,  $(\!(8,9,3)\!)_2$ and $(\!(10,5,4)\!)_2$.
	
	Following Boyd and Vandenberghe~\cite[Section 4.6.2, Example 5.12]{boyd2004convex}, we first explain such procedure for the case when the primal in its standard form [see Eq.~\eqref{eq:primal}] which reads
	\begin{align}\label{eq:primal_app}
		\text{minimize} & \quad \tr(CX) \nn \\
		\quad \text{subject to} & \quad \tr(A_i X)=b_i \quad \text{for} \quad i=1,\dots, p \,,\nn \\
		& \quad X\succeq 0 \,.
	\end{align}
	The transition between a primal and the dual problems is done via their Lagrangian:
	associate with each equality constraint $\tr(A_i X)=b_i$ a dual variable $y_i\in \R$ and with the semidefinite constraint $X\succeq 0$ a dual semidefinite variable $Y\succeq 0$. The Lagrangian is then defined as
	\begin{align}\label{eq:lagran}
		\mathcal{L}(X,y_i,Y)&= \tr(CX) - \tr(YX) + \sum^p_{i=1} y_i \big(\tr(A_i X)- b_i\big)   \nn \\
		&= \tr\Big(\big(C-Y + \sum^p_{i=1} y_iA_i \big) X \Big) - \sum_{i=1}^p y_i b_i\,,
	\end{align}
	so the dual objective function is
	\begin{align}
		\inf_X \mathcal{L}(X,y_i,Y) =
		\begin{cases} -\sum^p_{i=1} y_i b_i \quad \text{if}\quad  \sum^p_{i=1} y_iA_i-Y+C=0\,, \\ -\infty \quad\quad\quad\quad \text{otherwise}\,.
		\end{cases}
	\end{align}
	The dual problem can then be expressed as
	\begin{align}
		\text{maximize} \quad & -\sum^p_{i=1} y_i b_i \nn \\
		\text{subject to} \quad &  C+\sum^p_{i=1} y_iA_i=Y \quad \text{for} \quad  i=1,\dots,p \,,\nn \\
		\quad & Y\succeq 0 \,.
	\end{align}
	By eliminiating the slack variable $Y$ and redefining $y_i$ with  $-y_i$, one obtains
	\begin{align}
		\text{maximize} \quad & \sum^p_{i=1} y_i b_i \nn \\
		\text{subject to} \quad &  C-\sum^p_{i=1} y_iA_i\succeq0 \quad \text{for} \quad  i=1,\dots,p \,.
	\end{align}

	\subsection{Dual of the Lov\'asz bound}
	We use this to write the dual program of
	the feasibility version of the symmetry reduced quantum Lov\'asz theta number for self-dual codes,
	\begin{proposition}\label{prop:lovasz_dual}
	The dual program of Eq.~\eqref{eq:red_lovasz} is
			\begin{align}\label{eq:dual_sol_app_prop}
			\alpha \quad = \quad \max_{Y^{(a,k)}\,,\, {w}} \quad
			&   (2^n-1)w - y_{0,0}^{0,0} \nn \\
			\text{subject to}  \quad & Y^{(a,k)} \succeq 0\,, \nn \\
			& y^{t,p}_{i,j} = \frac{1}{\gamma^{t,p}_{i,j}} \sum^{\min(i,j)}_{k=0} \sum^k_{a=\max(i,j)+k-n} \alpha(i,j,t,p,a,k) Y^{(a,k)}_{i-k,j-k}\,, \nn \\
			& y_{i,i}^{i,i} +  w + 2 y_{i,0}^{0,0} =0
			\quad\quad\quad\text{if} \quad \delta \leq i \leq n \,, \nn \\
			&  y^{t,p}_{i,j}=0
			\hspace{8em}  \text{if} \quad i,j\neq 0\,, \quad t-p \quad  \text{is even } \quad  \nn \\
			&\hspace{11.6em}\quad \text{and} \quad i + j - t - p \geq \delta \,.
		\end{align}
		where $Y^{(a,k)}$ with $0\leq a\leq k\leq n+a-k$
		are real square matrices of size $(n+a-2k)$.
	\end{proposition}
	\begin{proof}
	
	For the derivation of the dual program,
	we slightly modify the conditions from Eq.~\eqref{eq:red_lovasz} when $x^{t,p}_{i,j}=0$  to obtain the equivalent formulation:
	\begin{align}\label{eq:red_relax}
		\text{find} 		\quad & x^{t,p}_{i,j}  
		\nn \\
		\text{subject to} 	\quad
		& x^{0,0}_{0,0}=1,\quad  x^{0,0}_{i,0}= x^{i,i}_{i,i}, \quad \sum^n_{i=0} \gamma^{0,0}_{i,0} x^{0,0}_{i,0}=2^n,
		\nn \\
		\quad &  x^{t,p}_{i,j}=0 \quad\quad \text{if}\quad  t-p \quad  \text{is odd} \quad \text{or} \quad 0 < i + j - t - p < \delta  \,,
		\nn \\
		\quad 
		& \bigoplus_{\substack{a,k \in \N_0\\ 0\leq a\leq k\leq n +a-k}}\left(\sum\limits_{\substack{t,p\in \N_0 \\ 0 \leq p \leq t \leq i,j \\ i+j\leq t+n}}\alpha(i,j,t,p,a,k)x_{i,j}^{t,p}\right)_{i,j=k}^{n+a-k} \succeq 0 \,.
	\end{align}
	where $\alpha(i,j,t,p,a,k)$ is given by Eq.~\eqref{eq:alpha}.
	Note that the condition $x^{t,p}_{i,j}=0$ if $0 < i+j-t-p <\delta$ is equivalent to the pure code condition from Eq.~\eqref{eq:klred_pure_xijtp}.
	This can be seen from that fact that when $j=0$, it simplifies to $x^{0,0}_{i,0}=x^{i,i}_{i,i}=0$ if $0 < i <\delta$.
	But $x_{i,i}^{i,i} = 0$ implies that $x_{i,j}^{t,p} = 0$ for any $j,t,p$:
	Consider $\x \in \by{4}^n$ with $s(\x)=i$.
	Then $\tilde\Gamma_{\x\x}=x^{i,i}_{i,i}=0$. The constraint $\tilde\Gamma \succeq 0$ implies that $\tilde\Gamma_{\x\y}=x^{t,p}_{i,j}=0$ for any $\y$ with $|s(\y)| = j$, $\lvert s(\x)\cap s(\y)\rvert = t$ and $\lvert\{m \,\mid\, \x_{m}=\y_{m}\neq0\}\rvert = p$.
	This is due to the fact that the $2\times 2$ minor of $\tilde \Gamma$ involving $x$ and $y$ must be positive-semidefinite, and thus if $\tilde \Gamma_{xx} = 0$ then also $\tilde \Gamma_{xy} = 0$.

	Consider now the Lagrangian of SDP~\eqref{eq:red_relax}.
	Associate with each equality constraint a dual variable
	 ${u},{v_i},{w},r^{t,p}_{i,j}\in \R$
	 and with each semidefinite constraint a dual semidefinite variable $Y^{(a,k)}\succeq0$.
	 Eq.~\eqref{eq:red_relax} is a feasibility problem and so the objective function equals zero.
	 Then
	\begin{align}\label{eq:Lagrange_Lovasz}
		\mathcal{L} &= \left(1-x^{0,0}_{0,0}\right) {u}+ \sum^n_{i=1} \left( x^{0,0}_{i,0}-x^{i,i}_{i,i}\right) {v_i} \nn \\
		&\quad
		+ \left(2^n-\sum^n_{i=0}\gamma^{0,0}_{i,0} x^{0,0}_{i,0}\right){w} \quad
		+ \sum_{\substack{ (i,j,t,p)\in \II(4,n) \\ (t-p) \,\text{odd} \text{ or} \\ 0<i+j-t-p<\delta}} x^{t,p}_{i,j} r^{{t,p}}_{i,j}
		\nn \\
		&\quad-
		\sum_{(i,j,t,p)\in \II(4,n)} \left(\sum^{\min(i,j)}_{k=0} \sum^k_{\substack{a=\max(i,j)\\+k-n}} \!\!\!\!\!\! \alpha(i,j,t,p,a,k) Y^{(a,k)}_{i-k,j-k}\right)x^{t,p}_{i,j}\,.
	\end{align}
	For the last term, we have developed the trace of each block of the direct sum in Eq.~\eqref{eq:red_relax} multiplied by the dual variable $Y^{(a,k)}$ to obtain
	\begin{align}\label{eq:dual_semidefinite1}
\sum_{\substack{a,k \in \N_0\\ 0\leq a\leq k\leq n +a-k}}\left(\sum^{n+a-k}_{i,j=k} \sum\limits_{\substack{t,p\in \N_0 \\ 0 \leq p \leq t \leq i,j \\ i+j\leq t+n}}\alpha(i,j,t,p,a,k)x_{i,j}^{t,p} Y^{(a,k)}_{i-k,j-k}\right) \,,
	\end{align}
	where we have used the fact that $Y^{(a,k)}$ is symmetric. Define the following set corresponding
	to the sums appearing in
	Eq.~\eqref{eq:dual_semidefinite1},
	\begin{align}\label{eq:R_set1}
		\mathcal{R}= \big\{ (i,j,t,p,a,k)\,\big|\,\, &k \leq i,j \leq n+a-k\,, \quad 0 \leq a\leq k\leq n +a-k\,, \nn \\
		  & i+j\leq t+n \quad \text{and} \quad 0 \leq p \leq t \leq i,j \big\}\,.
	\end{align}
	Eq.~\eqref{eq:dual_semidefinite1} then simplifies to
	$\sum_{(i,j,t,p,a,k)\in \mathcal{R}} \alpha(i,j,t,p,a,k) x^{t,p}_{i,j}  Y^{(a,k)}_{i-k,j-k}$.
	Note that $\mathcal{R}$ is equivalently defined by
	Eq.~\eqref{eq:dual_semidefinite1} is then simplify to $\sum_{(i,j,t,p,a,k)\in \mathcal{R}} \alpha(i,j,t,p,a,k) x^{t,p}_{i,j}  Y^{(a,k)}_{i-k,j-k}$. Note that the set $\mathcal{R}$ is equivalently defined by
	\begin{align}
		\mathcal{R}=\big\{ (i,j,t,p,a,k)\,\big|\,\, &(i,j,t,p)\in \II(4,n)\,, \quad 0 \leq k \leq \min(i,j) \,, \nn \\ & \max(i,j) \leq n+a-k\,, \quad \text{and} \quad a \leq k \big\}\,,
	\end{align}
	with $\II(4,n)$ given by Eq.~\eqref{eq:I-range}.
	To see this note that the first two inequalities in Eq.~\eqref{eq:R_set1} combine to
	$0\leq a \leq k \leq i,j \leq n+a-k$ from which $i,j\leq n$ follows.
	Then the remaining condition $0\leq p\leq t\leq i,j$ is captured by $\II(4,n)$.
	Therefore, $\sum_{(i,j,t,p,a,k)\in \mathcal{R}} \alpha(i,j,t,p,a,k) x^{t,p}_{i,j}  Y^{(a,k)}_{i-k,j-k}$ can also be written as
	\begin{align}\label{eq:dual_semidefinite2}
	\sum_{(i,j,t,p)\in \II(4,n)} \left(\sum^{\min(i,j)}_{k=0} \sum^k_{\substack{a=\max(i,j)+k-n}} \!\!\!\!\!\! \alpha(i,j,t,p,a,k) Y^{(a,k)}_{i-k,j-k}\right)x^{t,p}_{i,j}\,.
	\end{align}
	
	Define now $y^{t,p}_{i,j}\ = \sum^{\min(i,j)}_{k=0} \sum^k_{a=\max(i,j)+k-n} \alpha(i,j,t,p,a,k) Y^{(a,k)}_{i-k,j-k}$
	which is the inner term of the last expression in Eq.~\eqref{eq:Lagrange_Lovasz}. Since $Y^{(a,k)}_{i-k,j-k}$ and $\alpha(i,j,t,p,a,k)$ are symmetric functions with respect to $i$ and $j$, then also $y^{t,p}_{i,j}$ is symmetric.
	Factorize the primal variables in the Lagrangian of Eq.~\eqref{eq:Lagrange_Lovasz},
	\begin{align}
		\mathcal{L}&=({u}+2^n{w}) 
		+ x^{0,0}_{0,0}\Big( - {u} - {w} - y^{0,0}_{0,0}\Big) 
		%%%%%
		+\sum^{\delta-1}_{i=1} x^{0,0}_{i,0}\Big({v_i}  -\gamma^{0,0}_{i,0} {w} +r^{0,0}_{i,0}+r^{0,0}_{0,i} - 2y^{0,0}_{i,0} \Big) \nn \\
		%%%
		&\quad 		+\sum^{n}_{i=\delta} x^{0,0}_{i,0}\Big( {v_i} -\gamma^{0,0}_{i,0} {w}  - 2y^{0,0}_{i,0}\Big)
		%%%% 
		+\sum^{n}_{i=1} x^{i,i}_{i,i}\Big(- {v_i} - y^{i,i}_{i,i}\Big) \nn \\
		&\quad +\sum_{\substack{(i,j,t,p)\in \II(4,n) \\ i,j\neq 0\\ (t-p) \,\text{odd} \text{ or} \\ 0<i+j-t-p<\delta}} x^{t,p}_{i,j}\Big(r^{t,p}_{i,j}-y^{t,p}_{i,j}\Big)
		- \sum_{\substack{(i,j,t,p)\in \II(4,n) \\ i,j\neq 0\\ (t-p) \text{ even and} \\ i+j-t-p \geq \delta}} x^{t,p}_{i,j}y^{t,p}_{i,j} \,.
	\end{align}
	Here we used that $\gamma^{0,0}_{0,0}=1$, $x^{00}_{i0}=x^{00}_{0i}$, and $y^{00}_{i0}=y^{00}_{0i}$, to decompose
	
	\begin{align}\label{eq:clar_fact_7}
	\sum_{\substack{ (i,j,t,p)\in \II(4,n) \\ (t-p) \,\text{odd} \text{ or} \\ 0<i+j-t-p<\delta}} x^{t,p}_{i,j} r^{{t,p}}_{i,j}
	&= \sum^{\delta-1}_{i=1}x^{0,0}_{i,0} (r^{{0,0}}_{i,0}+r^{{0,0}}_{0,i}) +\sum_{\substack{(i,j,t,p)\in \II(4,n) \\ i,j\neq 0\\ (t-p) \,\text{odd} \text{ or} \\ 0<i+j-t-p<\delta}} x^{t,p}_{i,j} r^{{t,p}}_{i,j} \,, \\
	%%%%
	%%%%
	\sum_{(i,j,t,p)\in \II(4,n)} x^{t,p}_{i,j} y^{t,p}_{i,j}&= x^{0,0}_{0,0} y^{0,0}_{0,0} +  \sum^{n}_{i=1}2x^{0,0}_{i,0} y^{0,0}_{i,0} + \!\!\!\!\!\!\!\!\sum_{\substack{(i,j,t,p)\in \II(4,n) \\ i,j\neq 0\\ (t-p) \,\text{odd} \text{ or} \\ 0<i+j-t-p<\delta}} x^{t,p}_{i,j} y^{{t,p}}_{i,j} + \!\!\!\!\!\!\!\!\sum_{\substack{(i,j,t,p)\in \II(4,n) \\ i,j\neq 0,\,\\ (t-p) \text{ even and} \\ i+j-t-p \geq \delta, \\ }} x^{t,p}_{i,j}y^{t,p}_{i,j} \,.
	\end{align}

	The dual objective function is then 
		\begin{align}\label{eq:dualconst}
		\inf_{x^{t,p}_{i,j}} \mathcal{L}&=({u}+2^nw) \quad \text{if} \quad
		\begin{cases}
		 (1)\quad 0=-{u}- {w}-y^{0,0}_{0,0}\,, \\
		 (2)\quad 0={v_i}  -\gamma^{0,0}_{i,0} {w} +r^{0,0}_{i,0}+r^{0,0}_{0,i} - 2y^{0,0}_{i,0}
		  &\text{if} \quad 0 <i <\delta \,,\\
		(3)\quad 0={v_i} -\gamma^{0,0}_{i,0} {w}  - 2y^{0,0}_{i,0}
		&\text{if} \quad \delta \leq i \leq n \,,\\
		(4)\quad 0= -{v_i}
		-y^{i,i}_{i,i}
		&\text{if} \quad 0 < i \leq n \,,\\
		(5)\quad 0=r^{t,p}_{i,j} - y^{t,p}_{i,j}
		&\text{if} \quad (t-p) \,\, \text{odd} \quad \text{or} \\ \quad   &\quad 0<i+j-t-p<\delta \quad \\ & \quad\quad \text{with} \quad  i,j\neq 0 \,, \\
		(6)\quad 0=y^{t,p}_{i,j}
		&\text{if} \quad i,j\neq 0, \quad (t-p) \,\, \text{even } \\
		&\quad \text{and} \quad   \delta  \leq i+j-t-p\,.
			\end{cases}
	\end{align}
	Otherwise, it tends to $-\infty$.
	Note that conditions (2) and (5) appearing in Eq.~\eqref{eq:dualconst},
	\begin{align}\label{eq:dual_redun}
		{v_i}  -\gamma^{0,0}_{i,0} {w} +r^{0,0}_{i,0}+r^{0,0}_{0,i} - 2y^{0,0}_{i,0}=0
		\quad &\text{if} \quad 0 <i <\delta \,, \nn \\
		r^{t,p}_{i,j} - y^{t,p}_{i,j}=0	 \quad &\text{if} \quad (t-p) \,\, \text{odd} \quad \nn \\ &\quad\quad \text{or} \quad 0<i+j-t-p<\delta \quad \text{with} \quad  i,j\neq 0 \,.
	\end{align}
	do not constrain the objective function.
	This can be seen from the fact that the variables $r^{t,p}_{i,j}$ do not appear in the objective function and neither in the rest of constraints. Thus, they act as slack variables which can always take values so that the above two conditions are satisfied.
	
	All dual variables can be expressed in terms of the semidefinite dual constraints $Y^{(a,k)}$ and $w$
	using the remaining constraints in Eq.~\eqref{eq:dualconst}:
	use condition (1) to eliminate $u$ and condition (4) to eliminate $v_i$.

	The dual problem can then be expressed as Eq.~\eqref{eq:dual_sol_app_prop}.
	In particular, the objective function reads $u+2^nw = (2^n-1)w - y_{0,0}^{0,0}$,
	and the remaining constraints (3) and (6) read:
	\begin{align}
	 0 &= y_{i,i}^{i,i} + \gamma_{i,0}^{0,0} w + 2 y_{i,0}^{0,0} \quad &\text{if}\quad &\delta \leq i \leq n \nn\\
	 0 &=y^{t,p}_{i,j}		&\text{if} \quad
	   &i,j\neq 0, \quad (t-p) \,\, \text{even}\,, \quad   i+j-t-p\geq \delta \,.
	\end{align}
	
	The dual program then reads
	\begin{align}
		\alpha \quad = \quad \max_{Y^{(a,k)}\,,\, {w}} \quad
		&   (2^n-1)w - y_{0,0}^{0,0} \nn \\
		\text{subject to}  \quad & Y^{(a,k)} \succeq 0\,, \nn \\
		& y^{t,p}_{i,j} = \sum^{\min(i,j)}_{k=0} \sum^k_{a=\max(i,j)+k-n} \alpha(i,j,t,p,a,k) Y^{(a,k)}_{i-k,j-k}\,, \nn \\
		& y_{i,i}^{i,i} + \gamma_{i,0}^{0,0} w + 2 y_{i,0}^{0,0} =0
		\quad\quad\text{if} \quad \delta \leq i \leq n \,, \nn \\
		&  y^{t,p}_{i,j}=0
		\hspace{8.5em}  \text{if} \quad i,j\neq 0\,, \quad t-p \quad  \text{is even } \quad  \nn \\
		&\hspace{11.6em}\quad \text{and} \quad i + j - t - p \geq \delta \,.
	\end{align}
	Finally, one can check from Eq.~\eqref{eq:gamma} that $\gamma^{0,0}_{i,0}=\gamma^{i,i}_{i,i}$.
	Thus, by mapping the variable $y^{t,p}_{ij} \mapsto \gamma^{t,p}_{i,j}y^{t,p}_{ij}$, the constraints in the SDP above simplify to Eq.~\eqref{eq:dual_sol_app_prop} since  $\gamma^{t,p}_{ij}$ factorizes.
	This ends the proof.
	\end{proof}

	\subsection{Dual of the symmetry-reduced SDP}
	
	\begin{proposition}\label{prop:dual_all}
			Given the relaxation of symmetry reduced SDP in Eq.~\eqref{eq:sdpx_relax}, the dual program reads
\begin{align}\label{eq:dual_solall_prop_app}
	\alpha \quad = \quad \max_{Y^{(a,k)}, Q_i, C_i} \quad &   (KD_0-C_0) + \left(\frac{2^n}{K} -1\right)Q_0 - {y^{0,0}_{0,0}} \quad &  \nn \\
	\text{subject to}  \quad & Y^{(a,k)} \succeq0,\nn \\
	& y^{t,p}_{i,j}\ = \frac{1}{\gamma^{t,p}_{i,j}} \sum^{\min(i,j)}_{k=0} \sum^k_{a=\max(i,j)+k-n} \alpha(i,j,t,p,a,k) Y^{(a,k)}_{i-k,j-k}\,, \nn \\
	&D_i = 2^{-n} \sum^{\delta-1}_{j=0} K_j(i,n)  C_{j} \,, \nn \\
	&(KD_i-C_i)+ \left(\frac{2^n}{K}-2\right)Q_i - 2y^{0,0}_{i,0}
	-Q_0  -y^{i,i}_{i,i}=0
	\quad \text{if} \quad 0 <i <\delta \,,\nn \\
	& KD_i + \left(\frac{2^n}{K}-2\right)Q_i - 2y^{0,0}_{i,0}
	-Q_0  -y^{i,i}_{i,i}=0
	\quad\quad\quad\quad \text{if} \quad \delta \leq i \leq n \,, \nn \\
	& \sum_{\substack{ (i',j',t',p')\in \II(4,n) \\ (t-p)=(t'-p') \\ (i',j',i'+j'-t'-p') \, \text{a} \\
			\, \text{permutation of} \,  (i,j,i+j-t-p) \\ }} \!\!\!\!\!\!\!\!\!\!  
	\left(y^{t',p'}_{i',j'}+ Q_{i'+j'-t'-p'}\right) = 0 \quad  \hspace{3.3em} \text{if} \quad (t-p) \quad \text{even}  \nn \\ &  \hspace{6cm}  \quad \text{and} \quad i,j,i+j-t-p \neq 0 \,.
\end{align}
Here the dual variables are real matrices $Y^{(a,k)}$ of size $n+a-2k$ with $ 0\leq a\leq k\leq n+a-k$ and
$Q_i, C_i \in \R$ with $0\leq i\leq n$.
	\end{proposition}
	\begin{proof}
		
		For convenience, we perform the following variations to the SDP~\eqref{eq:sdpx_relax}:
		\begin{enumerate}
			\item Substitute $\sum^{n}_{i=0} \gamma^{0,0}_{i,0}x^{0,0}_{i,0}=\frac{2^n}{K}$ for $\frac{K}{2^n}\sum^n_{i=0} K_0(i;n) \gamma^{0,0}_{i,0} x^{0,0}_{i,0} =x^{0,0}_{0,0}\gamma^{0,0}_{0,0}$ since they are equivalent
			(see Proposition~\ref{prop:kl_cond_sdp} ff.).
			\item Consider $x^{0,0}_{i,0}=x^{i,i}_{i,i}$ independent from the remaining conditions in (iii) from Proposition~\ref{prop:Terwilliger_extra_constraints}.
		\end{enumerate}

		As result the SDP from Eq.~\eqref{eq:sdpx_relax} is now written as
		\begin{align}\label{eq:sdprefall}
			\text{find} 		\quad & x^{t,p}_{i,j}  \nn \\
			\text{subject to} 	\quad
			& x^{0,0}_{0,0}=1, \quad x^{0,0}_{i,0} = x^{i,i}_{i,i} \,,\nn \\ 
			&x^{t,p}_{i,j}=0 \quad\quad\quad \text{if}\quad  t-p \quad  \text{is odd}\,, \nn \\
			& x^{t,p}_{i,j}=x^{t',p'}_{i',j'}\,, \quad \text{if} \quad  t-p=t'-p' \quad \text{is even} \quad \text{and} \quad (i,j,i+j-t-p) \quad \text{a} \nn \\
			&\hspace{2.6cm}\text{permutation of} \,\,(i',j',i'+j'-t'-p') \quad \text{with} \quad i,j,i+j-t-p\neq 0 \,, \nn \\
			&\sum_{\substack{(i,j,t,p)\in \II(4,n) \\ k=i+j-t-p}}\gamma^{t,p}_{i,j}x^{t,p}_{i,j}  =\frac{2^n}{K}  \gamma^{0,0}_{k,0}x^{0,0}_{k,0}  \,, \nn\\
			& K2^{-n}\sum^n_{i=0} K_j(i,n) \gamma^{0,0}_{i,0} x^{0,0}_{i,0}=\gamma^{0,0}_{j,0} x^{0,0}_{j,0} \quad \text{if} \quad 0\leq j<\delta \,, \nn \\
			&\bigoplus_{\substack{a,k \in \N_0\\ 0\leq a\leq k\leq n+a-k}}
			\left(\sum\limits_{\substack{t,p\in \N_0 \\ 0 \leq p \leq t \leq i,j \\  i+j \leq t+n}}\alpha(i,j,t,p,a,k)x_{i,j}^{t,p}\right)_{i,j=k}^{n+a-k}   \succeq 0 \,,
		\end{align} 
		where $\alpha(i,j,t,p,a,k)$ is given by Eq.~\eqref{eq:alpha}. In all the development, we will impose that $\gamma^{0,0}_{0,0}=1$.

		Similarly to the previous section, we now construct the Lagrangian. Associate with each equality constraint a dual variable ${u},{v_i},r^{t,p}_{i,j}, s^{t,t',p,p'}_{i,i',j,j'}, Q_i, C_i \in \R$ and with each semidefinite constraint a dual semidefinite variable $Y^{(a,k)}\succeq0$ and the objective function equals to $0$. Then
		\begin{align}\label{eq:dual_relaxation_proof}
			\mathcal{L} &= \left(1 - x^{0,0}_{0,0}\right) {u}+ \sum^n_{i=1} \left( x^{0,0}_{i,0}-x^{i,i}_{i,i}\right) {v_i} \nn \\
			&\quad + \sum_{\substack{ (i,j,t,p)\in \II(4,n) \\ (t-p) \,\text{odd}}} x^{t,p}_{i,j}r^{{t,p}}_{i,j} 
			\quad + \sum_{\substack{ (i,j,t,p)\in \II(4,n) \\ i,j,i+j-t-p \neq 0\\ (t-p) \,\text{even}}} \sum_{\substack{ (i',j',t',p')\in \II(4,n) \\ (t-p)=(t'-p') \\ (i',j',i'+j'-t'-p') \, \text{a} \\
					\, \text{permutation of} \,  (i,j,i+j-t-p) \\ }} \!\!\!\!\!\!\!\!\!\!\!\!
					(x^{t,p}_{i,j} - x^{t',p'}_{i',j'} )s^{{t,t',p,p'}}_{i,i',j,j'} \nn \\
			& \quad + \sum^n_{k=0} \frac{2^n}{K} \gamma^{0,0}_{k,0}x^{0,0}_{k,0} Q_k
			%\quad
			\quad-\quad
			\sum^n_{k=0}\sum_{\substack{(i,j,t,p)\in \II(4,n) \\ k=i+j-t-p}}\gamma^{t,p}_{i,j} x^{t,p}_{i,j}   Q_{i+j-t-p} \nn\\
			& \quad
			+ \sum^n_{i=0} K \left(2^{-n} \sum^{\delta-1}_{j=0} K_j(i,n)  C_{j} \right)\gamma^{0,0}_{i,0} x^{0,0}_{i,0}
			- \sum^{\delta-1}_{j=0}\gamma^{0,0}_{j,0} x^{0,0}_{j,0} C_j  \nn \\
			&\quad-\sum_{(i,j,t,p)\in \II(4,n)} \left(\sum^{\min(i,j)}_{k=0} \sum^k_{\substack{a=\max(i,j)\\+k-n}} \!\!\!\!\!\! \alpha(i,j,t,p,a,k) Y^{(a,k)}_{i-k,j-k}\right)x^{t,p}_{i,j}\,.
		\end{align}
		As shown in the development from  Eq.~\eqref{eq:dual_semidefinite1} to Eq.~\eqref{eq:dual_semidefinite2}, the last term arises as the the trace of the product of the primal and dual semidefinite variables.
		Before proceeding with the Lagrangian, we simplify the last term in the second line
		of Eq.~\eqref{eq:dual_relaxation_proof} coming from the structure constraint
		by reorganizing its elements as
		\begin{align}\label{eq:last_term_second_line}
			&\sum_{\substack{ (i,j,t,p)\in \II(4,n) \\ i,j,i+j-t-p \neq 0\\ (t-p) \,\text{even}}} \!\!\!\!
			\sum_{\substack{ (i',j',t',p')\in \II(4,n) \\ (t-p)=(t'-p') \\ (i',j',i'+j'-t'-p') \, \text{a} \\
					\, \text{permutation of} \,  (i,j,i+j-t-p) \\ }} \!\!\!\!\!\!\!\!\!\!  
			(x^{t,p}_{i,j} - x^{t',p'}_{i',j'} )s^{{t,t',p,p'}}_{i,i',j,j'}   \nn\\
			&= 
			\sum_{\substack{ (i,j,t,p)\in \II(4,n) \\ i,j,i+j-t-p \neq 0\\ (t-p) \,\text{even}}} x^{t,p}_{i,j}
			\Bigg(\!\!\!\!\!\!\!\!\sum_{\substack{ (i',j',t',p')\in \II(4,n) \\ (t-p)=(t'-p') \\ (i',j',i'+j'-t'-p') \, \text{a} \\
					\, \text{permutation of} \,  (i,j,i+j-t-p) \\ }} \!\!\!\!\!\!\!\!\!\!  (s^{{t,t',p,p'}}_{i,i',j,j'} - s^{{t',t,p',p}}_{i',i,j',j} ) \Bigg)
			\nn \\
			&= \sum_{\substack{ (i,j,t,p)\in \II(4,n) \\ i,j,i+j-t-p \neq 0\\ (t-p) \,\text{even}}} x^{{t,p}}_{i,j} f^{{t,p}}_{i,j} \,,
		\end{align}
		where
		\begin{align}\label{eq:str_const_dual}
			f^{{t,p}}_{i,j} = \sum_{\substack{ (i',j',t',p')\in \II(4,n) \\ (t-p)=(t'-p') \\ (i',j',i'+j'-t'-p') \, \text{a} \\
					\, \text{permutation of} \,  (i,j,i+j-t-p) \\ }} \!\!\!\!\!\!\!\!\!\!  (s^{{t,t',p,p'}}_{i,i',j,j'} - s^{{t',t,p',p}}_{i',i,j',j} )\,.
		\end{align}
		From the definition of $f^{t,p}_{i,j}$ follows a first constraint:
		for $(a,b,c,d)\in \II(4,n)$ with $(c-d)$ even and $a,b,a+b-c-d~\neq~0$,
		\begin{align}\label{eq:dual_str_constr}
			\sum_{\substack{ (i,j,t,p)\in \II(4,n) \\ (c-d)=(t-p) \\ (i,j,i+j-t-p) \, \text{a} \\
					\, \text{permutation of} \,  (a,b,a+b-c-d) \\ }} \!\!\!\!\!\!\!\!\!\!  f^{t,p}_{i,j} = 0 \,.
		\end{align}
		This can be seen by taking Eq.~\eqref{eq:str_const_dual} and developing the left hand side of Eq.~\eqref{eq:dual_str_constr} as
		\begin{align}
			 \sum_{\substack{ (i,j,t,p)\in \II(4,n) \\ (c-d)=(t-p) \\ (i,j,i+j-t-p) \, \text{a} \\
						\, \text{permutation of} \,  (a,b,a+b-c-d) \\ }} \sum_{\substack{ (i',j',t',p')\in \II(4,n) \\ (t-p)=(t'-p') \\ (i',j',i'+j'-t'-p') \, \text{a} \\
					\, \text{permutation of} \,  (i,j,i+j-t-p) \\ }} \!\!\!\!\!\!\!\!\!\!  s^{{t,t',p,p'}}_{i,i',j,j'}  - 			 \sum_{\substack{ (i,j,t,p)\in \II(4,n) \\ (c-d)=(t-p) \\ (i,j,i+j-t-p) \, \text{a} \\
					\, \text{permutation of} \,  (a,b,a+b-c-d) \\ }} \sum_{\substack{ (i',j',t',p')\in \II(4,n) \\ (t-p)=(t'-p') \\ (i',j',i'+j'-t'-p') \, \text{a} \\
					\, \text{permutation of} \,  (i,j,i+j-t-p) \\ }} \!\!\!\!\!\!\!\!\!\! s^{{t',t,p',p}}_{i',i,j',j} 
		\end{align}
		Note that the right and left hand side of the above equation sum over the same elements and thus, they are equal.
		
		Define also $y^{t,p}_{i,j}\ = \sum^{\min(i,j)}_{k=0} \sum^k_{a=\max(i,j)+k-n} \alpha(i,j,t,p,a,k) Y^{(a,k)}_{i-k,j-k}$ and $ D_i = 2^{-n} \sum^{\delta-1}_{j=0} K_j(i,n)  C_{j}$.
		We then factorize the primal variables so that
		\begin{align}\label{eq:lag_factor}
			\mathcal{L} &={u} 
			+ x^{0,0}_{0,0}\left( -{u} + (KD_0-C_0) + \left(\frac{2^n}{K} -1\right)Q_0 - y^{0,0}_{0,0}\right) \nn\\ 
			&\quad +\sum^{\delta-1}_{i=1} x^{0,0}_{i,0}\Big( {v_i}  + \gamma^{0,0}_{i,0}(KD_i-C_i)+ \frac{2^n}{K}\gamma^{0,0}_{i,0}Q_i-2\gamma^{0,0}_{i,0}Q_i  - 2y^{0,0}_{i,0}\Big) \nn \\ 
			&\quad +\sum^{n}_{i=\delta} x^{0,0}_{i,0}\Big( {v_i}  + \gamma^{0,0}_{i,0} KD_i + \frac{2^n}{K}\gamma^{0,0}_{i,0}Q_i-2\gamma^{0,0}_{i,0}Q_i  - 2y^{0,0}_{i,0}\Big) \nn \\
			&\quad
			+\sum^{n}_{i=1} x^{i,i}_{i,i}\Big(- {v_i} -\gamma^{i,i}_{i,i}Q_0 - y^{i,i}_{i,i}\Big) \nn \\
			&\quad +\sum_{\substack{(i,j,t,p)\in \II(4,n) \\ (t-p) \,\text{odd}}} x^{t,p}_{i,j}\Big(r^{t,p}_{i,j} - \gamma^{t,p}_{i,j} Q_{i+j-t-p}-y^{t,p}_{i,j}\Big) \nn \\
			&\quad + \sum_{\substack{ (i,j,t,p)\in \II(4,n) \\ i,j,i+j-t-p \neq 0\\ (t-p) \,\text{even}}} x^{t,p}_{i,j} ( f^{t,p}_{i,j}-\gamma^{t,p}_{i,j} Q_{i+j-t-p} -y^{t,p}_{i,j} )\,.
		\end{align}
		As done in Eq.~\eqref{eq:clar_fact_7}, we have decomposed
		\begin{align}
	\sum_{(i,j,t,p)\in \II(4,n)} x^{t,p}_{i,j} y^{t,p}_{i,j}&= x^{0,0}_{0,0} y^{0,0}_{0,0} +  \sum^{n}_{i=1}2x^{0,0}_{i,0} y^{0,0}_{i,0} + \!\!\!\!\!\!\!\!\sum_{\substack{(i,j,t,p)\in \II(4,n) \\ (t-p) \,\text{odd} }} x^{t,p}_{i,j} r^{{t,p}}_{i,j} + \!\!\!\!\!\!\!\!\sum_{\substack{ (i,j,t,p)\in \II(4,n) \\ i,j,i+j-t-p \neq 0\\ (t-p) \,\text{even}}} x^{t,p}_{i,j}y^{t,p}_{i,j} \,.
		\end{align} 
		By Eq.~\eqref{eq:lag_factor}, the dual objective function is then
				\begin{align}\label{eq:inf_all}
			\inf_{x^{t,p}_{i,j}} \mathcal{L}&={u} \quad \text{if} \quad
			\begin{cases}
				(1) \quad 0= - {u} + (KD_0-C_0) + \left(\frac{2^n}{K} -1\right)Q_0 - {y^{0,0}_{0,0}}\,, \\
				(2) \quad 0={v_i}  + \gamma^{0,0}_{i,0}(KD_i-C_i)+ \frac{2^n}{K}\gamma^{0,0}_{i,0}Q_i-2\gamma^{0,0}_{i,0}Q_i  - 2y^{0,0}_{i,0}
				&\text{if} \quad 0 <i <\delta \,,\\
				(3) \quad0={v_i}  + \gamma^{0,0}_{i,0} KD_i + \frac{2^n}{K}\gamma^{0,0}_{i,0}Q_i-2\gamma^{0,0}_{i,0}Q_i  - 2y^{0,0}_{i,0}
				&\text{if} \quad \delta \leq i \leq n \,,\\
				(4) \quad 0=-{v_i} -\gamma^{i,i}_{i,i}Q_0 - y^{i,i}_{i,i}\quad 			&\text{if} \quad 0 < i \leq n \,,\\
				(5) \quad 0=r^{t,p}_{i,j} - \gamma^{t,p}_{i,j} Q_{i+j-t-p}-y^{t,p}_{i,j}	&	\text{if} \quad (t-p) \,\, \text{odd} \,, \\
				(6) \quad 0=f^{t,p}_{i,j}-\gamma^{t,p}_{i,j} Q_{i+j-t-p} -y^{t,p}_{i,j}				&\text{if} \quad (t-p) \,\, \text{even}
				\,, \\
				&\quad i,j\neq 0 \quad \text{and} \quad \\ &
				\quad\quad i+j-t-p\neq 0 \,.
			\end{cases}
		\end{align}
		Otherwise, it tends to $-\infty$. Similarly to Eq.~\eqref{eq:dual_redun}, condition (5) above,
		\begin{align}\label{eq:cond_slack}
		r^{t,p}_{i,j} - \gamma^{t,p}_{i,j} Q_{i+j-t-p}-y^{t,p}_{i,j}=0	\quad	\text{if} \quad (t-p) \,\, \text{odd} 
		\end{align}
		does not constrain the objective function. Again, $r^{t,p}_{i,j}$ acts as a slack variable
		which can always take values such that Eq.~\eqref{eq:cond_slack} is satisfied.
		
		All dual variables can be expressed in terms of $Y^{(a,k)},Q_i,C_i$ using the remaining of constraints in Eq.~\eqref{eq:inf_all}. First, use condition (1) to derive the objective function of the dual which reads ${u} = (KD_0-C_0) + \left(\frac{2^n}{K} -1\right)Q_0 - {y^{0,0}_{0,0}}$. Then apply condition (4) to (2) and (3) to obtain respectively,
	\begin{align}
	 \quad 0 &= \gamma^{0,0}_{i,0}(KD_i-C_i)+ \left(\frac{2^n}{K}-2\right)\gamma^{0,0}_{i,0}Q_i - 2y^{0,0}_{i,0}
	 -\gamma^{i,i}_{i,i}Q_0  -y^{i,i}_{i,i}
	\quad \text{if} \quad 0 <i <\delta \,,\\
	\quad0 & = \gamma^{0,0}_{i,0} KD_i + \left(\frac{2^n}{K}-2\right)\gamma^{0,0}_{i,0}Q_i - 2y^{0,0}_{i,0}
	-\gamma^{i,i}_{i,i}Q_0  -y^{i,i}_{i,i}
	\quad\quad\quad\quad \text{if} \quad \delta \leq i \leq n \,.
	\end{align}
	Finally, we use the property of $f^{t,p}_{i,j}$ described in Eq.~\eqref{eq:dual_str_constr}  together with condition (6) and thus,
	\begin{align}
	\sum_{\substack{ (i',j',t',p')\in \II(4,n) \\ (t-p)=(t'-p') \\ (i',j',i'+j'-t'-p') \, \text{a} \\
			\, \text{permutation of} \,  (i,j,i+j-t-p) \\ }} \!\!\!\!\!\!\!\!\!\!  
	\left(y^{t',p'}_{i',j'}+\gamma^{t',p'}_{i',j'} Q_{i'+j'-t'-p'}\right) = 0 \quad  \hspace{0.3em} \text{if} \quad (t-p) \quad \text{even}  \nn \\ \quad \text{and} \quad i,j,i+j-t-p \neq 0. 
	\end{align}
	The dual program then reads as
	\begin{align}
		\alpha \quad = \quad \max_{Y^{(a,k)}, Q_i, C_i} \quad &   (KD_0-C_0) + \left(\frac{2^n}{K} -1\right)Q_0 - {y^{0,0}_{0,0}} \quad &  \nn \\
		\text{subject to}  \quad & Y^{(a,k)} \succeq0,\nn \\
		& y^{t,p}_{i,j}\ = \sum^{\min(i,j)}_{k=0} \sum^k_{a=\max(i,j)+k-n} \alpha(i,j,t,p,a,k) Y^{(a,k)}_{i-k,j-k}\,, \nn \\
		&D_i = 2^{-n} \sum^{\delta-1}_{j=0} K_j(i,n)  C_{j} \,, \nn \\
		&\gamma^{0,0}_{i,0}(KD_i-C_i)+ \left(\frac{2^n}{K}-2\right)\gamma^{0,0}_{i,0}Q_i - 2y^{0,0}_{i,0}
		-\gamma^{i,i}_{i,i}Q_0  -y^{i,i}_{i,i}=0
		\quad \text{if} \quad 0 <i <\delta \,,\nn \\
		& \gamma^{0,0}_{i,0} KD_i + \left(\frac{2^n}{K}-2\right)\gamma^{0,0}_{i,0}Q_i - 2y^{0,0}_{i,0}
		-\gamma^{i,i}_{i,i}Q_0  -y^{i,i}_{i,i}=0
		\quad\quad\quad\quad \text{if} \quad \delta \leq i \leq n \,, \nn \\
		& \sum_{\substack{ (i',j',t',p')\in \II(4,n) \\ (t-p)=(t'-p') \\ (i',j',i'+j'-t'-p') \, \text{a} \\
				\, \text{permutation of} \,  (i,j,i+j-t-p) \\ }} \!\!\!\!\!\!\!\!\!\!  
		\left(y^{t',p'}_{i',j'}+\gamma^{t',p'}_{i',j'} Q_{i'+j'-t'-p'}\right) = 0 \quad  \hspace{5.6em} \text{if} \quad (t-p)  \nn \\ &  \hspace{6cm}  \text{even} \quad \text{and} \quad i,j,i+j-t-p \neq 0.
	\end{align}
	
	Finally, one can check from Eq.~\eqref{eq:gamma} that $\gamma^{t,p}_{ij}=\gamma^{t',p'}_{i'j'}$ if $t-p=t'-p'$ is even and $(i,j,i+j-t-p)$ is a permutation of $(i',j',i'+j'-t'-p')$. In particular, $\gamma^{0,0}_{i,0}=\gamma^{i,i}_{i,i}$.
	Thus, by mapping the variable $y^{t,p}_{ij} \mapsto \gamma^{t,p}_{i,j}y^{t,p}_{ij}$, the constraints in the SDP simplify to Eq.~\eqref{eq:dual_solall_prop_app} since $\gamma^{t,p}_{ij}$ factorizes.
	This ends the proof.
\end{proof}
	
	\section{Infeasibility certificate $(\!(7,1,4)\!)_2$ code}\label{app:inf}

	The following matrices $Y^{(a,k)}$ are a solution of the dual program~\eqref{eq:dual_sol} for $n=7$ and $\delta=4$ with $\alpha =0.58>0$ and a maximum constraint violation of
	$\mathcal{O}(10^{-16})$. 
	Note that the first constraint from Eq.~\eqref{eq:dual_sol_app_prop} allows to determine the value from $w$ using $Y^{(a,k)}$ since $w=-y^{i,i}_{i,i}-y^{0,0}_{i,0}$ if $\delta \leq i\leq n$. In particular, one can choose $w=-y^{n,n}_{n,n}-y^{0,0}_{n,0}$.
	The matrices thus provide an infeasibility certificate and show that a $(\!(7,1,4)\!)_2$ quantum code does not exist.

\begin{align}
Y^{(0,0)} &=
\begin{bmatrix}
	124 & 0 & 0 & 0 & -9 & -15 & -26 & -43\\
	0 & 115 & 0 & 0 & 0 & 0 & 0 & 0\\
	0 & 0 & 22 & 0 & 0 & 0 & 0 & 0\\
	0 & 0 & 0 & 3 & 0 & 0 & 0 & 0\\
	-9 & 0 & 0 & 0 & 1 & 1 & 2 & 1\\
	-15 & 0 & 0 & 0 & 1 & 3 & 2 & 2\\
	-26 & 0 & 0 & 0 & 2 & 2 & 7 & 11\\
	-43 & 0 & 0 & 0 & 1 & 2 & 11 & 43\\
\end{bmatrix}
\nn\\
Y^{(0,1)} &=
\begin{bmatrix}
	120 & 0 & 0 & 0 & 0 & 0\\
	0 & 238 & 0 & 0 & 0 & 0\\
	0 & 0 & 66 & 0 & 0 & 0\\
	0 & 0 & 0 & 9 & 10 & -10\\
	0 & 0 & 0 & 10 & 15 & 3\\
	0 & 0 & 0 & -10 & 3 & 92\\
\end{bmatrix}
\nn\\
Y^{(0,2)} &=
\begin{bmatrix}
	140 & 0 & 0 & 0\\
	0 & 274 & 0 & 0\\
	0 & 0 & 91 & 24\\
	0 & 0 & 24 & 330\\
\end{bmatrix}
\nn\\
Y^{(0,3)} &=
\begin{bmatrix}
	84 & 0\\
	0 & 12\\
\end{bmatrix}
\nn\\
Y^{(1,1)} &=
\begin{bmatrix}
	120 & 0 & 0 & 0 & 0 & 0 & 0\\
	0 & 120 & 0 & 0 & 0 & 0 & 0\\
	0 & 0 & 63 & 0 & 0 & 0 & 0\\
	0 & 0 & 0 & 11 & 4 & -5 & 5\\
	0 & 0 & 0 & 4 & 7 & 4 & -8\\
	0 & 0 & 0 & -5 & 4 & 29 & 56\\
	0 & 0 & 0 & 5 & -8 & 56 & 269\\
\end{bmatrix}
\nn\\
Y^{(1,2)} &=
\begin{bmatrix}
	117 & 0 & 0 & 0 & 0\\
	0 & 291 & 0 & 0 & 0\\
	0 & 0 & 130 & 56 & 285\\
	0 & 0 & 56 & 419 & 325\\
	0 & 0 & 285 & 325 & 831\\
\end{bmatrix}
\nn\\
Y^{(1,3)} &=
\begin{bmatrix}
	126 & 0 & 0\\
	0 & \frac{405}{4} & 207\\
	0 & 207 & 470\\
\end{bmatrix}
\nn\\
Y^{(1,4)} &=
\begin{bmatrix}
	\frac{93}{2}\\
\end{bmatrix}
\nn\\
Y^{(2,2)} &=
\begin{bmatrix}
	114 & 0 & 0 & 0 & 0 & 0\\
	0 & 109 & 0 & 0 & 0 & 0\\
	0 & 0 & 53 & 26 & 85 & 20\\
	0 & 0 & 26 & 108 & -69 & 123\\
	0 & 0 & 85 & -69 & 352 & 147\\
	0 & 0 & 20 & 123 & 147 & 937\\
\end{bmatrix}
\nn\\
Y^{(2,3)} &=
\begin{bmatrix}
	140 & 0 & 0 & 0\\
	0 & 127 & 65 & \frac{125}{3}\\
	0 & 65 & 57 & 46\\
	0 & \frac{125}{3} & 46 & 81\\
\end{bmatrix}
\nn\\
Y^{(2,4)} &=
\begin{bmatrix}
	18 & -29\\
	-29 & 56\\
\end{bmatrix}
\nn\\
Y^{(3,3)} &=
\begin{bmatrix}
	107 & 0 & 0 & 0 & 0\\
	0 & 24 & -3 & \frac{125}{4} & 70\\
	0 & -3 & 71 & 98 & 13\\
	0 & \frac{125}{4} & 98 & 224 & 149\\
	0 & 70 & 13 & 149 & 246\\
\end{bmatrix}
\nn\\
Y^{(3,4)} &=
\begin{bmatrix}
	99 & \frac{45}{2} & 70\\
	\frac{45}{2} & \frac{263}{4} & -60\\
	70 & -60 & 158\\
\end{bmatrix}
\nn\\
Y^{(3,5)} &=
\begin{bmatrix}
	\frac{39}{2}\\
\end{bmatrix}
\nn\\
Y^{(4,4)} &=
\begin{bmatrix}
	106 & 10 & \frac{-290}{3} & 0\\
	10 & \frac{16187991}{644740} & -44 & 38\\
	\frac{-290}{3} & -44 & \frac{6623}{36} & -82\\
	0 & 38 & -82 & 88\\
\end{bmatrix}
\nn\\
Y^{(4,5)} &=
\begin{bmatrix}
	\frac{25998073}{322370} & 44\\
	44 & \frac{379}{12}\\
\end{bmatrix}
\nn\\
Y^{(5,5)} &=
\begin{bmatrix}
	\frac{21984965}{392883} & \frac{-78}{5} & \frac{-186}{5}\\
	\frac{-78}{5} & \frac{1061}{80} & \frac{93}{5}\\
	\frac{-186}{5} & \frac{93}{5} & \frac{34838427}{857647}\\
\end{bmatrix}
\nn\\
Y^{(5,6)} &=
\begin{bmatrix}
	\frac{907}{8}\\
\end{bmatrix}
\nn\\
Y^{(6,6)} &=
\begin{bmatrix}
	\frac{289}{5} & \frac{-548}{5}\\
	\frac{-548}{5} & \frac{164534807}{670060}\\
\end{bmatrix}
\nn\\
Y^{(7,7)} &=
\begin{bmatrix}
	\frac{3813979}{522046}\\
\end{bmatrix}
\end{align}
\end{document}